\documentclass[11pt]{article}
\usepackage{fullpage}
\usepackage{subfig}
\usepackage[usenames,dvipsnames]{xcolor}
\usepackage[linktocpage=true,colorlinks,citecolor=blue,linkcolor=BrickRed]{hyperref}
	\usepackage{latexsym,graphicx,epsfig,color}
\usepackage{amsfonts,amssymb,amsmath,amsthm,amstext}
\usepackage{libertine}
\usepackage[libertine]{newtxmath}
\usepackage{enumitem}
\setitemize{itemsep=2pt,topsep=0pt,parsep=0pt}
\usepackage{url,setspace}
\usepackage{multirow}
\usepackage{rotating}
\usepackage{makeidx}
\usepackage{tikz}
\usepackage{accents}
\usepackage{xspace}
\usepackage{algorithm,algpseudocode}
\usepackage{bm}
\usepackage{changepage} 	
\usepackage{thmtools,thm-restate}

\newcommand{\IGNORE}[1]{}
\allowdisplaybreaks

\usetikzlibrary{decorations.markings}
\usetikzlibrary{arrows}
\tikzstyle{block}=[draw opacity=0.7,line width=1.4cm]
\tikzstyle{graphnode}=[circle, draw, fill=black!20, inner sep=0pt, minimum width=6pt]
\tikzstyle{point}=[circle, draw, fill=black!30, inner sep=0pt, minimum width=1pt]
\tikzstyle{input}=[rectangle, draw, fill=black!75,inner sep=3pt, inner ysep=3pt, minimum width=4pt]
\tikzstyle{unmatched}=[graphnode,fill=black!0]
\tikzstyle{shaded}=[graphnode,fill=black!20]
\tikzstyle{matched}=[graphnode,fill=black!100]  	
\tikzstyle{matching} = [ultra thick]
\tikzset{
    >=stealth',
    pil/.style={
           ->,
           thick,
           shorten <=2pt,
           shorten >=2pt,}
}
\tikzset{->-/.style={decoration={
  markings,
  mark=at position .5 with {\arrow{>}}},postaction={decorate}}}

\makeindex

\newtheorem{theorem}{Theorem}[section]
\newtheorem{claim}[theorem]{Claim}
\newtheorem{proposition}[theorem]{Proposition}
\newtheorem{lemma}[theorem]{Lemma}
\newtheorem{corollary}[theorem]{Corollary}

\newtheorem{observation}[theorem]{Observation}

\theoremstyle{definition}
\newtheorem{example}[theorem]{Example}

\newtheorem{defn}[theorem]{Definition}

\usepackage[compact]{titlesec}

\usepackage{thmtools,thm-restate} 

\newcommand{\poly}{\operatorname{poly}}

\newcounter{note}[section]

\newcommand{\REV}{\ifmmode\mathsf{Rev}\else\textsf{Rev}\fi}
\newcommand{\BREV}{\ifmmode\mathsf{BRev}\else\textsf{BRev}\fi}
\newcommand{\SREV}{\ifmmode\mathsf{SRev^*}\else\textsf{SRev^*}\fi}
\newcommand{\VAL}{\ifmmode\mathsf{Val}\else\textsf{Val}\fi}
\newcommand{\MBGN}{\ifmmode\mathsf{M}_{BGN}\else\textsf{M}_{BGN}\fi}
\newcommand{\MUS}{\ifmmode\mathsf{M}_{KMSSW}\else\textsf{M}_{KMSSW}\fi}

\newcommand{\mc}{\textsc{MC}}
\newcommand{\wsmc}{\textsc{WSMC}}
\newcommand{\ssmc}{\textsc{SSMC}}

\newcommand{\notshow}[1]{}
\newcommand{\revmax}{\textsc{RevMax}}
\newcommand{\modrevmax}{\textsc{ModRevMax}}
\newcommand{\VAR}{\ifmmode\mathsf{Var}\else\textsf{Var}\fi}
\newcommand{\srev}{\ifmmode\mathsf{SRev}\else\textsf{SREV}\fi}
\newtheorem{informal}{Main Result}
\newtheorem*{informal2}{Informal Theorem}
\newtheorem{openq}{Open Question}

\title{Approximation Schemes for a Unit-Demand Buyer with Independent Items via Symmetries}

	\author{Pravesh K. Kothari\thanks{
	(praveshk@cs.cmu.edu)
	Department of Computer Science,
        Carnegie Mellon University.
        }
        \and Divyarthi Mohan\thanks{
	(dm23@cs.princeton.edu)
	Department of Computer Science,
        Princeton University.
        }
        \and Ariel Schvartzman\thanks{
        (acohenca@cs.princeton.edu)
        Department of Computer Science,
        Princeton University.
        }
	\and Sahil Singla\thanks{
        (singla@cs.princeton.edu)
        Princeton University and Institute for Advanced Study.
        }
	\and S. Matthew Weinberg\thanks{
        (smweinberg@princeton.edu)
        Department of Computer Science,
        Princeton University.
        }
}

\date{ \today}

\begin{document}
\maketitle

\setlength{\abovedisplayskip}{2pt}
\setlength{\belowdisplayskip}{2pt}

\pagenumbering{roman}

\begin{abstract}{
We consider a revenue-maximizing seller with $n$ items facing a single buyer. We introduce the notion of \emph{symmetric menu complexity} of a mechanism, which counts the number of distinct options the buyer may purchase, up to permutations of the items. Our main result is that a mechanism of quasi-polynomial symmetric menu complexity suffices to guarantee a $(1-\varepsilon)$-approximation when the buyer is unit-demand over independent items, even when the value distribution is unbounded, and that this mechanism can be found in quasi-polynomial time.

Our key technical result is a polynomial time, (symmetric) menu-complexity-preserving black-box reduction from achieving a $(1-\varepsilon)$-approximation for unbounded valuations that are subadditive over independent items to achieving a $(1-O(\varepsilon))$-approximation when the values are bounded (and still subadditive over independent items). We further apply this reduction to deduce approximation schemes for a suite of valuation classes beyond our main result.

Finally, we show that selling separately (which has exponential menu complexity) can be approximated up to a $(1-\varepsilon)$ factor with a menu of \emph{efficient-linear} ($f(\varepsilon) \cdot n$) symmetric menu complexity.
}\end{abstract}

\thispagestyle{empty}


\setcounter{tocdepth}{3}


\clearpage

\pagenumbering{arabic}
\setcounter{page}{1}


\section{Introduction}\label{sec:intro}

Multi-item mechanism design has been at the forefront of Mathematical Economics since Myerson's seminal work resolved the single-item case~\cite{Myerson81}. Once it became clear that optimal multi-item mechanisms were prohibitively complex, even with just a single buyer (e.g.,~\cite{RochetC98,Thanassoulis04,ManelliV07}), the problem also entered the Theory of Computation through the lens of approximation. As a result, there is now a long line of work developing auctions which are simple, computationally-efficient, and \emph{approximately} optimal~\cite{ChawlaHK07, ChawlaHMS10, ChawlaMS15, KleinbergW12, HartN17, LiY13, BabaioffILW14, Yao15, RubinsteinW15, ChawlaM16, CaiDW16, CaiZ17}. 

These works take a binary view on simplicity, and aim to discover the best approximation guarantees achievable by simple mechanisms. Only recently have works begun to explore the tradeoff between simplicity and optimality, aiming instead to discover how complex a mechanism must be, as a function of $\varepsilon$, in order to guarantee a $(1-\varepsilon)$-approximation to the optimum. This question is studied formally through the lens of \emph{computational complexity} (how much computation is required to find a mechanism guaranteeing a $(1-\varepsilon)$-approximation on a given instance?) and \emph{menu complexity} (how many distinct outcomes need a mechanism induce in order to guarantee a $(1-\varepsilon)$-approximation?).\footnote{See Section~\ref{sec:prelim} for a formal definition of menu complexity, and formal statement of the computational problem.} Prior to our work, neither subexponential upper bounds nor superpolynomial lower bounds were known in any multi-dimensional setting for either measure. Our main results provide the first subexponential upper bound through both lenses.

\subsection{Main Result Part 1: Quasi-Polynomial Computational Complexity}

Our main results concern a single \emph{unit-demand buyer} with independently drawn values for $n$ items, the same setting considered in seminal work of Chawla, Hartline and Kleinberg which introduced this  domain to TCS~\cite{ChawlaHK07}. Specifically, there is a single seller with $n$ heterogeneous items facing a single buyer with value $v_i$ for each item $i$ and value $\max_{i \in S} \{v_i\}$ for any set $S$. The seller has independent Bayesian priors $D_i$ over each $v_i$ (so we say that the buyer is drawn from $D:= \prod_i D_i$). The seller presents the buyer with a menu of (randomized allocation, price) pairs $(S,p)$, and the buyer purchases whichever option maximizes her expected utility ($v(S) - p$).\footnote{Throughout this paper, we abuse notation and write $v(S):=\mathbb{E}[v(S)]$ when $S$ is a set-valued random variable and $v$ is fixed.} The seller's goal is to find, over all menus, the one which optimizes her expected revenue.

Our first main result considers the computational complexity of this problem, which is known to be computationally hard to solve exactly (unless $P^{NP}=P^{\#P}$), even when each $D_i$ has support three~\cite{ChenDOPSY15}. On the other hand, works of~\cite{ChawlaHK07,ChawlaHMS10, ChawlaMS15} establish that a $1/4$-approximation can be found in polynomial time. It was previously unknown whether a $(1-\varepsilon)$-approximation (or even a $(1/4+\varepsilon)$-approximation) could be achieved in subexponential time. Part one of our main result provides the first subexponential-time approximation scheme.

\begin{informal}[Informal, see Theorem~\ref{thm:unit}] For all $\varepsilon > 0$, a $(1-\varepsilon)$-approximation to the optimal revenue for a unit-demand buyer with independent values for $n$ items can be found in time quasi-polynomial in $n$. 
\end{informal}

\subsection{Main Result Part 2: Quasi-Polynomial Symmetric Menu Complexity}
Part two of our main result considers the same problem through the lens of menu complexity. Menu complexity, first defined in~\cite{HartN13}, is widely regarded as an insightful yet imperfect measure. Imagine for example the menu which allows the buyer to purchase any desired set $S$ for a price of $|S|$ (``selling separately''). This is ubiquitously accepted as a fairly simple menu (perhaps it has ``intrinsic complexity'' $n$, since there are $n$ non-trivial ``kinds'' of possible outcomes), yet technically it has menu complexity $2^n$ (because there are $2^n$ different sets the buyer can purchase). 

Prior work addresses this concern in two ways. The first simply proposes alternative definitions, such as additive menu complexity~\cite{HartN13}.\footnote{A formal definition is not relevant to this discussion, but essentially the definition is designed to address the specific concern raised: the menu contains a list of (randomized allocation, price) pairs, and the buyer may (adaptively or non-adaptively) select any subset of options to purchase. So ``selling separately'' has additive menu complexity $n$.} This particular definition, however, was later shown to be ill-defined (in the sense that there exist optimal menus for which the additive menu complexity is undefined)~\cite{BabaioffNR18}, and there are no prior alternative proposals. The second approach is to argue that while selling separately may technically have menu complexity $2^n$, it is always well-approximated by a menu of polynomial size~\cite{BabaioffGN17}. The implication is that while the definition is still imperfect, it is at least impossible for a distribution for which selling separately is optimal to witness a super-polynomial lower bound on the menu complexity required for a $(1-\varepsilon)$-approximation.

We propose the first alternative which is well-defined for all menus, the \emph{symmetric menu complexity}. Informally, a menu respects a permutation group $\Sigma$ if whenever $(S,p)$ is an option on the menu, $(\sigma(S),p)$ is an option as well for all $\sigma \in \Sigma$. A menu has symmetric menu complexity $C$ if there exists a $\Sigma$ such that the menu is symmetric with respect to $\Sigma$ and contains at most $C$ distinct equivalence classes under $\Sigma$.\footnote{For ease of exposition, this definition is slightly imprecise, see Section~\ref{sec:menudefs} for a formal definition.} 

Observe that in our motivating example where a buyer could pick $S$ at price $|S|$ has symmetric menu complexity $n$ (the menu is invariant under all permutations), and that importantly the symmetric menu complexity is well-defined for all menus (take $\Sigma$ to contain only the identity permutation). Of course, the definition is still imperfect, as selling separately at $n$ distinct prices still has  symmetric menu complexity $2^n$, but the improvement over standard menu complexity is significant (discussed in Section~\ref{sec:future}). Part two of our main result establishes that quasi-polynomial symmetric menu complexity suffices for a $(1-\varepsilon)$-approximation. Importantly, note that both parts of our main result hold even for {unbounded} distributions.

\begin{informal}[Informal, see Theorem~\ref{thm:unit}] For all $\varepsilon > 0$, a $(1-\varepsilon)$-approximation to the optimal revenue for a unit-demand buyer with independent values for $n$ items exists with symmetric menu complexity quasi-polynomial in $n$. 
\end{informal}

\subsection{Main Result 3: A Reduction from Unbounded to Almost-Bounded}
Our proof of the above main results is cleanly broken down into two steps, the first of which we now overview. We provide a black-box reduction from proving computational/menu/symmetric menu complexity bounds for {unbounded} distributions to proving the same bounds for \emph{almost-bounded} distributions. Roughly, a distribution $D$ is almost-bounded if for each $i$, distribution $D_i$ is supported on $[0,1]\cup \{W\}$ (think of $W$ as some large number $\gg n$). That is, $D_i$ has at most one value in its support exceeding $1$.

This step in our proof applies quite generally, in fact to any distribution which is \emph{subadditive over independent items} (see Section~\ref{sec:prelim} for definition). This constitutes a key result in its own right due to significant gaps in tractability between unbounded and almost-bounded instances. For example, it was only recently shown that \emph{some} $f(n,\varepsilon)<\infty$ menu complexity suffices for a $(1-\varepsilon)$-approximation on all unbounded distributions which are additive\footnote{A valuation function is additive if $v(S)=\sum_{i \in S} v(\{i\})$.} over $n$ independent items~\cite{BabaioffGN17} (and the proof is quite involved), whereas the analogous result follows for almost-bounded distributions by a folklore discretization argument.

\begin{informal}[Informal, see Theorem~\ref{thm:reduction}]\label{inf:reduction} There is a poly-time reduction from a multiplicative $(1-\varepsilon)$-approximation for unbounded distributions which are subadditive over independent items to an additive $O(\varepsilon^5)$-approximation for almost-bounded distributions which are subadditive over independent items. If the $O(\varepsilon^5)$-approximation produced on the almost-bounded instance has (symmetric) menu complexity $C$, the $(1-\varepsilon)$-approximation for the unbounded instance has (symmetric) menu complexity $\leq nC+n$. 
\end{informal}

Readers familiar with~\cite{BabaioffGN17} may notice a relationship to their main result, and a detailed comparison is warranted. The main result of~\cite{BabaioffGN17} asserts that a $(1-\varepsilon)$-approximation for an additive buyer over independent items can be achieved with bounded ($(\ln(n)/\varepsilon)^{O(n)}$) menu complexity, which can now alternatively be deduced from Theorem~\ref{thm:reduction} plus the aforementioned folklore discretization argument (e.g., Corollary~\ref{cor:expectedNudge}, also called a ``nudge-and-round''). In comparison to~\cite{BabaioffGN17}, the main qualitative improvement in Theorem~\ref{thm:reduction} is that we provide a true reduction from unbounded to almost-bounded distributions.\footnote{In contrast,~\cite{BabaioffGN17} wraps up their additive approximation on a bounded distribution via a specific nudge-and-round tailored to the rest of their proof, and are explicit that care is required in this step.} The main quantitative improvements are an extension to subadditive over independent items (versus additive) and that the approximation required on the almost-bounded distribution is independent of $n$ (versus $O(\varepsilon^3/n^3)$). It is worth noting that this quantitative improvement is necessary for our  previous quasi-polynomial results (see discussion following Theorem~\ref{thm:reduction}), so the removal of dependence on $n$ is significant. It is also worth noting that our proof of Theorem~\ref{thm:reduction} indeed makes use of several ideas developed in~\cite{BabaioffGN17}, and we identify the connections where appropriate. 

\subsection{Approximating Almost-Bounded Distributions via Symmetries}
Theorem~\ref{thm:reduction} takes care of reducing unbounded distributions to almost-bounded ones, but we still need to figure out how to get an additive $O(\varepsilon^5)$-approximation on almost-bounded distributions that are unit-demand over independent items. The folklore discretization argument (roughly: round all values down to the nearest multiple of $O(\varepsilon^{10})$) establishes only that an exponential $1/\varepsilon^{O(n)}$ computational/menu complexity suffices. Perhaps shockingly, no better bounds were previously known, so our remaining task is to improve this.

The unique special case where progress was previously made is if $D$ is heavily symmetric (that is, $D$ is i.i.d., or there are only $o(n)$ distinct marginals of $D$)~\cite{DaskalakisW12}. In this case,~\cite{DaskalakisW12} establish that an additive $\varepsilon$-approximation with symmetric menu complexity $n^{O(s/\varepsilon^2)}$ can be found in time $n^{O(s/\varepsilon^2)}$ for any distribution $D$ that is unit-demand over independent items with at most $s$ distinct marginals. Of course, our given $D$ may have $n$ distinct marginals, rendering a direct application of their theorem useless. So our key argument here (captured mostly by Lemma~\ref{lem:carefulcoupling}) is to show that every $D$ which is almost-bounded and unit-demand over independent items is ``close'' in a precise metric (Definition~\ref{def:advancedCoupling}) to some $D'$ which is almost-bounded and unit-demand over independent items with at most $\ln(n)^{1/\varepsilon^{O(1)}}$  distinct marginals.

\subsection{Extensions}
Beyond our main results, Theorem~\ref{thm:reduction} also allows us to conclude the following corollaries:
\begin{itemize}
\item For all $\varepsilon>0$ and $n \in \mathbb{N}$, there exists a finite $f(n,\varepsilon)$ s.t. for all subadditive $D$ over $n$ independent items, a $(1-\varepsilon)$-approximation  can be found in $f(n,\varepsilon)$ time which has menu complexity $f(n,\varepsilon)$ (Theorem~\ref{thm:subadd}).
\item For all $\varepsilon>0$ and all unit-demand $D$ over \emph{i.i.d.} items, a $(1-\varepsilon)$-approximation can be found in polynomial time which has polynomial symmetric menu complexity (Theorem~\ref{thm:iiddemand}).
\item For all $\varepsilon>0$ and all additive $D$ over \emph{i.i.d.} items, a $(1-\varepsilon)$-approximation can be found in quasi-polynomial time which has quasi-polynomial symmetric menu complexity (Theorem~\ref{thm:iiddemand}).
\item For all $\varepsilon > 0$ and all additive $D$ over independent items \emph{where for all $i$, $|\text{support}(D_i)|= O(1)$}, a $(1-\varepsilon)$-approximation can be found in quasi-polynomial time which has quasi-polynomial symmetric menu complexity (Theorem~\ref{thm:constant}). Note that the supports of each $D_i$ may be distinct.
\end{itemize}

The proofs of the first three bullets follow by first applying our reduction (Theorem~\ref{thm:reduction}), and then applying standard (albeit somewhat subtle) nudge-and-round arguments (e.g., Corollary~\ref{cor:expectedNudge}) on the resulting almost-bounded distribution. Note that exploiting symmetries for i.i.d. distributions is considerably simpler than for non-i.i.d. distributions as the input is already symmetric (so~\cite{DaskalakisW12} can be applied almost immediately). The final bullet considers a non-i.i.d. setting, but again establishes that all distributions whose marginals have constant support are close (by Definition~\ref{def:advancedCoupling}) to symmetric ones. Note that exact solutions for this setting are computationally intractable~\cite{DaskalakisDT14, ChenDOPSY15}, and that no subexponential-time approximation schemes were previously known (and posed as an open problem in~\cite{DaskalakisDT14}, even for marginals with support two).

Additionally, like~\cite{BabaioffGN17}, we also consider the (symmetric) menu complexity necessary to approximate the revenue achieved by selling separately. For the standard menu complexity,~\cite{BabaioffGN17} establishes that there is always a menu of size $n^{1/\varepsilon^{O(1)}}$ which guarantees a $(1-\varepsilon)$-fraction of the best revenue achievable by selling separately. We establish an even stronger claim for symmetric menu complexity: there is always a menu of \emph{efficient-linear} size which guarantees a $(1-\varepsilon)$-fraction of selling separately.

\begin{informal2}[See Theorem~\ref{thm:apxsrev}] For all menus $M$ which sell separately, and all $D$ which are additive over independent items, there exists a menu with symmetric menu complexity $f(\varepsilon)\cdot n$ which achieves a $(1-\varepsilon)$-approximation to the revenue of $M$ when buyers are drawn from $D$ (here $f(\varepsilon) = 2^{O(1/\varepsilon^3)}$).
\end{informal2}

Again, recall that symmetric menu complexity is still an imperfect definition which assigns complexity $2^n$ to a menu which sells separately at $n$ distinct prices (whereas the ``intrinsic complexity'' of such a menu is $n$). But Theorem~\ref{thm:apxsrev} asserts that for all $\varepsilon$, there is a menu of linear symmetric menu complexity which achieves the same (up to $(1-\varepsilon)$) revenue guarantees. 

Finally, in Section~\ref{sec:approachLimits} we analyze a barrier example for extending our quasi-polynomial bounds for a unit-demand buyer to an additive buyer. Essentially, the distribution is almost-bounded, but (provably) no previous approaches, nor our approach for almost-bounded unit-demand distributions can guarantee better than an $8/9$-approximation. We believe that resolving this example will be a fruitful direction for future work to circumvent current barriers. 

\subsection{Related Work}\label{sec:related}
The most related work to ours is~\cite{BabaioffGN17}, whose main result establishes that finite menu complexity suffices to guarantee a $(1-\varepsilon)$-approximation for an additive buyer over independent items. As discussed above, our black-box reduction provides both qualitative and quantitative improvements on their work, and makes use of tools they develop. There are numerous other works which study the menu complexity of optimal and approximately optimal auctions, but there is not much technical overlap~\cite{BriestCKW15,HartN13,FiatGKK16, SaxenaSW18, Gonczarowski18}. 

(Quasi-)Polynomial Time Approximation Schemes for a single buyer have been considered in prior works from a few different perspectives. For example,~\cite{CaiD11} develops a PTAS for the optimal \emph{deterministic} item pricing for a unit-demand buyer over independent MHR items, and a QPTAS for a unit-demand buyer over independent regular items.\footnote{That is, each $f_i(v)/(1-F_i(v))$ is monotone non-decreasing (MHR) or $v-(1-F_i(v))/f_i(v)$ monotone non-decreasing (regular).}~\cite{Rubinstein16} develops a PTAS for the optimal ``partition mechanism'' for an additive buyer over independent items.\footnote{A partition mechanism partitions the items into disjoint bundles and allows the buyer to purchase any subset of bundles.} The simplest comparison to these works is that we are searching for a good approximation to the optimal (possibly randomized) mechanism, versus a restricted class of mechanisms.~\cite{DaskalakisW12} develop a PTAS for a bounded unit-demand buyer over i.i.d. items by exploiting symmetries. As noted previously, our work provides the first approximation schemes towards the true optimum in unrestricted settings (and also the first application of~\cite{DaskalakisW12} in asymmetric settings). 

A series of works also considers the multi-bidder, multi-item case. Works such as~\cite{DaskalakisW12, CaiH13} consider special cases (such as i.i.d., MHR, etc.), and are able to exploit symmetries or concentration to prove that simple auctions can approach optimal guarantees. In the general case,~\cite{CaiDW12a, CaiDW12b,CaiDW13b} develop fully-polynomial randomized approximation schemes. These works achieve polynomial dependence on the number of bidders, and the size of a single bidder's support (so with independent items, this would be exponential in $n$), and bear no technical similarity. Indeed, one of the open questions left by these works is whether it is possible to improve the dependence on $n$ when the items are independent, and our work resolves this affirmatively in the case of a single unit-demand buyer. 

Finally, note that the interesting questions indeed surround a $(1-\varepsilon)$-approximation, and not exact solutions. For example,~\cite{ManelliV07, DaskalakisDT17} establish that optimal mechanisms may have \emph{uncountable} menu complexity. Moreover, even in the case where the marginal of each item has constant support (two, for additive~\cite{DaskalakisDT14}, three for unit-demand~\cite{ChenDOPSY15}), exact solutions are computationally intractable and subexponential-time approximation schemes were unknown prior to our work (Theorem~\ref{thm:constant}).

\subsection{Discussion and Open Problems}\label{sec:future}
We introduce the notion of symmetric menu complexity, and provide the first subexponential time approximation schemes and subexponential bounds on the symmetric menu complexity of $(1-\varepsilon)$-approximately optimal auctions for an unbounded unit-demand buyer over independent items. Our main technical innovations are: (a) a black-box reduction from computational/menu/symmetric menu complexity bounds on unbounded distributions to almost-bounded ones (Theorem~\ref{thm:reduction}), and (b) establishing that a wide class of (asymmetric) almost-bounded distributions are ``close to'' symmetric distributions in a formal sense (including unit-demand over independent items, or additive over independent items of constant support). We also conclude approximation-schemes for a suite of additional classes of valuation functions.

The notion of symmetric menu complexity itself will likely be of independent interest for future work. Symmetric menu complexity is well defined for any menu, and is always at most the menu complexity. Additionally, an additive or unit-demand buyer can always find their favorite option on a menu of symmetric menu complexity $C$ in $\poly(n,C)$ value queries (see Lemma~\ref{lem:polytimemenu} for a short proof). Moreover, Theorem~\ref{thm:apxsrev} establishes that selling separately can be approximated arbitrarily well by a menu of linear symmetric menu complexity. These arguments suggest that symmetric menu complexity is a convincing simplicity measure for additive/unit-demand buyers, and the following open questions are directly relevant:

\begin{openq} Does there exist a polynomial-time (respectively, polynomial symmetric menu complexity) approximation scheme for a single unit-demand buyer over independent items?

Does there exist a subexponential-time (respectively, subexponential symmetric menu complexity) approximation scheme for a single additive (or subadditive) buyer over independent items? 
\end{openq}

\begin{openq} Does there exist a subexponential-time approximation scheme for \emph{multiple} buyers who are (unit-demand/additive/subadditive) over independent items?
\end{openq}

Still, symmetric menu complexity by no means ``dominates'' the traditional menu complexity (for example, a subadditive buyer can find her favorite option on a menu of menu complexity $C$ in $\poly(n,C)$ value queries, but the same is not necessarily true for a menu of symmetric menu complexity $C$, or additive menu complexity $C$). It is therefore also an important open question (Open Problem~1.6 in~\cite{BabaioffGN17}) to understand the standard menu complexity required to achieve a $(1-\varepsilon)$ approximation in any of the settings considered in this paper. In this direction, note importantly that our Theorem~\ref{thm:reduction} allows future work to restrict attention only to almost-bounded distributions.

\subsection{Roadmap}
Preliminaries are split into two sections: Section~\ref{sec:prelim} contains the minimal notation necessary to formally state and overview our results, and Section~\ref{sec:setup} contains the remaining preliminaries necessary for complete proofs. Section~\ref{sec:reduction} contains a formal statement and brief overview of our reduction from unbounded to almost-bounded instances. Section~\ref{sec:symmetries} overviews our use of symmetries to derive approximation schemes for asymetric distributions. Section~\ref{sec:srev} overviews the connection between selling separately and symmetric menu complexity.

Section~\ref{sec:nudge} starts our proofs with a complete analysis of various ``nudge-and-round'' arguments (including a new one, Proposition~\ref{prop:advancedNudge}).  Sections~\ref{sec:reductionproofs} through~\ref{sec:srevproofs} contain complete proofs of our main results. Section~\ref{sec:approachLimits} overviews a barrier example (details in Appendix~\ref{app:approachLimits}).


\section{Preliminaries} 
\label{sec:prelim}
In the interest of brevity, we first provide the minimal notation necessary to understand our precise statements and proof overviews. Additional notation for detailed proofs is provided in  Section~\ref{sec:setup}.

\subsection{Classes of Distributions}
This paper considers instances with a single buyer and $n$ items. The buyer's valuation function $v(\cdot)$ for the items is drawn from some distribution $D$  (written as $v \leftarrow D$), which will always have \emph{independent items}: $D := \prod_i D_i$. We will consider the following  classes of valuations: 
\begin{itemize}

\item \textbf{$k$-demand over independent items.} Each $D_i$ is a single-dimensional distribution. The buyer's value $v_i$ for item $i$ is drawn independently from $D_i$, and her value for a set $S$ is $\max_{U \subseteq S, |U| \leq k}\{\sum_{i \in U} v_i\}$. When $k=1$ we say the distribution is \textbf{unit-demand} and when $k = n$ we say it is \textbf{additive}. 

\item \textbf{Subadditive over independent items.} Each $D_i$ is an arbitrary distribution.\footnote{\cite{RubinsteinW15} defines $D_i$ to be a distribution over a compact subset of a normed space, but this is not necessary.} Denote by $X_i$ a random variable with distribution $D_i$. There exists a function $V(\cdot, \cdot):\text{support}(D) \times 2^{[n]} \rightarrow \mathbb{R}_+$ for which a buyer with type $\vec{X}$ has valuation function $v_{\vec{X}}(\cdot)$ satisfying $v_{\vec{X}}(S):= V(\vec{X},S)$. Moreover, for all $\vec{X} \in \text{support}(D)$, function $V(\vec{X},\cdot)$ is monotone and subadditive.\footnote{That is, for all $S, T$: $v(S) \leq v(S \cup T) \leq v(S) + v(T)$.} We will often abuse notation and think of the valuation function $v(\cdot)$ as being drawn directly from $D$.\footnote{We refer the reader to~\cite{RubinsteinW15} for some examples of natural distributions satisfying this definition.}
\end{itemize}

We will use $D_S:= \prod_{i\in S} D_i$ to refer to the distribution $D$ restricted to items in $S$. Note that we will often abuse notation and use $v(S)$ to refer to $\mathbb{E}_S[v(S)]$ when $S$ is a randomized allocation.

\subsection{Revenue Benchmarks}
We will also be interested in the following quantities. If the parameter $D$ is clear from context, we may drop it (but sometimes it will not be clear, and we will make sure to include it). 
\begin{itemize}
\item $\REV(D)$: the optimal revenue achievable by any mechanism for a single buyer drawn from $D$ (formally, the supremum of achievable revenues). We will always assume that $\REV(D)$ is finite.
\item $\REV_M(D)$: for mechanism $M$, the expected revenue of $M$ for distribution $D$.
\item $\VAL(D)$: the expected value of $v([n])$, when $v(\cdot)$ is drawn from $D$ (not necessarily finite).
\end{itemize}

\subsection{Bounded and Truncated Distributions}
In order to formally state our results, we will be interested in the following restrictions on $D$.

\begin{itemize}
\item A distribution $D$ is \emph{unbounded} if $\REV(D) < \infty$ (but maybe $\VAL(D) = \infty$, no other constraints).
\item A distribution $D$ is $c$-\emph{bounded} if $\Pr_{v \leftarrow D}\big[v(\{i\}) \leq c\cdot \REV(D) \big] = 1$ for all $i$. 
\item $D$ is \emph{almost $c$-bounded} if $\exists X\in \mathbb{R}$ so that for all $i$, $\Pr_{v \leftarrow D}\big[v(\{i\}) \in [0,c \cdot \REV(D)] \cup \{X\} \big]=1$.
\end{itemize}

Importantly, observe that whenever $D$ is almost $c$-bounded, we can normalize so that all $v(\{i\}) \in [0,1]\cup \{X/(c\cdot \REV(D))\}$ with probability $1$ (by dividing all values by $c\cdot \REV(D)$). Now, an additive $\varepsilon/c$-approximation to $\REV(D)$ immediately implies a multiplicative $(1-\varepsilon)$-approximation to $\REV(D)$. 

Our main results will also involve truncating unbounded distributions into ones which are nearly-bounded. Below we define these truncations formally. The definition below is parameterized by a value $T > 0$ and a vector $\vec{p}$. Intuitively, the truncation operation first replaces all item values $>T$ with exactly $T$, and then for each item $i$ independently sets a huge value $n^2 \cdot (\max\{1,T\})^3$.

\begin{defn}[Canonical truncations] Let $D$ be subadditive over independent items. Let $T \in \mathbb{R}_+$, and let $\vec{p} \in \mathbb{R}^n$ be a vector of probabilities. Denote by $D(T,\vec{p})$ the truncation of $D$ with respect to $T,\vec{p}$. To sample from the distribution $D(T,\vec{p})$:
\begin{enumerate}
\item Draw $v \leftarrow D$. For each item $i$ such that $v(\{i\}) > T$, add $i$ to $S$. These items will have their value truncated at $T$.
\item For each item $i$, independently add $i$ to $W$ with probability $\min \big\{\frac{p_i}{n^2 \cdot(\max\{1,T\})^3},1 \big\}$. Update $S:=S\setminus W$. These items will have their value set at $n^2 \cdot (\max \{1,T\})^3$. 
\item Set $v'(\{i\}): = T$ for all $i \in S$, $v'(\{i\}) = n^2\cdot (\max\{1,T\})^3$ for all $i \in W$.
\item \emph{(Additive truncation)} Output $v'(\cdot)$ with $v'(U):= v(U \cap \bar{S}\cap \bar{W}) +\sum_{i \in U \cap (S \cup W)} v'(\{i\})$.
\item \emph{(Max truncation)} Output $v'(\cdot)$ with $v'(U):= \max\big\{v(U \cap \bar{S}\cap \bar{W}),\max_{i \in U \cap (S \cup W)} \{v'(\{i\})\} \big\}$. 
\end{enumerate}

We also use the notation $D(T):=D(T,\vec{0})$.
\end{defn}

Our reduction from unbounded to almost-bounded requires truncating the original distribution, and holds for either the additive or max truncation (or many others), so we will not emphasize which is used. We quickly parse what is going on in the definition. Both truncations first initialize $v'(\{i\}):= \min\{v(\{i\}),T\}$. For each $i$, both truncations then independently select each $i$ with tiny probability\footnote{In all applications of this definition, we will have $p_i/T \ll \varepsilon$.} and update $v'(\{i\}):=n^2 \cdot \max\{1,T\}^3$. Afterwards, in order to output a complete set function, $v'(\cdot)$ must be defined on all sets (not just the singletons), and the two truncations extend differently. 

Observe that when $D$ is additive over independent items and $\Pr\big[v(\{i\})\leq T\big] =1$, then $D(T) = D$ under the additive truncation. The same holds for unit-demand and the max truncation. If all we know is that $D$ is subadditive (and $\Pr\big[v(\{i\})\leq T\big] = 1$), then $D(T)$ does not necessarily equal $D$ under either truncation (but this is fine from the perspective of our results). More importantly, observe that if $D$ is subadditive (resp. XOS, submodular) over independent items, then $D(T,\vec{p})$ is also subadditive (resp. XOS, submodular) over independent items under both truncations. If $D$ is additive (resp. gross substitutes) over independent items, then $D(T,\vec{p})$ is additive (resp. gross substitutes) over independent items under the additive truncation. If $D$ is unit-demand over independent items, then $D(T,\vec{p})$ is unit-demand over independent items under the max truncation. So all of these classes are ``closed'' under (at least) one of the canonical truncations.

\subsection{Menu Complexity}\label{sec:menudefs}

We will consider two menu complexity measures in this paper. Recall that the Taxation Principle~\cite{Hammond79, Guesnerie81} asserts that any mechanism for a single buyer can be represented as a menu of (randomized allocation, non-negative price) pairs, where the buyer selects their favorite pair from the menu (that is, the pair which maximizes the buyer's expected value for the randomized allocation minus the price paid). We will therefore directly refer to a mechanism $M$ as a menu/list of such pairs (which implicitly includes the pair $(\emptyset, 0)$). The first notion we consider is the standard menu complexity from~\cite{HartN13}.

\begin{defn}[Menu Complexity~\cite{HartN13}] The \emph{menu complexity} of a menu $M$ is simply the size of the list $|M|$. We denote by $\mc(M)$ the menu complexity of $M$.
\end{defn}

The following two definitions introduce our notion of symmetric menu complexity. 

\begin{defn}[Symmetries in a Menu] Let $S$ be a randomized allocation, $p$ be a price, and $\Sigma$ be a subgroup of permutations of $[n]$. Then we denote by $(S, p, \Sigma)$ the set of (randomized allocation, price) pairs $\bigcup_{\sigma \in \Sigma}\{(\sigma(S), p)\}$. That is, the set $(S,p, \Sigma)$ contains, for all $\sigma \in \Sigma$, the option to receive  for price $p$ the randomized allocation which instantiates the random set $S$, and then permutes the items according to $\sigma$.
\end{defn}

\begin{defn}[Weak/Strong Symmetric Menu Complexity] We say that a mechanism $M$ has \emph{strong symmetric menu complexity} equal to the smallest $c$ such that there exists an index set $\mathcal{I}$ of size $c$, collection of (randomized allocation, price) pairs $\{(S_i, p_i)\}_{i \in \mathcal{I}}$, and subgroup $\Sigma$ of item permutations such that $M$ can be written as $\bigcup_{i \in \mathcal{I}} \{(S_i,p_i,\Sigma)\}$. We refer to the strong symmetric menu complexity of $M$ as $\ssmc(M)$.

We say that $M$ has \emph{weak symmetric menu complexity} equal to the smallest $c$ such that there exists an index set $\mathcal{I}$ and menus $\{M_i\}_{i \in \mathcal{I}}$ such that each $\ssmc(M_i) = d_i$ for all $i$, menu $M=\bigcup_i M_i$, and $\sum_i d_i = c$. We will refer to the weak symmetric menu complexity of $M$ as $\wsmc(M)$.
\end{defn}

Above, the idea is that the mechanism designer can present any mechanism $M$ to the buyer with a description of $\Sigma$ via its generating set, together with a list of $\ssmc(M)$ (randomized allocation, price) pairs. Similarly, the designer can present any mechanism $M$ to the buyer with a set of such lists, totaling $\wsmc(M)$ (randomized allocation, price) pairs (again representing each $\Sigma_i$ via its generating set). 

In principle, one might find some subgroups $\Sigma$ to be simpler than others (e.g., the subgroup of all permutations, or all permutations on even elements, etc.), but Jerrum's filter establishes that all subgroups have a generating set of size at most $n$~\cite{Jerrum82}. So while some subgroups may indeed be conceptually simpler than others, from the point of view of how much space is needed to define $\Sigma$, the space is always $n\ln(n)$ (this sanity checks, for instance, that it is not the case that all menus have low symmetric menu complexity simply because they can be cleverly partitioned into few heavily-symmetric parts. See further discussion in Section~\ref{sec:nudgeDiscussion}).

Note also that the weak/strong symmetric menu complexity is well-defined for any menu $M$ (by taking $\Sigma$ to be the trivial subgroup), and that for all $M$, we have $\wsmc(M) \leq \ssmc(M) \leq \mc(M)$. This is in contrast to previously posed notions such as ``additive menu complexity''~\cite{HartN13}, as some menus may simply not admit an additive description (and therefore their additive menu complexity is undefined)~\cite{BabaioffNR18}.

To simplify presentation, we formally define what it means for a class of distributions to have a low $(1-\varepsilon)$-approximation menu complexity.

\begin{defn}[$\varepsilon$-Menu Complexity of a Class of Distributions]
Let $\mathcal{D}$ be a class of distributions. Define the $\varepsilon$-Menu Complexity $\mc(\mathcal{D},\varepsilon)$ of $\mathcal{D}$ to be the minimum $c$ such that for all $D \in \mathcal{D}$ there exists a menu $M$ with $\mc(M) \leq c$ and $\REV_M(D) \geq (1-\varepsilon)\REV(D)$. We also define $\wsmc(\mathcal{D},\varepsilon)$ and $\ssmc(\mathcal{D},\varepsilon)$ similarly.
\end{defn}

\subsection{Computational Problems}
Finally, we define the computational problem we consider for our PTAS/QPTAS. Below, when we describe a distribution $D$ as being input, we do not explicitly specify how the input is given, other than (a)~it is possible to sample from $D$ in time $\poly(n)$ and (b)~for any $T \in \mathbb{R}$, $\varepsilon > 0$, and all items $i$, it is possible to find $\sup_{p\geq T}\big\{p\cdot \Pr[v(\{i\}) \geq p]\big\}$, along with an $r$ satisfying $r \cdot \Pr[v(\{i\})\geq r] \geq (1-\varepsilon) \cdot \sup_{p \geq T}\big\{p \cdot \Pr[v(\{i\})\geq p] \big\}$ in time $\poly_{\varepsilon}(n)$.\footnote{Observe that this supremum is always finite when $\REV(D)$ is finite.} Observe that if the support of each $D_i$ is explicitly listed and of size $\poly(n)$, then both these properties are satisfied (even though the support of $D$ is exponential in $n$).

\begin{defn}[Implicit description of a menu~\cite{DaskalakisDT14}] An implicit description of a menu $M$ is a Turing machine which takes as input a valuation $v(\cdot)$ and outputs $\arg\max_{(S,p) \in M\cup \{(\emptyset,0)\}}\{v(S)-p\}$. The description has overhead $c$ if on input $v(\cdot)$ described using $b$ bits, the Turing machine terminates in time $\poly(c,b)$. 
\end{defn}

\begin{defn}[Computational Revenue Maximization] A $(1-\varepsilon)$ approximation for the problem $\revmax_{\mathcal{D}}$ takes as input $D \in \mathcal{D}$ and outputs an implicit description of a menu $M$ such that $\REV_M(D) \geq (1-\varepsilon)\REV(D)$. Whenever we say that an algorithm for $\revmax_{\mathcal{D}}$ runs in time $c$, we mean both that the implicit description is found in time $c$, and that the implicit description itself has overhead $c$.
\end{defn}



\section{Overview: Reduction from Unbounded to Bounded}
\label{sec:reduction}

In this section, we overview our poly-time (symmetric) menu-complexity preserving reduction from unbounded distributions to almost-bounded distributions. Theorem~\ref{thm:reduction} is the main result of this section.

\begin{theorem}\label{thm:reduction}
For any $n \in \mathbb{N}$ and $\varepsilon > 0$, let $\mathcal{D}$ be a class of distributions that is closed under one of the canonical truncations, such that every $D \in \mathcal{D}$ is subadditive over $n$ independent items. Let $\mathcal{D}_B$ denote the subset of $\mathcal{D}$ that is also almost $1/\varepsilon^4$-bounded. Then there is a $\poly(n,1/\varepsilon)$-time reduction from achieving a $(1-\varepsilon)$-approximation to $\revmax_{\mathcal{D}}$ to achieving a $(1-O(\varepsilon))$-approximation to $\revmax_{\mathcal{D}_B}$. Moreover:
\begin{align*}
\mc(\mathcal{D},\varepsilon) \leq n + n\mc(\mathcal{D}_B,O(\varepsilon))\quad \text{and} \quad \wsmc(\mathcal{D},\varepsilon) \leq n + n\wsmc(\mathcal{D}_B,O(\varepsilon)).
\end{align*}
\end{theorem}

The proof of Theorem~\ref{thm:reduction} is broken down into two main parts. The first half, captured in Proposition~\ref{prop:main4},\footnote{Note that Proposition~\ref{prop:main4} is a less precise version of Proposition~\ref{prop:main3}.} asserts that for any distribution which is subadditive over independent items, there exists a $(1-\varepsilon)$-approximate menu of a particular form. Readers familiar with~\cite{BabaioffGN17} will notice simliarity to their Lemma~2.4; we discuss the differences shortly after.

\begin{proposition}\label{prop:main4}
Let $E \geq \REV(D)/\varepsilon^3$. Then for all $D$ that are subadditive over independent items, there exists a menu $M$ such that $\REV_M(D) \geq (1-O(\varepsilon))\REV(D)$ and:
\begin{itemize}
\item For all $(S,p) \in M$, either $p \leq E$, or there exists at most one $i$ such that $\Pr[i \in S] > 0$.
\item For each item $i$, there exist at most two distinct $(S,p) \in M$ such that $p > E$ and $\Pr[i \in S] > 0$. 
\end{itemize}
\end{proposition}

The structure of the promised $M$ is identical to Lemma~2.4 of~\cite{BabaioffGN17}. The key difference is that we take $E \geq \REV(D)/\varepsilon^3$, versus their $E \geq n^3\REV(D)/\varepsilon^2$. This quantitative improvement is crucial for Theorem~\ref{thm:reduction}: without it, instead of reducing to $1/\varepsilon^4$-bounded distributions, we would only reduce to $\poly(n/\varepsilon)$-bounded distributions. Our positive results for $c$-bounded distributions require runtime/symmetric menu complexity exponential in $c$, so the quantitative difference is significant. The second difference is the extension to distributions which are subadditive over independent items (their Lemma~2.4 holds for additive). 

The second half of Theorem~\ref{thm:reduction} is Proposition~\ref{prop:reduction} below.\footnote{Again, Proposition~\ref{prop:reduction} is slightly less precise than Propositions~\ref{prop:reduce1} and~\ref{prop:reduction2}.}  We first need some definitions. Definition~\ref{def:exclusive} describes an operation which appears in~\cite{BabaioffGN17}, which takes a menu $M$ and replaces all ``expensive'' options in $M$ with options which award at most a single item with non-zero probability. Definition~\ref{def:excloptions} defines a new operation, which takes a menu $M$ and concatenates it with $n$ new options which each offer a single item deterministically.

\begin{defn}[Making a menu $E$-exclusive]\label{def:exclusive} For a given menu $M$, let the menu $M|_E$ (``$M$ made exclusive above $E$'') denote the menu constructed from $M$ as follows:
\begin{itemize}
\item For any $(S,p) \in M$ with $p \leq E$, add $(S, (1-\varepsilon)p)$ to $M|_E$.
\item For any $(S,p) \in M$ with $p > E$, and all items $i$, let $X_i(S)$ denote the (randomized) set that is $\{i\}$ with probability $\Pr[i \in S]$ and $\emptyset$ otherwise. Add $(X_i(S),(1-\varepsilon)p)$ to $M|_E$.
\end{itemize}
\end{defn}

\begin{defn}[Concatenating a menu with exclusive options]\label{def:excloptions} Let $\vec{r} \in \mathbb{R}^{[n]}$ be a vector of reserve prices, and $T \in \mathbb{R}$. Let also $(S_i, q_i)$ be the option in $M$ that would be purchased by a buyer with value $v(\cdot)$ satisfying $v(S) = T\cdot \mathbb{I}(i \in S)$.\footnote{Morally, one should think of $(S_i,q_i)$ as the option in $M$ which awards $i$ with highest probability. However, if $M$ has infinite menu complexity, then this option need not be well-defined. The definition is given as such to avoid overly cumbersome notation with supremums.} Then $M^{T,\vec{r}}$ (``$M$ concatenated with exclusive options $\vec{r}$'') is the menu $M$ with the $n$ additional options $\bigcup_i \big\{(\{i\}~,~q_i + r_i \cdot (1-\Pr[i \in S_i])) \big\}$, and then multiplying all prices by $(1-\varepsilon)$.

That is, concatenating a menu with exclusive options $\vec{r}$ adds, for all $i$, an option to purchase item $i$ deterministically. The price-per-additional-probability of getting item $i$ beyond what is already allocated by $S_i$ is $r_i$ (and then all prices are multiplied by $(1-\varepsilon)$).
\end{defn}

\begin{proposition}\label{prop:reduction} Let $M$ be the mechanism promised by Proposition~\ref{prop:main4}, and let $T\geq E/\varepsilon \geq \REV(D)/\varepsilon^4$. Also, let $p_i:=\sup_{r \geq T} \big\{r\cdot \Pr[v(\{i\})\geq r] \big\}$ and let $r_i\geq T$ be such that $r_i\cdot \Pr[v(\{i\})\geq r_i] \geq (1-\varepsilon)p_i$. Then:
\begin{itemize}
\item $(1-O(\varepsilon)) \cdot \REV_M(D) \leq \REV(D(T,\vec{p})).$
\item For any $M'$, we have $\REV_{(M'|_E)^{T,\vec{r}}}(D) \geq \REV_{M'}(D(T,\vec{p})) - O(\varepsilon)\cdot \REV(D)$.
\end{itemize}

\end{proposition}

Proposition~\ref{prop:reduction} establishes both that the optimal revenue for $D$ and $D(T,\vec{p})$ is close, and also that any mechanism for $D(T,\vec{p})$ can be efficiently transformed into one which achieves similar guarantees for $D$. Proposition~\ref{prop:reduction} has no real analogue in~\cite{BabaioffGN17}, but replaces their Lemma~2.5. Their Lemma~2.5 specifies a particular discretization of the ``cheap part'' of the menu promised in Lemma~2.4 (priced $\leq E$) which is compatible with the remaining $2n$ expensive options (priced $>E$). The entire challenge of this process is ensuring that a buyer with value $v(\cdot)$ who chooses to purchase one of the $2n$ expensive options from the promised $M$ does not all of a sudden wish to purchase a cheap option instead after discretization. As a result, one cannot simply view the cheap part and expensive part as separate subproblems. The main insight in~\cite{BabaioffGN17}'s Lemma~2.5 is that a particular discretization of the cheap part does not interfere with the $2n$ expensive options. The main insight in our Proposition~\ref{prop:reduction} is more general: the \emph{only} property of the cheap part which interacts with the expensive part is the maximum probability with which item $i$ is ever allocated (and this is captured in a somewhat roundabout way by the inserted point-masses at $n^2(\max\{1,T\})^3$ --- see Lemma~\ref{lem:modrev}). So in comparison to Lemma~2.5 in~\cite{BabaioffGN17}, the key contribution of Proposition~\ref{prop:reduction} is that it provides a true reduction from unbounded to almost-bounded distributions. Future work (and the remainder of the present work) can simply focus on almost-bounded distributions, rather than separately ensuring that the resulting menu is compatible with Proposition~\ref{prop:main4}. 

Observe that Theorem~\ref{thm:reduction} now follows from Propositions~\ref{prop:main4} and~\ref{prop:reduction}. For any $D \in \mathcal{D}$, $D(T,\vec{p})$ is almost $1/\varepsilon^4$-bounded. So if we can find a $(1-O(\varepsilon))$-approximation for $D(T,\vec{p})$, we can efficiently make it $E$-exclusive and concatenate the expensive options, and these changes increase the (symmetric) menu complexity by at most a factor of $n$, plus an additional $n$. Chaining the inequalities in Propositions~\ref{prop:main4} and~\ref{prop:reduction} establishes that the resulting menu is a $(1-O(\varepsilon))$-approximation. A more detailed outline and complete proofs can be found in Section~\ref{sec:reductionproofs}. 

\section{Overview: Symmetries and Optimal Mechanisms}\label{sec:symmetries}
With Theorem~\ref{thm:reduction} in hand, our task is now to design good menus for almost-bounded distributions. Unfortunately, there is not much prior work in this direction (even for bounded distributions). It is only known that, via an application of nudge-and-round arguments, one can discretize all values into multiples of $\varepsilon/n$ while losing at most an $\varepsilon$ fraction of the optimal revenue, at which point an exponentially large linear program can output an explicit description of the optimal mechanism (for the discretized distribution). Note that the linear program has size polynomial in the support of $D$, which would remain exponential in $n$ even if $D$ is unit-demand/additive and each marginal has support size $2$ (as $D$ is a product distribution). We overview this linear program in Section~\ref{sec:symmetriesproofs}.

Despite this simple result for bounded distributions, prior work of~\cite{BabaioffGN17} was the first to establish even that \emph{some} bounded menu complexity suffices for unbounded additive distributions over independent items, and Theorem~\ref{thm:reduction} now allows us to do the same for unbounded distributions which are subadditive over independent items. A proof of the following theorem appears in Appendix~\ref{app:subadd}, which formalizes the above paragraph and mostly follows from Theorem~\ref{thm:reduction}.

\begin{theorem}\label{thm:subadd} Let $\mathcal{D}$ be the class of distributions which are subadditive over $n$ independent items. Then for all $\varepsilon$, there exists a finite number $C(n,\varepsilon) < \infty$ such that $\mc(\mathcal{D},\varepsilon) \leq C(n,\varepsilon)$.
\end{theorem}

One special case where progress was made is if the underlying distribution $D$ is \emph{symmetric}~\cite{DaskalakisW12}. Specifically, if $D$ is invariant under all permutations in $\Sigma$,~\cite{DaskalakisW12} shows that the canonical LP referenced above can be simplified to have size only $|\text{support}(D)/ \Sigma|$ (that is, the LP needs to only consider one representative from each equivalence class in the support of $D$ under $\Sigma$), which allowed them to conclude a PTAS when $D$ was bounded and unit-demand over i.i.d. items. We recap (a slight generalization of) their main result below, and include a proof in Section~\ref{sec:symmetriesproofs} for completeness.

\begin{defn}[Invariant under item permutations] We say that $\Sigma$ is an \emph{item permutation group} if there exists a partition $T_1 \sqcup \ldots \sqcup T_s$ of $[n]$ such that $\Sigma$ is the subgroup generated by $\bigcup_i \{(x,y)\}_{x,y \in T_i}$. That is, $\Sigma$ is generated by all swaps of pairs of elements in the same $T_i$.\footnote{Put another way, $\Sigma$ contains exactly permutations which separately permute items in $T_i$, for all $i$.} We further let $s$ (the number of parts necessary for the partition) denote the \emph{partition size} of $\Sigma$.

We say that $D$ is \emph{symmetric with respect to $\Sigma$} if for all $\sigma \in \Sigma$, the distribution which first draws $v(\cdot) \leftarrow D$ and then outputs $v'(\cdot)$ with $v'(S):= v(\sigma(S))$ is itself the distribution $D$.
\end{defn}

For example, if $D$ is $k$-demand over i.i.d. items, then $D$ is symmetric with respect to the group of all permutations on $n$ items, and this permutation group has partition size one. If $D$ is $k$-demand over independent items and all values for even items are drawn i.i.d., and odd items are drawn i.i.d. (but from a different distribution than the even items), $D$ is symmetric with respect to the item permutation group $\Sigma$ with $T_1:=$ even items and $T_2:=$ odd items, and therefore $\Sigma$ has partition size two.

\begin{theorem}[\cite{DaskalakisW12}]\label{thm:symmetries} For any item permutation $\Sigma$ of partition size $s$, let $\mathcal{D}$ be the class of distributions $D$ which are $k$-demand over independent items, with each $|\text{support}(D_i)| \leq c$, and symmetric with respect to $\Sigma$. Then an optimal solution to $\revmax_{\mathcal{D}}$ can be found in time $\poly(n^{cs})$. Moreover, the mechanism output $M$ has $\ssmc(M) \leq n^{cs}$.
\end{theorem}

Our applications of Theorem~\ref{thm:symmetries} will first discretize a distribution into one which is symmetric with respect to a $\Sigma$ of low partition size, and also where each marginal is supported on not many values. Theorem~\ref{thm:iiddemand} considers i.i.d. items, where the input is already heavily symmetric and we just need to bound the loss from discretizing values. We treat this case as a warmup, but even here our results make quantitative improvements on~\cite{DaskalakisW12}, and also now extend their work to the unbounded case due to Theorem~\ref{thm:reduction}. In particular, observe that Theorem~\ref{thm:iiddemand} concludes a PTAS for distributions which are unit-demand over i.i.d. items, and a QPTAS for distributions which are additive over i.i.d. items.

\begin{theorem}\label{thm:iiddemand} Let $\mathcal{D}$ be the class of valuations which are $k$-demand over i.i.d. items. Then $\wsmc(\mathcal{D},\varepsilon) \leq n^{O(\ln(k)/\varepsilon^{12})}$, and there exists a $(1-\varepsilon)$-approximation to $\revmax_\mathcal{D}$ which runs in time $n^{O(\ln(k)/\varepsilon^{12})}$.
\end{theorem}

Our main application of Theorem~\ref{thm:symmetries} will be on \emph{arbitrary} distributions which are unit-demand over independent items. That is, the initial distribution might have no symmetries whatsoever. Still, we show that it is possible to discretize the distribution in a way which creates symmetries. Note that discretizing only the values clearly no longer suffices, as there are fully asymmetric distributions even when each marginal has support size two. So we additionally need to discretize the probabilities. The main challenge here is that unless our discretization is excessively fine (that is, too fine to improve the runtime/menu complexity), there will almost certainly be \emph{some} item for which the original and the discretized values are quite different. So we need to carefully dive into the details of an advanced nudge-and-round argument (Proposition~\ref{prop:advancedNudge}) to figure out exactly which item values contribute to lost revenue. This careful dive is possible for unit-demand valuations because optimal menus award at most one item without loss of generality. For additive valuations, there is a barrier to this approach, which we expound in Section~\ref{sec:approachLimits}.

\begin{theorem}\label{thm:unit}Let $\mathcal{D}$ be the class of valuations which are unit-demand over independent items. Then there exists a $(1-\varepsilon)$-approximation to $\revmax_\mathcal{D}$ which runs in time $n^{O(\ln(n/\varepsilon))^{1/\varepsilon^{7}})}$ and $\wsmc(\mathcal{D},\varepsilon) \leq n^{O(\ln(n/\varepsilon))^{1/\varepsilon^{7}})}$.
\end{theorem}

Finally, Theorem~\ref{thm:constant} below establishes that a similarly careful nudge-and-round yields a quasi-polynomial approximation scheme for distributions where each marginal has support at most $c$ (even if that support is distinct for each marginal). Recall that even when $c =2$, no subexponential approximation schemes are previously known, and this is left open by~\cite{DaskalakisDT14}.

\begin{theorem}\label{thm:constant} Let $\mathcal{D}$ be the class of valuations which are $k$-demand over independent items and satisfy $\ |\text{support}(D_i)|\leq c$ for all $i$. Then there exists a $(1-\varepsilon)$-approximation to $\revmax_\mathcal{D}$ which runs in time $n^{O(\ln(n/\varepsilon))^{c}}$ and $\wsmc(\mathcal{D},\varepsilon) \leq n^{O(\ln(n/\varepsilon))^{c}}$.
\end{theorem}


\section{Overview: Selling Separately with Low Symmetric Menu Complexity}\label{sec:srev}
One justified critique of menu complexity is that it assigns menu complexity $2^n$ to the ``selling separately'' mechanism, which places price $p_i$ on each item and allows the buyer to purchase any set $S$ for price $\sum_{i \in S} p_i$. Symmetric menu complexity is still imperfect in this regard: if all $p_i$ are distinct, the menu will still have strong/weak symmetric menu complexity $2^n$. 

\cite{BabaioffGN17} provide a nice response to this critique, by proving that while technically selling separately is deemed to have $2^n$ menu complexity, for every $M$ which sells separately and for all $D$ which are additive over independent items, there exists another $M'$ for which $\REV_{M'}(D) \geq (1-\varepsilon)\REV_M(D)$ and $\mc(M') \leq n^{1/\varepsilon^{O(1)}}$.\footnote{Note that ``selling separately'' is not obviously simple when $D$ is not additive over independent items, so it is not clear that one should expect/demand such an $M'$ unless $D$ is additive over independent items.} So while the definition of menu complexity is certainly still imperfect, we at least now know that if (for instance) a distribution $D$ admits no good mechanisms of polynomial menu complexity, it is not because selling separately is close to optimal.

We provide an even stronger response in the case of symmetric menu complexity: when $D$ is additive over independent items, and $M$ sells separately, there exists another $M'$ for which $\REV_{M'}(D) \geq (1-\varepsilon)\REV_M(D)$ and $\ssmc(M') \leq f(\varepsilon)n$. That is, the blow up from the ``intrinsic complexity'' (or description complexity) of selling separately $n$ items to the strong symmetric menu complexity of a menu that is almost as good is just a multiplicative factor independent of $n$. The proof can be found in Section~\ref{sec:srevproofs}.

\begin{theorem}\label{thm:apxsrev} Let $D$ be additive over independent items and $M$ be a mechanism which sells separately. Then there exists a mechanism $M'$ with $\REV_{M'}(D) \geq (1-\varepsilon)\REV_M(D)$ and $\ssmc(M') \leq f(\varepsilon)n$, for $f(\varepsilon) = 2^{O(1/\varepsilon^3)}$.
\end{theorem}


\section{Setup for Complete Proofs}\label{sec:setup}
Below we introduce the necessary concepts to follow our complete proofs. First, our proofs will make use of the following additional revenue bounds:
\begin{itemize}
\item $\BREV(D)$: the optimal revenue achievable by ``bundling together'' (that is, putting a single take-it-or-leave-it price $p$ on $[n]$ and letting the buyer purchase iff $v([n]) \geq p$).
\item $\srev(D)$: the optimal revenue achievable by ``selling separately'' (that is, putting a take-it-or-leave-it price $p_i$ on each item $i$ and letting the buyer purchase the subset maximizing $v(S) - \sum_{i \in S} p_i$).
\item $\SREV(D)$: the optimal revenue achievable by ``selling exclusively'' (that is, putting a price $p_i$ on each singleton set $\{i\}$ and letting the buyer purchase the item maximizing $v(\{i\}) - p_i$)).\footnote{This serves the same purpose as $\SREV_{\vec{q}}(D)$ from~\cite{RubinsteinW15}, and could be substituted everywhere in their proofs. We find this definition to be conceptually cleaner.}
\end{itemize}

\subsection{Basic Lemmas}
Our work will make use of the following basic definition and lemma. Some definitions we use will instantiate a conditional distribution, and may be conditioned on probability zero events. To ease notation, we will define any such distribution to be ``null'' (that is, output $0$ with probability one).\footnote{We only use conditional distributions to analyze a restricted distribution. So if the conditioned event has probability $0$, the restricted distribution is null and requires no analysis. For ease of exposition, we will not explicitly note this every time it occurs.}

\begin{defn}[Conditional and Restricted Distributions] For a given event $E$, we define the distribution $D|E$ (``$D$ conditioned on $E$'') to be the distribution which outputs a valuation $v(\cdot)$ drawn from $D$, conditioned on event $E$. We also define $D\cdot \mathbb{I}(E)$ (``$D$ restricted to $E$'') to be the distribution which samples a valuation $v(\cdot)$ from $D$, then outputs $v(\cdot)$ if event $E$ occurs, and the zero function otherwise. 
\end{defn}

\begin{observation}[\cite{HartN17}] For all $D, E$ and  mechanisms $M$, we have $\Pr[E]\cdot \REV_M(D|E) = \REV_M(D\cdot \mathbb{I}(E))$.
\end{observation}

Lemma~\ref{lem:sds} below simply argues that we can count the revenue of a mechanism for distribution $D$ by counting the revenue of $M$ from each of $D$'s subdomains separately. This simple lemma is a surprisingly useful starting point for much of the work on approximately-optimal auctions.

\begin{lemma}[Sub-Domain Stitching~\cite{HartN17}]\label{lem:sds} Let $Y_1,\ldots, Y_k$ partition $\text{support}(D)$. Then for any menu $M$:
$$\REV_M(D) = \sum_i \REV_M(D \cdot \mathbb{I}(v \in Y_i)) \quad \text{and} \quad \REV(D) \leq \sum_i\REV(D\cdot \mathbb{I}(v \in Y_i)).$$
\end{lemma}

\subsection{Modified Revenue Maximization}
As mentioned in Section~\ref{sec:reduction}, the key property that matters for interaction between the ``cheap part'' of a menu and the ``expensive part'' is the maximum probability with which each item $i$ is ever awarded for a cheap price. We also noted that, in a somewhat roundabout fashion, this is captured by adding a point-mass to a bounded distribution. Making this roundabout connection results in cleaner theorem statements, because we do not need to introduce a new problem, but the more conceptually clear approach is to design a modified revenue maximization problem, which we call \modrevmax. We will also formally establish the connection between almost-bounded distributions and \modrevmax\ (so that our cleaner theorem statements are completely proven), but our proofs will directly deal with \modrevmax, before transferring to almost-bounded distributions at the very end via Lemma~\ref{lem:modrev} below.

\begin{defn}[Leftovers of a mechanism] For any menu $M$, define $\ell_i(M):= 1-\sup_{(S,p) \in M} \big\{\Pr[i \in S]\big\}$ to be the fraction of item $i$ which is never allocated in $M$. We refer to $\vec{\ell}(M)$ as the \emph{leftovers} of $M$.
\end{defn}

\begin{defn}[Modified Revenue Maximization]
A $(1-\varepsilon)$-approximation for $\modrevmax_{\mathcal{D},\mathcal{W}}$ takes as input a distribution $D \in \mathcal{D}$, and item weights $w_i \geq 0$ so that $\vec{w} \in \mathcal{W}$, and outputs an implicit description of a menu $M$ such that $\REV_M(D) + \sum_i w_i \ell_i (M) \geq (1-\varepsilon)\sup_{M'}\big\{\REV_{M'}(D) + \sum_i w_i \ell_i(M') \big\}$.
\end{defn}

Intuitively, \modrevmax\ can be thought of as the following: before the buyer drawn from $D$ arrives, the seller can set aside an $\ell_i(M)$ fraction of item $i$, receiving bonus revenue $\sum_i w_i \ell_i(M)$ for it, and then run mechanism $M$ and receiving revenue $\REV_M(D)$. The connection to almost-bounded distributions is as follows: if the support of $D_i$ is (say) $[0,1] \cup \{w_i/p_i\}$, then the maximum price we can charge a buyer with value $w_i/p_i$ to receive item $i$ with probability $1$, while using menu $M$, is $\ell_i(M)w_i/p_i + n$. So if we further have that $\Pr[v(\{i\}) = w_i/p_i] = p_i$, and $w_i \gg 1$, and $p_i \ll 1/n$, then our revenue for using menu $M$ concatenated with exclusive options to purchase item $i$ at $w_i/p_i$ is roughly $\REV_M(D) + \sum_i \ell_i(M) w_i$. We make this formal with Lemma~\ref{lem:modrev} below, whose proof is in Appendix~\ref{app:setup}.

\begin{lemma}\label{lem:modrev} Let $\mathcal{D}$ be a class of distributions that is subadditive over independent items, and $c$-bounded. Let also $\mathcal{D}'$ denote the closure of $\mathcal{D}$ under either of the canonical truncations. Then there is a polytime reduction (which makes a single black-box call) from $(1-\varepsilon)$-approximations for $\modrevmax_{\mathcal{D},\mathcal{W}}$ to $(1-O(\varepsilon))$-approximations for $\revmax_{\mathcal{D}'}$. Moreover, if the menu produced on the call to $\revmax_{\mathcal{D}'}$ has (symmetric) menu complexity $C$, the menu output for $\modrevmax_{\mathcal{D},\mathcal{W}}$ has (symmetric) menu complexity $C$ as well.
\end{lemma}

\section{Careful Nudge-and-Rounds}\label{sec:nudge}
One key idea in our arguments (and those of~\cite{BabaioffGN17}) is to bound the loss in revenue from a given menu as the input distribution is perturbed slightly, or for a given distribution if the menu is perturbed slightly. Such arguments are typically termed ``nudge-and-round,'' and the initial such argument is commonly attributed to Nisan~\cite{BalcanBHM05,ChawlaHK07}. 

The simple arguments commonly used do not suffice for our results (see brief discussion follow Proposition~\ref{prop:carefulNudge}). Below, we introduce two ``careful'' nudge-and-round arguments. We will actually need the most careful Proposition~\ref{prop:advancedNudge} in designing our quasi-polynomial sized menu for unit-demand buyers (and we will highlight there why we need the careful statement). For the rest of our arguments, either Corollary~\ref{cor:expectedNudge} which appears in~\cite{RubinsteinW15} will suffice, or Proposition~\ref{prop:carefulNudge} will suffice.

\begin{defn}[Coupling error with respect to $M$]\label{def:advancedCoupling} Let $D$ and $D'$ be coupled valuation distributions (with coupled draws $v(\cdot), v'(\cdot)$), and let $M$ be some menu. Let also $S_M(\cdot)$ be some function that takes as input a valuation $v(\cdot)$ and selects a randomized allocation $S$ available on menu $M$. Then we say that the coupling error of $D, D'$ with respect to menu $M$ is (the infimum over all couplings between $D$ and $D'$ of):
$$\delta_M(D, D'):= \sup_{S_M(\cdot)} \big\{\mathbb{E}_{(v, v') \leftarrow (D, D')}[v(S_M(v)) - v'(S_M(v))]\big\} + \sup_{S_M(\cdot)} \big\{\mathbb{E}_{(v, v') \leftarrow (D, D')}[v'(S_M(v')) - v(S_M(v'))] \big\}.$$

That is, $\delta_M(D, D')$ denotes the supremum, over all mappings from valuations $v$ to randomized allocations $S_M(v)$ available in menu $M$ of the expected difference between $v(S_M(v))$ and $v'(S_M(v))$, plus the same quantity swapping the roles of $v$ and $v'$. 
\end{defn}

\begin{defn}[Coupling error] Let $D$ and $D'$ be coupled valuation distributions (with coupled draws $v(\cdot), v'(\cdot)$). Then the coupling error between $D$ and $D'$ is (the infimum over couplings between $D,D'$ of):\footnote{Note that this is the same as the first Wasserstein distance between $D$ and $D'$ if the underlying metric space for valuation functions has $d(v,v'):= \max_{S \subseteq [n]}\{v(S) - v(S')\} + \max_{S \subseteq [n]}\{v'(S) - v(S)\}$.}
$$\delta(D, D') := \mathbb{E}_{(v, v') \leftarrow (D, D')} \big[\max_{S \subseteq [n]} \{v(S) - v'(S)\} + \max_{S \subseteq [n]}\{v'(S) - v(S)\} \big].$$
\end{defn}

\begin{observation}\label{obs:simpleNudge} For all mechanisms $M$, we have $\delta_M(D, D') \leq \delta(D, D')$.
\end{observation}

The high-level distinction between $\delta_M(D, D')$ and $\delta(D, D')$ is swapping the order of an expectation and a supremum. The function $\delta_M(D, D')$ first picks the supremum mapping, and then takes an expectation whereas $\delta(D, D')$ takes the supremum inside the expectation, and is therefore larger. Most prior work (e.g.~\cite{DaskalakisW12, BabaioffGN17}) typically takes a very loose bound on $\delta(D,D')$ with $\max_{v,v' \in \text{support}(D)}\{v([n])-v'([n])\}$ (or perhaps replaces $[n]$ with $S$ if $v,v'$ are known to be identical on all items $\notin S$). Some prior work (e.g.~\cite{RubinsteinW15}) does try to get extra mileage out of the expectation and use tighter bounds on $\delta(D,D')$. No prior work has used the truly refined $\delta_M(D,D')$ before. As mentioned previously, this extremely precise nudge-and-round is necessary for our approximation schemes for a unit-demand buyer. We first quickly establish that both errors are in fact metrics, which will allow us to later use the triangle inequality. The proof of Lemma~\ref{lem:triangle} appears in Appendix~\ref{app:triangle}.

\begin{lemma}\label{lem:triangle} $\delta(\cdot,\cdot)$ is a metric. For all $M$, $\delta_M(\cdot,\cdot)$ is a metric.
\end{lemma}

Proposition~\ref{prop:advancedNudge} is our main nudge-and-round argument. The proof is indeed similar to simpler claims which have appeared in prior work, but more carefully keeps track of the error induced. We include the complete proof below.

\begin{proposition}\label{prop:advancedNudge} Consider two valuation distributions $D$ and $D'$. For any mechanism $M$ and $\varepsilon > 0$, let $M'$ denote the mechanism with the same randomized allocations as $M$, but all prices multiplied by $(1-\varepsilon)$. Then:
$$\REV_{M'}(D') \geq (1-\varepsilon)\REV_M(D) - \delta_M(D, D')/\varepsilon.$$
\end{proposition}
\begin{proof}
Consider taking the mechanism $M$, where buyer with valuation $v$ chooses to purchase randomized allocation $S(v)$ for price $p(v)$. Then we certainly have, for all coupled $(v,v')$, and $(S(v'),p(v'))$ on menu $M$:
\[ v(S(v)) - p(v) \geq v(S(v')) - p(v').
\]

Now consider taking the mechanism $M$, and multiplying all prices by $(1-\varepsilon)$. Consider $v$'s couple, $v'$, who chooses to purchase some (randomized allocation, price) pair $(S(v'),(1-\varepsilon)p(v'))$ instead of the option $(S(v),(1-\varepsilon)p(v))$. Then we have:
\[ v'(S(v')) - (1-\varepsilon)p(v') \geq v'(S(v)) - (1-\varepsilon)p(v).
\]

Summing these two inequalities together, and dividing both sides by $\varepsilon$ yields:
\begin{align*}
p(v') &\geq p(v) +\big(v(S(v')) - v'(S(v'))\big)/\varepsilon - \big(v(S(v))-v'(S(v))\big)/\varepsilon
\end{align*}

Importantly, observe that the randomized allocation $S(v')$ depends only on $v'$, and $S(v)$ depends only on $v$. Therefore, when we take an expectation over both sides, we get that:
\begin{align*}
\mathbb{E}_{(v, v') \leftarrow (D, D')}[p(v')] & \geq \mathbb{E}_{(v, v') \leftarrow (D, D')}\Big[p(v) -\big(v'(S(v')) -v(S(v'))\big)/\varepsilon - \big(v(S(v))-v'(S(v))\big)/\varepsilon \Big]\\
& \geq \REV(D) - \delta_M(D, D')/\varepsilon,
\end{align*}
where the last inequality follows by observing that $S(\cdot)$ is one candidate function for the supremum in $\delta_M(D, D')$, so the actual value of $\delta_M(D, D')$ can only be larger.
\end{proof}

\begin{corollary}[\cite{RubinsteinW15}]\label{cor:expectedNudge} Consider two coupled valuation distributions $D,D'$ (with coupled draws $(v(\cdot), v'(\cdot))$). Thsen for all $\varepsilon > 0$:
$$\REV(D') \geq (1-\varepsilon)\REV(D) - \delta(D, D')/\varepsilon.$$
\end{corollary}

Again, the main difference between Corollary~\ref{cor:expectedNudge} and the simpler nudge-and-rounds used (for instance) in~\cite{DaskalakisW12,BabaioffGN17} is that Corollary~\ref{cor:expectedNudge} attempts to get some traction via expectations, whereas simpler arguments from prior work upper bound $\delta(D,D')$ by (e.g.) $\max_{v \in \text{support}(D)}\{v([n])\}$ (or if $D$ and $D'$ only differ in values for items $S$, then $\max_{v \in \text{support}(D)}\{v(S)\}$). We will also need the following argument, which changes a menu (rather than changing the distribution). The proof is similar, and deferred to Appendix~\ref{app:setup}.

\begin{proposition}\label{prop:carefulNudge} Let $M$ be a menu, and let $(S(v), p(v))$ denote the (randomized allocation, price) selected by a buyer with value $v$ from $M$. Let also $M'$ denote a menu which for each $(S(v),p(v)) \in M$ adds $(S'(v),(1-\varepsilon)p)$ to $M'$, where for all $v' \in \text{support}(D), v'(S(v)) \geq v'(S'(v))$ (that is, it makes all allocations no better, but offers a multiplicative discount). Then:
$$\REV_{M'}(D) \geq (1-\varepsilon)\REV_M(D) - \mathbb{E}[v(S(v))-v(S'(v))]/\varepsilon.$$
\end{proposition}

\subsection{Brief Discussion on Nudge-and-Rounds}\label{sec:nudgeDiscussion}
Here, we provide some brief context for the above two nudge-and-round arguments, their intended use, and the relationship to symmetric menu complexity. We'll refer to nudge-and-rounds provided via arguments such as Proposition~\ref{prop:advancedNudge} and Corollary~\ref{cor:expectedNudge} as ``nudging the distribution'', and those via arguments such as Proposition~\ref{prop:carefulNudge} as ``nudging the allocation.''

Distribution-nudges have been used in works such as~\cite{RubinsteinW15} to prove a ``Core-Tail Decomposition'' (en route to a constant-factor approximation). More related to our work, distribution-nudges have been used in works such as~\cite{DaskalakisW12} to discretize an input distribution, and then find an optimal menu for that discretized distribution (indeed, we use Proposition~\ref{prop:advancedNudge} for a simliar purpose). Algorithms based on distribution-nudges tend to be constructive, since one can explicitly discretize an input, and then perform operations on it.

Allocation-nudges have been used in works such as~\cite{BabaioffGN17} to argue that menus of a certain format with certain revenue guarantees exist. For example, Lemma~2.5 in~\cite{BabaioffGN17} rounds (almost) all allocation probabilities down to the nearest multiple of $\varepsilon/n$, concluding that a near-optimal menu exists where each option has all-but-one allocation probability equal to an integral multiple of $\varepsilon/n$ (which in turn allows them to eventually conclude that a near-optimal menu exists with a bounded number of distinct allocations). Arguments which apply allocation-nudges to an optimal mechanism are typically not constructive (since one cannot explicitly find the optimal mechanism and nudge it without a separate argument). We will also use allocation-nudge arguments in the proof of our reduction (Theorem~\ref{thm:reduction}).

Finally, we briefly discuss the relationship between symmetric menu complexity and these arguments. One use of a distribution-nudge (e.g. in Section~\ref{sec:symmetries}) would be to start from an asymmetric distribution $D$, nudge it to a symmetric distribution $D'$, exploit symmetries in $D'$ to claim that optimal mechanisms for $D'$ have low symmetric menu complexity, and then use Proposition~\ref{prop:advancedNudge} to claim symmetric menu complexity bounds on nearly-optimal mechanisms for $D$. 

For allocation-nudges, one might initially wonder whether symmetric menu complexity is so permissive a measure that (for instance) any menu which only uses allocation probabilities and prices in integral multiples of $\varepsilon$ happens to have low symmetric menu complexity, perhaps just because every such menu can be partitioned into a small number of heavily-symmetric parts. We quickly confirm that that this is not the case via a simple counting argument. Indeed, observe that there are \emph{doubly-exponentially-many} different menus which only use allocation probabilities in integral multiples of $\varepsilon$ (because there are exponentially-many different such allocations, and a menu can have any subset of these at a finite price). But recall that any menu $M$ with $\wsmc(M)$ subexponential in $n$ can be written as a list of subexponentially-many (allocation, price, permutation group) triples. If further each allocation/price is an integral multiple of $\varepsilon$, and each permutation group can be written using $n\ln(n)$ bits by Jerrum's filter~\cite{Jerrum82}, then the entire menu can be described with subexponentially-many bits. Therefore, the total number of menus which have subexponential weak symmetric menu complexity is sub-doubly-exponential, implying that we cannot in fact simply get lucky and write any discretized menu with low symmertic menu complexity. This confirms that there really is some structure not only to menus with low strong symmetric menu complexity, but also those with low weak symmetric menu complexity.

\section{Outline and Proofs: Reduction from Unbounded to Bounded}\label{sec:reductionproofs}

This section contains a complete proof of Theorem~\ref{thm:reduction}. The reader familiar with~\cite{BabaioffGN17} will find the first half of this outline similar to that of their section~2, but we remind the reader of the outline. Recall still that there are key quantitative differences (which we note at the beginning of each step), and an extension to subadditive buyers (which we will not repeatedly note). The second half of the outline deviates significantly from anything in~\cite{BabaioffGN17}.

Section~\ref{sec:nudge} contains formal statements of ``nudge-and round'' arguments, which bound the loss in revenue from a particular menu $M$ between two similar distributions, and will be used repeatedly. As mentioned previously, our proofs really make use of these careful statements, rather than simpler statements from prior work. Section~\ref{sec:stepone} asserts that we may restrict attention to the subset of the support of $D$ where the buyer values \emph{at most one} item above $H=\SREV(D)/\varepsilon$ without losing more than an $\varepsilon$ fraction of the revenue. Section~\ref{sec:steptwo} asserts further that we may make any menu $E$-exclusive (for $E \geq \max\{H/\varepsilon,\BREV(D)\}/\varepsilon$) without losing more than an $\varepsilon$ fraction of the revenue. Section~\ref{sec:stepthree} concludes that we may greatly simplify the singleton options priced above $E$ without losing more than an $\varepsilon$ fraction of the revenue, reducing them to only two expensive options per item (so $2n$ options total priced $\geq E$). The conclusion of section~\ref{sec:stepthree} is Proposition~\ref{prop:main3}, which is a slightly more precise version of Proposition~\ref{prop:main4}. 

Section~\ref{sec:stepfour} is where we deviate significantly from anything in~\cite{BabaioffGN17}. We argue that while one absolutely cannot treat the search for a cheap menu and $2n$ expensive options separately, the \emph{only} property coupling the searches together is the supremum probability with which the cheap menu offers item $i$ (for all $i$). The conclusion of Section~\ref{sec:stepfour} is two propositions (Propositions~\ref{prop:reduce1} and~\ref{prop:reduction2}) which together are a slightly more precise version of Proposition~\ref{prop:reduction}. Again, we emphasize that the qualitative contribution of section~\ref{sec:stepfour} is that future work (and the remainder of the present work) can proceed to study bounded distributions, without worrying about how the resulting menu will interact with additional expensive options added.

We also briefly note that in our proofs, at some points we need for a parameter $H, E, T$, etc. to be larger than $\SREV(D)$, and sometimes to be larger than $\BREV(D)$. As $\REV(D) = \Theta(\max\{\SREV(D),\BREV(D)\})$~\cite{RubinsteinW15}, we would not gain more than a constant factor by tightening these parameters from (e.g.) $E \geq \REV(D)/\varepsilon^3$ to $E \geq \max\{\SREV(D),\BREV(D)\}/\varepsilon^3$, so we write the bounds in terms of $\REV(D)$ instead of $\max\{\SREV(D),\BREV(D)\}$ for ease of exposition.

\subsection{Step 0: Further Notation}
Throughout the subsequent sections, let:
\begin{itemize}
\item $H \in \mathbb{R}$ be a cutoff (which will be instantiated in each statement as $\SREV(D)/\varepsilon$). 
\item $B_i$ denote the event that $v(\{i\}) > H$, and $p_i = \Pr[B_i]$.
\item $B_{ij}$ denote the event that $v(\{i\}) > H$ and also $v(\{j\}) > H$ (i.e. $B_{ij} = B_i \cap B_j$). 
\item $B$ denote the event that $v(\{i\}) > H$ for \emph{at most one} item $i$ (i.e. $B=\overline{\bigcup_{i,j} B_{ij}}$).
\item $B'_A$ denote the event that $v(\{i\}) > H$ for all $i \in A$, and $v(\{i\}) \leq H$ for all $i \notin A$.
\end{itemize}

\subsection{Step 1: Ignoring Multiple High Values}\label{sec:stepone}

The first step in the reduction is to find a sufficiently large value $H$ such that we can safely ignore revenue contributed by any valuation function $v(\cdot)$ such that there exist \emph{at least two} distinct items $i,j$ for which $v(\{i\}) > H$ and also $v(\{j\}) > H$. Proposition~\ref{prop:main1} is the main contribution of this section, which parallels~\cite[Section~2.3]{BabaioffGN17}. The quantitative difference to~\cite{BabaioffGN17} is that our $H = \frac{\SREV(D)}{\varepsilon}$, whereas they instead take $H=\Omega(\frac{n^2\REV(D)}{\varepsilon})$. 

\begin{proposition}\label{prop:main1} Let $H\geq \frac{\SREV(D)}{\varepsilon}$. Then:
$$\sum_{i=1}^n \sum_{j > i} \REV(D\cdot \mathbb{I}(B_{ij})) \leq \frac{36\varepsilon}{1-2\varepsilon} \SREV(D) + \frac{6\varepsilon^2}{(1-\varepsilon)^2}\REV (D).$$
Therefore, $\REV(D\cdot \mathbb{I}(B)) \geq (1-O(\varepsilon)) \REV(D)$, and for any $M$, we have $\REV_M(D\cdot \mathbb{I}(\bar{B})) \leq O(\varepsilon)\REV(D)$.
\end{proposition}

Proposition~\ref{prop:main1} asserts that we can ignore the revenue contributed from cases where the buyer values two or more items each above $H$, because it contributes at most an $O(\varepsilon)$ fraction of the optimal revenue. We'll conclude the proof of Proposition~\ref{prop:main1} at the end of this subsection, and first build up the necessary lemmas. We begin with Claim~\ref{claim:sumpi} below, which asserts that we are unlikely to see many items with values exceeding $H$.

\begin{claim}\label{claim:sumpi} For any $t$, we have $\sum_i \Pr[v(\{i\}) \geq t] \leq \frac{\SREV(D)/t}{1-\SREV(D)/t}$. In particular, if $H \geq \frac{\SREV}{\varepsilon}$ then $\sum_i \Pr[v(\{i\}) > H] \leq \varepsilon/(1-\varepsilon)$.
\end{claim}	

\begin{proof}
Consider the mechanism which offers the buyer to purchase any singleton set at price $t$. Then the probability that the buyer chooses to purchase an item must be at most $\SREV(D)/t$ (otherwise this contradicts the definition of $\SREV(D)$). Therefore, for each item $i$, the probability that item $i$ is \emph{purchased} in this mechanism is at least $\Pr[v(\{i\})\geq t] \cdot (1-\SREV(D)/t)$ (because whenever the buyer is willing to purchase item $i$, they are uninterested in any of the others with probability at least $(1-\SREV(D)/t)$, since all the values $v(\{j\})$ are independent). Therefore, we can also write that the revenue from this mechanism is at least $t \sum_i \Pr[v(\{i\}) \geq t] \cdot (1-\SREV(D)/t)$. As this must be at most $\SREV(D)$, we get our claim.
\end{proof}

We need just one more lemma, from~\cite{RubinsteinW15}, in order to wrap up:
\begin{lemma}[\cite{RubinsteinW15}]\label{lem:RW15}
Let $D$ be subadditive over independent items, and let $S, \bar{S}$ partition $[n]$. Then:
$$\REV(D) \leq 6\REV(D_S)+6\REV(D_{\bar{S}}).$$
\end{lemma} 

\begin{proof}[Proof of Proposition~\ref{prop:main1}]
Consider any two items $i,j$, and consider the distribution $D|B_{ij}$. Observe that it is still subadditive over independent items (as conditioning on $B_{ij}$ simply reduces the domain of $D_i$ and $D_j$, but the parameters for all items are otherwise drawn independently conditioned on this). Therefore we can apply Lemma~\ref{lem:RW15} to get:
\begin{align*}
\REV(D|B_{ij}) &\leq 6 \REV (D_{\{i,j\}}|B_{ij}) + 6 \REV(D_{[n]\setminus \{i,j\}}|B_{ij})\\
&\leq 36\REV(D_i|B_{ij}) + 36 \REV(D_j|B_{ij}) + 6 \REV(D)\\
&\leq 36\REV(D_i|B_i) + 36 \REV(D_j|B_j) + 6 \REV(D).
\end{align*}
Here the first line follows from a direct application of Lemma~\ref{lem:RW15}. The second line follows again from an application of Lemma~\ref{lem:RW15}, and then the observation that conditioning on $B_{ij}$ does not affect the distribution of values for items in $[n]\setminus \{i,j\}$ at all, and therefore $\REV(D_{[n]\setminus \{i,j\}}|B_{ij}) = \REV(D_{[n]\setminus \{i,j\}}) \leq \REV(D)$. The final line follows by observing similarly that conditioning on $B_j$ doesn't affect the distribution of $v(\{i\})$, and so $D_i|B_{ij} = D_i|B_i$ (and $D_j|B_{ij} = D_j|B_j$). 

Now, let us multiply both sides by $p_ip_j$ to get (observing that $p_i \REV(D_i|B_i)$ is exactly $\REV(D_i \cdot \mathbb{I}(B_i))$:
$$\REV(D \cdot \mathbb{I}(B_{ij})) = p_ip_j \REV(D|B_{ij}) \leq 36p_j \REV(D_i\cdot \mathbb{I}(B_i))+36 p_i \REV(D_j\cdot \mathbb{I}(B_j)) + 6 p_i p_j \REV(D).$$

Finally, we can sum over all pairs to get:
\begin{align*}
\sum_{i=1}^n \sum_{j > i}  \REV (D\cdot \mathbb{I}(B_{ij})) &\leq 36\sum_i \left(\sum_{j \neq i} p_j\right) \REV(D_i\cdot \mathbb{I}(B_i)) + 6 \REV(D) \sum_{i=1}^n \sum_{j > i} p_i p_j\\
&\leq \frac{36\varepsilon}{1-\varepsilon} \sum_i \REV(D_i\cdot \mathbb{I}(B_i)) + 6\left(\sum_i p_i\right)^2 \REV(D)\\
&\leq \frac{36\varepsilon}{1-2\varepsilon} \SREV(D) + \frac{6\varepsilon^2}{(1-\varepsilon)^2}\REV (D).
\end{align*}
Here the first inequality just sums the previous inequality over all pairs, and reorganizes terms. The second inequality makes use of Claim~\ref{claim:sumpi} to upper bound $\sum_{j \neq i} p_j$. The third inequality again uses Claim~\ref{claim:sumpi} to upper bound $\sum_i p_i$, and also to observe that $\SREV(D) \geq (1-2\varepsilon)\sum_i \REV(D_i\cdot \mathbb{I}(B_i))$. To see this final claim, observe that each $D_i\cdot \mathbb{I}(B_i)$ is a single-dimensional distribution. So the optimal mechanism is a posted-price $r_i$, and that price is certainly $\geq H$ (since the entire support of $D_i\cdot \mathbb{I}(B_i)$ is $\geq H$). So consider setting each of these prices. Observe that $\REV(D_i\cdot \mathbb{I}(B_i))$ is exactly $r_i \cdot \Pr[v(\{i\}) \geq r_i]$. At the same time, item $i$ will be purchased at least as often as both $v(\{i\}) \geq r_i$ \emph{and} $v(\{j\} < r_j)$ for all $j \neq i$. As each $r_j \geq H$, the probability of this latter event is at least $1-\frac{\varepsilon}{1-\varepsilon}$ by Claim~\ref{claim:sumpi}. Therefore, we get that these prices witness $\SREV(D) \geq \frac{1-2\varepsilon}{1-\varepsilon}\sum_i \REV(D_i\cdot \mathbb{I}(B_i))$. 

To derive the ``Therefore,'' portion of the proposition, we use Sub-Domain Stitching (Lemma~\ref{lem:sds}). Recall that $B'_A$ is the event that exactly the items $i\in A$ have $v(\{i\}) > H$. Then for any mechanism $M$, we have:
$$\REV_M(D) = \sum_{A \subseteq [n]} \REV_M(D \cdot \mathbb{I}(B'_A)).$$
Also, observe that for any $i,j$, we have:
$$\REV_M(D \cdot \mathbb{I}(B_{ij})) = \sum_{A \supseteq \{i,j\}} \REV_M(D \cdot \mathbb{I}(B'_A)).$$
And after summing over all $i,j$, we get that:
$$\sum_{A, |A| \geq 2}\REV_M(D \cdot \mathbb{I}(B'_A)) \quad \leq \quad \sum_{i=1}^n \sum_{j > i} \REV(D \cdot \mathbb{I}(B_{ij})) \quad \leq \quad \frac{36\varepsilon}{1-2\varepsilon} \SREV(D) + \frac{6\varepsilon^2}{(1-\varepsilon)^2} \REV(D).$$
Therefore, as $B$ is the event that $v(\{i\}) > H$ for \emph{at most one} $i$, $\cup_{A, |A| \geq 2} B'_A = \bar{B}$, we conclude that:
$$\REV_M(D\cdot \mathbb{I}(\bar{B})) \quad \leq \quad  \frac{36\varepsilon}{1-2\varepsilon} \SREV(D) + \frac{6\varepsilon^2}{(1-\varepsilon)^2} \REV(D) \quad = \quad O(\varepsilon) \cdot \REV(D).$$
Hence, when $M$ is the optimal mechanism for $D$, we get:
\[ \REV_M(D \cdot \mathbb{I}(B)) \quad \geq \quad \left(1-\frac{42\varepsilon}{(1-\varepsilon)^2}\right)\REV(D) \quad = \quad (1-O(\varepsilon))\REV(D). \qedhere \]
\end{proof}


\subsection{Step 2: Simplifying Expensive Options}\label{sec:steptwo}
In this section, we build on Step 1 to argue that there exists a $(1-\varepsilon)$-optimal menu such that for all options priced above $E\geq \frac{H}{\varepsilon^2}$, they offer only a single item with non-zero probability. This is again parallels~\cite[Section~2.4]{BabaioffGN17}, except that we take $E\geq\frac{H}{\varepsilon^2}$ whereas they take $E:= \frac{nH}{\varepsilon}$. 

\begin{proposition}\label{prop:main2} Let $E \geq \max\{H/\varepsilon^2,\BREV(D)/\varepsilon\}$, and $H \geq \SREV(D)/\varepsilon$. Then for all $D$ that are subadditive over independent items, there exists a mechanism $M$ such that:
\begin{itemize}
\item $\REV_M(D) \geq \REV_M(D \cdot \mathbb{I}(B)) \geq (1-O(\varepsilon))\REV(D)$.
\item For all (randomized allocation, price) pairs $(S,p) \in M$, either $p \leq E$, or there exists at most one $i$ such that $\Pr[i \in S] > 0$. 
\end{itemize}
\end{proposition}

The main idea in the proposition, just as in~\cite{BabaioffGN17}, is an application of nudge-and-round.

\begin{proposition}\label{prop:exclusive} For any menu $M$, we gave $\REV_{M|_E}(D\cdot \mathbb{I}(B)) \geq (1-\varepsilon)\REV_M(D\cdot \mathbb{I}(B)) - O(\varepsilon)\cdot\REV(D)$.

\end{proposition}

\begin{proof}
Observe that the menus $M, M|_E$ satisfy the hypotheses of Proposition~\ref{prop:carefulNudge}, where for any $v$ with $p(v) \leq E$ we define $S'(v):= S(v)$, and for any $v$ with $p(v) >E$ we define $S'(v):=X_i(S(v))$, where {$i:=\arg\max_j\{v(\{j\})\}$} (i.e. the option which keeps only $v$'s favorite item). 

To proceed, let us now use Proposition~\ref{prop:carefulNudge} for mechanisms $M, M|_E$ and distribution $D \cdot \mathbb{I}(B)$. The only task is to bound $\mathbb{E}_{v \leftarrow D \cdot \mathbb{I}(B)}[v(S(v)) - v(S'(v))]$. There are $n+1$ disjoint events making up $B$: $B'_\emptyset$, and $B'_{\{i\}}$ for all $i \in [n]$. Observe that the distribution $D|B'_A$ is still subadditive over independent items, for all $A$ (this is simply because each conditioning just restricts the support of $D_i$, but conditioned on this it is still a product distribution). 

The first observation we make is the following: if $v$ purchased some option from $M$ which had price $\leq E$, then we defined $S'(v)=S(v)$, and so the only $v$ for which $v(S(v)) - v(S'(v)) >0$ are those for which (at minimum) $v([n]) > E$ (as otherwise they are certainly not choosing to purchase an option priced $> E$). 

Moreover, we observe that if $v$ purchased an option with price $> E$ in $M$, and $i$ denotes their favorite item, then then because all $v(\cdot)$ are subadditive, $v(S(v)) - v(X_i(S(v)))\leq v([n] \setminus \{i\})$. 

Now, we wish to figure out how much expected value can possibly be lost from both cases. For both arguments, we'll make use of the following concentration inequality due to Schechtman~\cite{Schechtman99}, and rephrased by~\cite{RubinsteinW15}.

\begin{proposition}[\cite{Schechtman99}]\label{prop:schechtman} Let $D$ be subadditive over independent items, and let $v(\{i\}) \leq t$ with probability one, for all $i$. Then for any $q,a,k > 0$:
$$\Pr[v([n]) \geq (q+1)a + kt] \leq \Pr[v([n]) \leq a]^{-q} q^{-k}.$$
In particular, if $a$ denotes the median of $v([n])$ and $q = 2$,\footnote{That is, $a$ is such that $\Pr[v([n]) \geq a] \geq 1/2$ and also $\Pr[v([n]) \leq a] \geq 1/2$.} then for all $k$:
$$\Pr[v([n]) \geq 3a+kt] \leq 4\cdot 2^{-k}.$$
\end{proposition}

This proposition suffices to bound both the terms that we care about.

\begin{corollary}\label{cor:Schechtman} Let $D$ be subadditive over independent items, and let $v(\{i\}) \leq t$ with probability one, for all $i$. Let also $a$ denote the median of $v([n])$. Then for any $T \geq 3a$:
$$\mathbb{E}[v([n]) \cdot \mathbb{I}(v([n]) \geq T)] \leq \frac{4t}{\ln(2)} \cdot 2^{-\frac{T-3a}{t}}.$$
\end{corollary}
\begin{proof} Observe that
\begin{align*}
\mathbb{E}[v([n]) \cdot \mathbb{I}(v([n]) \geq T)] &= \int_T^\infty \Pr[v([n]) \geq x]dx\\
&\leq \int_T^\infty 4\cdot 2^{-\frac{x-3a}{t}}dx \quad= \quad \frac{-4t}{\ln(2)} \cdot 2^{-\frac{x-3a}{t}}|_T^\infty \quad = \quad \frac{4t}{\ln(2)} \cdot 2^{-\frac{T-3a}{t}}. \qedhere
\end{align*}
\end{proof}

Plugging in above with $t=H$ and $T = E$, we get that $\mathbb{E}_{v \leftarrow D|B'_\emptyset}[v([n])\cdot \mathbb{I}(v([n]) \geq E)] \leq 4H2^{-\frac{E - 6\BREV(D)}{H}} \leq O(\exp(-1/\varepsilon)\cdot \SREV(D))$. The calculations follow by observing that $\BREV(D)$ is at least half the median of $v([n])$ (when $v$ drawn from $D|B'_\emptyset$, or $D$ --- because setting the median as a price on $[n]$ sells with probability at least $1/2$ by definition), and therefore we can upper bound $3$ times the median with $6\BREV(D)$ (and also because $E \geq \BREV(D)/\varepsilon$). The above reasoning therefore concludes the following claim:

\begin{claim}\label{claim:allsmall} $\mathbb{E}_{v \leftarrow D|B'_\emptyset}[v([n])\cdot \mathbb{I}(v([n])> E)] \leq O(\exp(-1/\varepsilon))\cdot \SREV(D)$.
\end{claim}

Next, we need to upper bound the loss from $B'_{\{i\}}$ for all $i$. Here, we'll cover two cases, breaking $B'_{\{i\}}$ further into two events depending on whether $v(\{i\})$ is unusually large or not. Specifically, let $L_i$ denote the event that $v(\{i\}) > E/2$ (and $v(\{j\}) \leq H$ for all $j \neq i$), and $A_i$ denote the event that $v(\{i\}) \in (H,E/2]$ (and $v(\{j\}) \leq H$ for all $j \neq i$). 

$\mathbb{E}_{v \leftarrow D|A_i}[v([n] \setminus \{i\})]$ can be handled almost identically to the previous case. Specifically, observe that in order for the buyer to possibly purchase an option priced above $E$, conditioned on $A_i$, their value for $[n]\setminus\{i\}$ must exceed $E/2$. So we care about $\mathbb{E}_{v\leftarrow D|A_i}[v([n]\setminus \{i\})\cdot \mathbb{I}(v([n]\setminus \{i\})\geq E/2)]$. As $D|A_i$ is subadditive over independent items, we can again just directly apply Corollary~\ref{cor:Schechtman} (with $t = H, T = E/2$) to conclude that $\mathbb{E}_{v\leftarrow D|A_i}[v([n]\setminus \{i\})\cdot \mathbb{I}(v([n]\setminus \{i\})\geq E/2)] = O(\exp(-1/\varepsilon))\cdot \SREV(D)$. That is, we conclude the following claim:

\begin{claim}\label{claim:onemedium}
$\mathbb{E}_{v \leftarrow D|A_i}[v([n]\setminus \{i\}) \cdot \mathbb{I}(v([n]\setminus\{i\}) \geq E/2)] = O(\exp(-1/\varepsilon))\cdot \SREV(D)$.
\end{claim}

For $D|L_i$, it is now wholly possible that the buyer chooses to purchase an option priced $> E$ despite having low values for items in $[n]\setminus \{i\}$. In this case, we just bound the expectation directly using similar ideas. For this case, the bound we want has already been explicitly stated in~\cite{RubinsteinW15} (which follows from similar calculations to Corollary~\ref{cor:Schechtman}).

\begin{proposition}[\cite{RubinsteinW15}]\label{prop:RW} Let $D$ be subadditive over independent  items, and let $v(\{i\})$ be supported on $[0,t]$ for all $v \in \text{support}(D)$. Then $\mathbb{E}_{v\leftarrow D}[v([n])] \leq 6\BREV(D) + 4t/\ln(2).$
\end{proposition}

As $D |L_i$ is subadditive over independent items and each $v(\{j\}) \leq H$ for all $j \neq i$, we immediately get that $\mathbb{E}_{v \leftarrow D|L_i}[v([n]\setminus\{i\})] \leq 6\BREV(D^{[n]\setminus \{i\}}|L_i) + 4H/\ln(2)$. As the conditioning on $L_i$ only makes $v([n]\setminus\{i\})$ smaller, $\BREV(D^{[n]\setminus \{i\}}|L_i) \leq \BREV(D)$, and so we get:

\begin{claim}\label{claim:onebig}$\mathbb{E}_{v \leftarrow D|L_i}[v([n]\setminus\{i\})] \leq 6\BREV(D) + 4H/\ln(2)$.
\end{claim}

Finally, we are ready to wrap up:
\begin{align*}
\REV_{M|_E}(D \cdot \mathbb{I}(B)) &\geq (1-\varepsilon)\REV_M(D\cdot \mathbb{I}(B)) - \frac{1}{\varepsilon}\Big(\mathbb{E}_{v \leftarrow D \cdot \mathbb{I}(B'_\emptyset)}[v([n])] - \sum_i \mathbb{E}_{v \leftarrow D \cdot \mathbb{I}(B'_{\{i\}})}[v([n \setminus \{i\}])]\Big)\\
& \geq \REV_M(D\cdot \mathbb{I}(B))-O(\varepsilon)\cdot \REV(D) - O(\exp(-1/\varepsilon))\cdot \SREV(D)\\
&\quad- \frac{1}{\varepsilon}\Big( \sum_i \left(\mathbb{E}_{v \leftarrow D \cdot \mathbb{I}(A_i)}[v([n \setminus\{i\}])] + \mathbb{E}_{v \leftarrow D \cdot \mathbb{I}(L_i)}[v([n \setminus\{i\}])]\right)\Big)\\
&\geq \REV_M(D) -O(\varepsilon)\cdot\REV(D) \\&\quad-\frac{1}{\varepsilon}\Big( \ \sum_i \left(\Pr[A_i] \cdot \mathbb{E}_{v \leftarrow D|A_i}[v([n \setminus \{i\}])] +\Pr[L_i] \cdot \mathbb{E}_{v \leftarrow D|L_i}[v([n \setminus \{i\}])]\right)\Big)\\
&\geq \REV_M(D) -O(\varepsilon)\cdot\REV(D)  \\
&\quad -\frac{1}{\varepsilon}\Big( \sum_i \left(\Pr[A_i] \cdot O(\exp(-1/\varepsilon))\cdot \SREV(D) +\Pr[L_i] \cdot (6\BREV(D) + 4H/\ln(2) )\right)\Big)\\
& \geq \REV_M(D) -O(\varepsilon)\cdot\REV(D)  - (6\BREV(D) + 4H/\ln(2))\cdot \sum_i \frac{\Pr[L_i]}{\varepsilon}\\
& \geq\REV_M(D) -O(\varepsilon)\cdot\REV(D)  - (6\BREV(D) + 4H/\ln(2))\cdot \frac{\SREV/E}{\varepsilon(1-\SREV/E)}\\
&\geq \REV_M(D) -O(\varepsilon)\cdot\REV(D) .
\end{align*}
The first line follows just by a direct application of Proposition~\ref{prop:carefulNudge}. The second line follows from Claim~\ref{claim:allsmall}, and by partitioning $B'_{\{i\}}$ into $A_i \sqcup L_i$.

The third line follows by grouping together terms that are $O(\varepsilon) \cdot \REV(D)$, and rewriting the draw of $v$ from $D$ restricted to $A_i$ (resp. $L_i$) with drawing $v$ from $D$ conditioned on $A_i$ (resp. $L_i$) and multiplying by the probability of $A_i$ (resp. $L_i$). The fourth line just upper bounds the two remaining expectations using Claims~\ref{claim:onemedium} and~\ref{claim:onebig}.

The fifth line observes that $\sum_i \Pr[A_i] \leq 1$ as they are disjoint events, and again groups terms that are $O(\varepsilon) \cdot \REV(D)$. The sixth line uses the fact that $\sum_i \Pr[L_i] \leq \frac{\SREV/E}{1-\SREV(D)/E}$, which follows from Claim~\ref{claim:sumpi}.

The final line then follows by recalling that $E\geq H/\varepsilon^2\geq \SREV(D)/\varepsilon^3$. So the fraction $\frac{\SREV(D)/E}{\varepsilon(1-\SREV(D)/E)}$ is both $O(\varepsilon^2)$ and also $O(\varepsilon \SREV(D)/H)$, so the entire remaining term is $O(\varepsilon)\cdot \REV(D)$.
\end{proof}

\begin{proof}[Proof of Proposition~\ref{prop:main2}]
Simply observe that $M|_E$ clearly satisfies the second bullet of Proposition~\ref{prop:main2} for any $M$. Moreover, by Proposition~\ref{prop:exclusive}, if we take $M$ to be the revenue-optimal (or arbitrarily close to optimal) mechanism for $D$, then we get:
\[ \REV_{M|_E}(D) ~\geq~ \REV_{M|_E}(D \cdot \mathbb{I}(B)) ~\geq~ \REV_M(D\cdot \mathbb{I}(B)) - O(\varepsilon)\cdot \REV(D) ~\geq~ (1-O(\varepsilon)) \REV(D). \qedhere \]
\end{proof}

\subsection{Step 3: Trimming Expensive Options}\label{sec:stepthree}
Here, we show that any mechanism of the form promised by Proposition~\ref{prop:main2} can be further simplified to have only $2n$ options priced above $E$. The approach we take is again similar to~\cite[Section~2.5]{BabaioffGN17}, with the main substantive difference that we require only $E \geq H/\varepsilon^2$ (versus $E \geq nH$). 
\begin{proposition}\label{prop:main3}
 Let $E \geq \max\{H/\varepsilon^2,\BREV(D)/\varepsilon\}$, and $H \geq \SREV(D)/\varepsilon$. Then for all $D$ that are subadditive over independent items, there exists a mechanism $M$ such that:
\begin{itemize}
\item $\REV_M(D) ~~\geq~~ \REV_M\big(D \cdot \mathbb{I}(B)\big) ~~\geq~~ (1-O(\varepsilon))\REV(D)$.
\item For all (randomized allocation, price) pairs $(S,p) \in M$, either $p \leq E$, or there exists at most one $i$ such that $\Pr[i \in S] > 0$.
\item For each item $i$, there exist at most two distinct (randomized allocation, price) pairs $(S,p) \in M$ such that $p > E$ and $\Pr[i \in S] > 0$. Moreover, one of these options has $\Pr[i \in S] =1$. 
\end{itemize}
\end{proposition}
 
We will again conclude the proof of Proposition~\ref{prop:main3} at the end of the section, and build up lemmas first. We will provide a complete proof below, and identify which lemmas have appeared in prior work. 

To begin, let $M$ denote the mechanism promised by Proposition~\ref{prop:main2}, and let $M_i$ denote the subset of the $M$ containing all $(S,p)$ for which $p > E$ and $\Pr[i \in S] > 0$. Let also $M^*:= M\setminus (\cup_i M_i)$ denote the options on $M$ priced $\leq E$. We first argue that it is without loss to remove from $M_i$ any options which are never purchased. Lemma~\ref{lem:RZ} below characterizing options in $M_i$ which are never purchased comes from~\cite{RileyZ83}.

\begin{lemma}[\cite{RileyZ83}]\label{lem:RZ} Say that $(S,p) \in M_i$ dominates a distinct $(T, q) \in M_i$ if:
\begin{itemize}
\item $\Pr[i \in S] \geq \Pr[i \in T]$.
\item $p/\Pr[i \in S] \leq q/\Pr[i \in T]$. 
\end{itemize}

Then if $M'_i$ denotes the subset of $M_i$ which is undominated by any other element of $M_i$, and we update $M':=M^* \cup (\cup_i M'_i)$, we have $\REV_M(D) = \REV_{M'}(D)$.
\end{lemma}

For the remainder of this section, we will therefore assume w.l.o.g. that no $M_i$ contains any dominated option (which implies that $M_i$ contains no distinct options $(S,p) \neq (T,q)$ with $\Pr[i \in S] = \Pr[i \in T]$, and also that for any $(S,p),(T,q) \in M_i$ that $\Pr[i \in S] > \Pr[i \in T] \Rightarrow p > q$). Now define $q_i:= \inf_{(S,p) \in M_i}\{\Pr[i \in S]\}$ and $p_i := \inf_{(S,p) \in M_i}\{p\}$. Like~\cite{BabaioffGN17}, we  think of each of the options in $M_i$ as allocating at least a $q_i$ fraction of item $i$ for $p_i$, and then additionally allocating a $\Pr[i \in S] - q_i$ fraction of item $i$ for $p - p_i$. 

\begin{defn} We will denote the \emph{expensive interval} $C_i$ for item $i$ as the range of values $c$ such that a buyer with value $c$ for item $i$ and $0$ for all other items would choose to purchase an option in $M_i$ from menu $M$. Observe that certainly $C_i \subseteq [E,\infty)$, and in fact $C_i \subseteq [p_i/q_i,\infty)$, but also perhaps $C_i$ is an open interval, or perhaps the lower endpoint is much larger. Denote by $c_i$ the lower endpoint of $C_i$ (noting that maybe $c_i \in C_i$, or maybe not). 
\end{defn} 

Observation~\ref{obs:oneway} observes an upper bound on the revenue achievable by $M$, by breaking down its revenue into that coming from $M^*$, and that coming from the $M_i$s. This idea is explicit in~\cite{BabaioffGN17}.

\begin{observation}\label{obs:oneway} $\REV_{M}(D\cdot\mathbb{I}(B)) \leq \REV_{M^*}(D \cdot \mathbb{I}(B \wedge \forall i,\ v(\{i\})\notin C_i)) + \sum_i \REV_{M_i}(D_i \cdot \mathbb{I}(B \wedge v(\{i\}) \in C_i))$.
\end{observation}
\begin{proof}
First, consider any $v$ such that $v(\{i\}) \notin C_i$ for all $i$. Then certainly such $v$ will purchase an option on $M^*$, by definition of $C_i$. So the revenue that such $v$ contributes to $M$ is exactly the same as their contribution to $M^*$. Next, consider any $v$ such that $v(\{i\}) \in C_i$ for some $i$ (note that coinciding with event $B$ means there is at most one such $i$). Say that this buyer purchases an option from $M_i$ in $M$. Then because all options in $M_i$ offer only item $i$ with non-zero probability, if we zero out their values for all other options, they would still purchase exactly the same option from $M_i$ when presented with only options in $M_i$. Therefore, the contribution of such a $v(\cdot)$ to $M$ is exactly their contribution to $M_i$, and is counted correctly (in fact, it is not even overcounted, since there can be at most one such $i$ with $v(\{i\}) \in C_i$ when event $B$ occurs). 

When $v$ is such that $v(\{i\}) \in C_i$ for some $i$, but $v$ chooses to purchase an option in $M^*$ in $M$, they pay at most $E$. As $v(\{i\}) \in C_i$, they will certainly purchase a non-trivial option from $M_i$ when presented only with $M_i$, and will pay at least $E$. Therefore, the contribution of such $v$ to $M$ is also (over-)counted by the right-hand side.

We have shown that for all $v$ in the support of $D\cdot \mathbb{I}(B)$, the contribution of purchases from $v$ to the LHS is always at most the contribution to the RHS.
\end{proof}

Observation~\ref{obs:oneway} asserts that the revenue achieved by menu $M$ on $D$ can be upper bounded cleanly by a contribution from $M^*$ in isolation, and contributions from $M_i$ \emph{on single-dimensional distributions} in isolation. We next argue that certain mechanisms for $D_i\cdot \mathbb{I}(v(\{i\}) \in C_i)$ can be included with $M^*$ without losing much of their revenue. Lemma~\ref{lem:otherway} is also explicit in~\cite{BabaioffGN17}, although quantitatively weaker.

\begin{lemma}\label{lem:otherway} Let $\{M'_i\}_{i \in [n]}$ be a set of menus each offering only item $i$ such that:
\begin{itemize}
\item Each $M'_i$ contains only undominated options.
\item At least one option $(S,p) \in M'_i$ has $\Pr[i \in S] \geq q_i$ and $p = \Pr[i \in S] \cdot c_i$.
\item All options $(S,p) \in M'_i$ have $p \geq \Pr[i \in S] \cdot c_i$.
\end{itemize} 
Define $M':= M^* \cup (\cup_i M'_i)$, and define $M''$ to be $M'$ after multiplying all prices by $(1-\varepsilon)$. Then:
\begin{align*} 
 \REV_{M''}\big(D\cdot\mathbb{I}( B)\big) &\geq (1-\varepsilon)\cdot \REV_{M^*}\big(D\cdot \mathbb{I}(B\wedge \forall i, \ v(\{i\})\notin C_i)\big)  \\*
 & \qquad + (1-\varepsilon)^2\sum_i \REV_{M'_i} \big(D_i\cdot \mathbb{I}(B \wedge v(\{i\})\in C_i)\big) ~~-~~ O(\varepsilon)\cdot \REV(D). 
\end{align*}
\end{lemma}

\begin{proof}
We will do this with parallel applications of Proposition~\ref{prop:advancedNudge}, by breaking down event $B$ into events $B\wedge (\forall i, \ v(\{i\})\notin C_i)$ and $B \wedge( v(\{i\}) \in C_i)$ (separately for each $i$).

First, consider $v$ drawn such that $v(\{i\}) \notin C_i$ for all $i$. Then on the left-hand side (when purchasing from $M''$), such a $v$ either purchases its favorite option from $M^*$ (after the discount), or some option from $M'_i$ (after the discount). 

Let $(S,p)$ denote $v$'s favorite option from $M^*$ (before the discounts). Then such a $v$ will contribute $(1-\varepsilon)p$ to the RHS. After the discounts, their favorite option from $M^*$ might change, but not to one priced $\leq (1-\varepsilon)p$ (because higher-priced options are discounted more heavily). Also, they may choose to purchase from $M''$ an option from one of the discounted $M'_i$s, but not paying $\leq (1-\varepsilon)E$, which is certainly $\geq (1-\varepsilon)p$, as every option on $M^*$ has price at most $E$. So a $v$, such that $v(\{i\}) \notin C_i$, contributes at least $(1-\varepsilon)p$ to the LHS, and at most $(1-\varepsilon)p$ to the RHS. This allows us to conclude:

\begin{claim}\label{claim:M*}
$\REV_{M''}\big(D\cdot\mathbb{I}( B\wedge \forall i, \ v(\{i\})\notin C_i)\big) \geq (1-\varepsilon)\cdot \REV_{M^*}\big(D\cdot \mathbb{I}(B\wedge \forall i, \ v(\{i\})\notin C_i)\big).$
\end{claim}

Next, consider a $v$ such that $v(\{i\}) \in C_i$ and event $B$ is satisfied. Note that the additional restriction to $B$ implies that $v(\{j\}) < H$ for all $j \neq i$. If we then denote $B^*_i$ the event that $v(\{i\}) \in C_i$ and $v(\{j\}) < H$ for all $j \neq i$, then we again see that $D|B^*_i$ is subadditive over independent items. Consider also the abuse of notation which defines $D_i|(v(\{i\}) \in C_i)$ to be the distribution which first draws $X_i \leftarrow D_i|(v(\{i\}) \in C_i)$, and then defines $v(S):= v(S \cap \{i\})$ (that is, the distribution is essentially only over values for item $i$, but we extend it to a distribution over $n$-item valuations by ignoring any items $\neq i$). 

We first observe that, because any $v' \leftarrow D_i | (v(\{i\}) \in C_i)$ has zero value for all other items, and because there is an option on $M'_i$ which provides \emph{at least as much} utility as receiving item $i$ with probability $q_i$ and paying $p_i$ for it, that such $v'$ will always purchase an option from $M'_i$ (before or after the discounts) when presented with menu $M'$ (this is because, by definition of $C_i$, such buyers have higher utility for receiving item $i$ with probability $q_i$ and paying $p_i$ than any option on $M^*$, and such options only become more attractive after the discounts, because they are more expensive than any option on $M^*$). We therefore immediately get that:
$$\REV_{M'}\big(D_i | v(\{i\}) \in C_i \big) = \REV_{M'_i}\big(D_i | v(\{i\}) \in C_i \big).$$

Now we wish to use this to bound $\REV_{M''}(D_i | B^*_i)$. So consider two coupled draws $v, v'$ from $D|B^*_i$ and $D_i | (v(\{i\}) \in C_i)$, coupled by $X_i$. Then as $v(\{i\}) = v'(\{i\})$, and both distributions are subadditive (over independent items), we have that the maximum possible value of $|v(S) - v'(S)|$ is $v([n]\setminus \{i\})$. That is, for any menu $M'$ (and in particular, the $M'$ we've defined above):
$$\delta_{M'}(D|B^*_i,\ D_i | v(\{i\})\in C_i) \quad \leq \quad \mathbb{E}_{v \leftarrow D|B^*_i}[v([n]\setminus\{i\})]  \quad = \quad  \mathbb{E}_{v \leftarrow D^{[n]\setminus\{i\}}|B^*_i}[v([n]\setminus\{i\})].$$

Observe also that as $D^{[n]\setminus\{i\}}|B^*_i$ is subadditive over independent items, and has $v(\{j\}) \leq H$ for all $j$, we get immediately from Proposition~\ref{prop:RW} that: 
\[ \mathbb{E}_{v \leftarrow D^{[n]\setminus\{i\}}|B^*_i}[v([n]\setminus\{i\})] \quad \leq \quad 6\BREV \big(D^{[n]\setminus\{i\}}|B^*_i \big) + 4H/\ln(2) \quad \leq \quad 6\BREV(D) + 4H/\ln(2). \] 
Therefore, Proposition~\ref{prop:advancedNudge} immediately lets us conclude the following claim:

\begin{claim}\label{claim:icase} We have
\begin{align*}\REV_{M''}(D|B^*_i) &\geq (1-\varepsilon) \REV_{M'}(D_i | v(\{i\}) \in C_i) - (6\BREV(D) + 4H/\ln(2))/\varepsilon\\
&= (1-\varepsilon) \REV_{M'_i}(D_i | v(\{i\}) \in C_i) - (6\BREV(D) + 4H/\ln(2))/\varepsilon.
\end{align*}
\end{claim}

Now using the above claims we can wrap up the proof of Lemma~\ref{lem:otherway}  as:
\begin{align*}
\REV_{M''}(D\cdot \mathbb{I}(B)) &= \REV_{M''}(D\cdot \mathbb{I}(B \wedge \forall i,\ v(\{i\}) \notin C_i)) + \sum_i \REV_{M''}(D \cdot \mathbb{I}(B \wedge v(\{i\}) \in C_i)).\\
&\geq (1-\varepsilon)\REV_{M^*}(D\cdot \mathbb{I}(B \wedge \forall i,\ v(\{i\})\notin C_i) )+ \sum_i \Pr[B^*_i] \cdot \REV_{M''}(D|B^*_i).\\
&\geq (1-\varepsilon)\REV_{M^*}(D\cdot \mathbb{I}(B \wedge \forall i,\ v(\{i\})\notin C_i)) \\
&\quad+ \sum_i \Pr[B^*_i] \cdot \left((1-\varepsilon)\REV_{M'_i}(D_i | v(\{i\})\in C_i) - \big(6\BREV(D)+4H/\ln(2)\big)/\varepsilon\right).\\
&\geq (1-\varepsilon)\REV_{M^*}(D\cdot \mathbb{I}(B \wedge \forall i,\ v(\{i\})\notin C_i) \\
&\quad+ \sum_i (1-\varepsilon)^2 \REV_{M'_i}(D_i \cdot \mathbb{I}( v(\{i\})\in C_i)) - \sum_i \Pr[B^*_i]\big(6\BREV(D)+4H/\ln(2)\big)/\varepsilon.\\
&\geq (1-\varepsilon)\REV_{M^*}(D\cdot \mathbb{I}(B \wedge \forall i,\ v(\{i\})\notin C_i)) \\
&\quad+ \sum_i (1-\varepsilon)^2 \REV_{M'_i}(D_i \cdot \mathbb{I}( v(\{i\})\in C_i)) - \frac{\SREV(D)/E}{1-\SREV(D)/E}\big(6\BREV(D)+4H/\ln(2)\big)/\varepsilon.\\
&\geq (1-\varepsilon)\REV_{M^*}(D\cdot \mathbb{I}(B \wedge \forall i,\ v(\{i\})\notin C_i) \\
&\quad+ \sum_i (1-\varepsilon)^2 \REV_{M'_i}(D_i \cdot \mathbb{I}( v(\{i\})\in C_i)) - O(\varepsilon) \cdot \REV(D).
\end{align*}
Here the first equality above follows directly from Sub-Domain Stitching (Lemma~\ref{lem:sds}). The second inequality follows from Claim~\ref{claim:M*}, and by rewriting the revenue of a restricted distribution as that of a conditional distribution times the probability of the event (and recalling that $B^*_i$ is exactly $B \wedge v(\{i\}) \in C_i$).

The third inequality follows directly from Claim~\ref{claim:icase}. The fourth inequality follows by observing that event $B$ occurs with probability at least $1-\varepsilon$ conditioned on $v(\{i\}) \in C_i$ (because $B$ occurs as long as all other values are $\leq H$, and the complement of this can occur with probability at most $\varepsilon$ without contradicting the definition of $\SREV(D)$), and converting back from a conditional distribution to a restricted distribution.

The fifth inequality follows by observing that $\Pr[B^*_i] \leq \Pr[v(\{i\}) \in C_i] \leq \Pr[v(\{i\}) \geq E]$. Moreover, we must have that $\sum_i \Pr[v(\{i\}) \geq E] \leq \frac{\SREV(D)/E}{1-\SREV(D)/E}$ (Claim~\ref{claim:sumpi}). Therefore, $\sum_i \Pr[B^*_i] \leq \frac{\SREV(D)/E}{1-\SREV(D)/E}$. 

The final inequality follows by the assumptions on $E$. Certainly $\SREV(D)/E = O(\varepsilon^3)$, so $\frac{\SREV(D)/E}{1-\SREV(D)/E}$ times $6\BREV(D)/\varepsilon$ is $O(\varepsilon^2) \cdot \REV(D)$. Similarly, $E \geq H/\varepsilon^2$, so $\SREV(D)/E \leq \varepsilon^2\SREV(D)/H$, so $(\SREV(D)/E)\cdot H/\varepsilon = O(\varepsilon) \cdot \SREV(D))$. 
\end{proof}

Observation~\ref{obs:oneway} and Lemma~\ref{lem:otherway} together assert that the only difference between $\REV_M(D \cdot \mathbb{I}(B))$ and $\REV_{M''}(D\cdot \mathbb{I}(B))$ comes via the revenues $M_i$ vs. $M'_i$ on single-dimensional distributions. So the last remaining step is to argue that there exists an $M'_i$ of the desired simple form with $\REV_{M'_i}(D_i \cdot \mathbb{I}(B\wedge v(\{i\}) \in C_i)) \geq \REV_{M_i}(D_i \cdot \mathbb{I}(B\wedge v(\{i\}) \in C_i))$. Lemma~\ref{lem:wrapup3} is also explicit in~\cite{BabaioffGN17}, although the reader familiar with Myerson's theory of virtual values may find the proof below simpler. 

\begin{lemma}\label{lem:wrapup3} Let $r_i$ denote the Myerson reserve for the distribution $D_i \cdot \mathbb{I}(v(\{i\}) \in C_i)$, and let $M'_i$ have only two options: receive item $i$ with probability $q_i$ for price $p_i$, or receive item $i$ with probability $1$ for price $p_i + (1-p_i)r_i$. Then:
$$\REV_{M'_i}(D_i \cdot \mathbb{I}(v(\{i\})\in C_i)) \geq \REV_{M_i}(D_i \cdot \mathbb{I}(v(\{i\})\in C_i).$$
\end{lemma}
\begin{proof}
Intuitively, the proof will follow by arguing that $M'_i$ is the optimal mechanism which contains the option to receive item $i$ with probability $q_i$ for price $p_i$, and that $M_i$ is \emph{some} auction with this property, and therefore the revenue of $M'_i$ can only be higher. There is some work to take care of the fact that $M_i$ doesn't necessarily contain the option to receive item $i$ with probability $q_i$ for price $p_i$ (but just a sequence of options approaching this), and also it is rather messy to prove that a mechanism is optimal from first principles. To give a short and direct proof, we will rely on Myerson's theory of virtual values~\cite{Myerson81}. Specifically, let $\bar{\varphi}_i(v)$ denote the Myerson ironed virtual value associated with distribution $D_i |(v(\{i\}) \in C_i)$ for value $v$. Then the following facts all follow from the definition of ironed virtual values~\cite{Myerson81, Hartlinebook}:
\begin{itemize}
\item $\bar{\varphi}_i(\cdot)$ is defined on the support of $D_i | (v(\{i\}) \in C_i)$, which is just $C_i$.
\item $r_i = \inf\{v|\ \bar{\varphi}_i(v) \geq 0\}$ (note the optimal reserve for $D|F$ and $D\cdot \mathbb{I}(F)$ are identical for any event $F$). 
\item $\bar{\varphi}_i(v) \geq 0$ for all $v > r_i$.
\item $\bar{\varphi}_i(v) < 0$ for all $v < r_i$.
\item The revenue from any menu is upper bounded by $\mathbb{E}_{v \leftarrow D_i \cdot \mathbb{I}(v(\{i\})\in C_i)}[x_i(v) \cdot \bar{\varphi}_i(v)]$, where $x_i(v)$ denotes the allocation probability of item $i$ for the menu option which $v$ chooses to purchase.
\item Equality holds in the bullet above whenever $x_i(\cdot)$ is constant in any interval contained in $C_i$ where $\bar{\varphi}_i(\cdot)$ is constant.
\end{itemize}

The bullets above suffice for our claim. Indeed, observe that for any $v \in C_i$, this value $v$ certainly purchases an option from $M_i$ which awards item $i$ with probability at least $q_i$ (by definition of $q_i, p_i, C_i$). In addition, observe that any $v > r_i$ purchases an option which awards item $i$ with probability at most $1$.

On the other hand, in $M'_i$, every value $v \in C_i \cap [0,r_i)$ chooses to receive the item with probability only $q_i$, and every $v > r_i$ receives the item with probability $1$ (depending on the exact value of $\bar{\varphi}(r_i)$, $x_i(r_i)$ could be set to be either $q_i$ or $1$. Specifically, if $\bar{\varphi}_i(r_i) \geq 0$, set $x_i(r_i) = 1$. Otherwise, set $\bar{\varphi}_i(r_i) = 0$). Therefore, we conclude that:
\begin{itemize}
\item In the range where $\bar{\varphi}_i(v)$ is negative, $M_i$ awards a weakly higher allocation probability than $M'_i$. 
\item In the range where $\bar{\varphi}_i(v)$ is non-negative, $M_i$ awards a weakly lower allocation probability than $M'_i$.
\item As $M'_i$ has constant allocation probability whenever $\bar{\varphi}(\cdot)$ is constant, the revenue from $M'_i$ is equal to the ironed virtual value upper bound.
\end{itemize}

The first two bullets conclude that the ironed virtual value upper bound for $M'_i$ exceeds that of $M_i$. Therefore, the revenue of $M_i$ is at most the ironed virtual value upper bound of $M_i$, which is at most the ironed virtual value upper bound of $M'_i$, which is the revenue of $M'_i$.
\end{proof}

And now we can wrap up the proof of Proposition~\ref{prop:main3}.

\begin{proof}[Proof of Proposition~\ref{prop:main3}]
We have the following sequences of inequalities, starting from any mechanism $M$ of the form promised by Proposition~\ref{prop:main2}. For this $M$, define $M^*,M_i, C_i, M'_i, M''$ as above. Now
\begin{align*}
(1-O(\varepsilon))\REV(D) &\leq (1-\varepsilon)^2\REV_M(D\cdot \mathbb{I}(B))\\
&\leq (1-\varepsilon)^2\REV_{M^*}(D\cdot \mathbb{I}(B) \wedge \forall i,\ v(\{i\}) \notin C_i) + (1-\varepsilon)^2\sum_i \REV_{M_i} \big(D_i \cdot \mathbb{I}(B^*_i) \big)\\
&\leq (1-\varepsilon)^2\REV_{M^*}(D\cdot \mathbb{I}(B) \wedge \forall i,\ v(\{i\}) \notin C_i) + (1-\varepsilon)^2\sum_i \REV_{M'_i} \big(D_i \cdot \mathbb{I}(B^*_i)\big)\\
&\leq \REV_{M''}\big(D\cdot\mathbb{I}(B)\big) + O(\varepsilon) \cdot \REV(D)\\
\Rightarrow \REV_{M''}(D) &\geq \REV_{M''}\big(D\cdot \mathbb{I}(B)\big) \geq (1-O(\varepsilon))\cdot \REV(D).
\end{align*}
Here the first line follows directly from Proposition~\ref{prop:main2}. The second line is a direct application of Observation~\ref{obs:oneway}. The third line is a direct application of Lemma~\ref{lem:wrapup3}. The fourth line is a direct application of Lemma~\ref{lem:otherway}. The fifth line just observes that $\REV_{M''}(D) \geq \REV_{M''}(D \cdot \mathbb{I}(B))$. 
\end{proof}

\subsection{A Reduction from Unbounded to Bounded} \label{sec:stepfour}
From here, the next step is to separately find a good menu for the role of $M^*$ with low menu complexity, observing that it need only consider options priced $\leq E$, and therefore seems like a bounded problem. This part is true, but one needs to be careful with changing $M^*$ as it will also impact what options on $M'_i$ are purchased. Indeed, if we carelessly change $M^*$ to (for instance) offer item $i$ with probability much larger than the original $M^*$, this will cause buyers who previously purchased an expensive option on $M'_i$ to now possibly prefer a cheaper option in $M^*$. As such, a careless nudge-and-round discretization to get a new $M^*$ runs a substantial risk, and~\cite{BabaioffGN17} addresses this by directly tailoring a nudge-and-round so as not to interfere with the expensive parts of the menu.

We take a different approach, and observe that while it is true that we can't be completely careless with $M^*$, we don't need to be particularly careful either. We show that the only important aspect of $M^*$ which must be preserved is the supremum probability with which any option awards item $i$ (and this matters for all $i$). But as long as this is preserved, there is little interference with the expensive part of the menu. Parts of this proof will look similar to step three (and indeed, an argument like the one we present below could replace step three entirely), but we provide it after step three because notation will be greatly simplified with only $2n$ expensive options. The proof comes in two steps. First, we derive an upper bound on the revenue of $M$, and then we show that this upper bound can be achieved by separately designing a menu for $D(T)$ and later adding the expensive options.

\begin{proposition}\label{prop:reduce1} Let $T \geq E/\varepsilon \geq \max\{H/\varepsilon^3,\BREV(D)/\varepsilon^2\}$, and $H \geq \SREV(D)/\varepsilon$. Let $M$ be the mechanism promised by Proposition~\ref{prop:main3}. Then: 
$$(1-O(\varepsilon))\REV_M(D\cdot \mathbb{I}(B)) \leq \sup_{M'} \Big\{\REV_{M'}(D(T)\cdot \mathbb{I}(B)) + \sum_i \ell_i(M')\REV(D_i \cdot \mathbb{I}( v(\{i\}) \geq T)) \Big\}.$$
\end{proposition}

\begin{proof}
Let us begin by establishing an upper bound on $\REV_M(D)$. We first introduce notation:
\begin{itemize}
\item For each $i$, let $(S_{i,1},p_{i,1}),(\{i\},p_{i,2})$ denote the two options for item $i$ which are priced higher than $E$ (if there aren't two distinct options, just set $(S_{i,1},p_{i,1}):=(\{i\},p_{i,2})$. If neither option exists, then this definition is void and will never be invoked). 
\item Let $b_i:=\Pr[i \in S_{i,1}]$.
\item Let $y_i$ denote the minimum value such that a buyer with value $v(S)=y_i |S \cap \{i\}|$ would purchase $(S_{i,1},p_{i,1})$ from $M$.
\item Let $z_i$ denote the minimum value such that a buyer with value $v(\{i\}) = z_i$ would purchase $(\{i\},p_{i,2})$ over $(S_{i,1},p_{i,1})$.
\item Let $a_i$ denote the supremum allocation probability of getting item $i$ from an option in $M$ priced at most $E$.
\end{itemize}

\begin{observation}\label{obs:boundprices}The following bounds on $p_{i,1},p_{i,2}$ hold:
\begin{align*}
p_{i,2} &= p_{i,1} + z_i(1-b_i).\\
p_{i,1} &\in [y_i(b_i-a_i), y_i(b_i-a_i)+E].\\
p_{i,2} &\in [z_i(1-b_i)+y_i(b_i-a_i), z_i(1-b_i)+y_i(b_i-a_i)+E].
\end{align*}
\end{observation}
\begin{proof}
The first line follows immediately from the definition of $z_i$. The third line follows immediately from the first two. The second line follows from incentive constraints: there is some option on the menu $M$ with probability $a_i$ of giving item $i$ (or a sequence of options approaching this probability). The price of this option is between $0$ and $E$ (or the prices of all options in the sequence are between $0$ and $E$). The price therefore cannot be less than $(y_i-\varepsilon)(b_i-a_i)$, or else we'd contradict the definition of $y_i$: a buyer with value $v(S)=(y_i-\varepsilon)|S \cap \{i\}|$ would prefer $(S_{i,1},p_{i,1})$ to any option which awards item $i$ with probability $\leq a_i$ \emph{even for free}. Similarly, the price cannot be more than $y_i(b_i-a_i+\varepsilon)+E$, or else a buyer with value $v(S)=y_i |S \cap \{i\}|$ would strictly prefer to pay $E$ and get item $i$ with probability $a_i-\varepsilon$ for price $E$, and such an option is certainly on the menu $M$, contradicting the definition of $y_i$.
\end{proof}
Now, we wish to use this notation to specify an upper bound on $\REV_M(D\cdot \mathbb{I}(B))$. The plan of attack will be as follows. First, we will define a distribution $D' \cdot \mathbb{I}(B)$ which zeroes out values for all items $j \neq i$ whenever $v(\{i\})$ is high, and claim that this can only improve the revenue that $M$ gets. Next, we will argue that truncating $D'$ at $T$ does not hurt the revenue of $D'$ by much. Finally, we will couple $D'(T)$ with $D(T)$ using Proposition~\ref{prop:advancedNudge} and claim that the revenue further lost is not much.

We begin by defining $D' \cdot \mathbb{I}(B)$. First, draw $v \leftarrow D\cdot \mathbb{I}(B)$. If for all $i$, $v(\{i\}) < y_i$, output $v':= v$. Otherwise, there must exist at most one $i$ for which $v(\{i\}) \geq y_i$ (because of event $B$). Output the valuation $v'$ with $v'(S):= v(S \cap \{i\})$. That is, zero out all values for $v'(\{j\})$ for $j \neq i$ (and keep $v'$ subadditive). We first claim that:

\begin{claim}\label{claim:D'} $\REV_M(D\cdot \mathbb{I}(B)) \leq \REV_M(D'\cdot \mathbb{I}(B))$.
\end{claim}
\begin{proof}
Consider the contribution of a value $v$ to the LHS and its couple $v'$ to the RHS. If $v(\{i\}) < y_i$ for all $i$, then $v = v'$, and they purchase the same option from $M$ and therefore contribute the same revenue. If there exists an $i$ for which $v(\{i\}) \geq y_i$, then $v$ and $v'$ have the same preferences between $(S_{i,1},p_{i,1})$ and $(\{i\},p_{i,2})$, but maybe $v$ prefers a cheaper option priced $\leq E$ (whereas $v'$ certainly prefers one of these options). Therefore, $v'$ makes a more expensive purchase than $v$.
\end{proof}

Next, we want to argue that truncating $D'$ at $T$ doesn't affect the revenue of $D'$ by much. Let $M'$ denote the mechanism $M$ after removing all options which are never purchased by any buyer with $v(\{i\}) < T$ for all $i$.

\begin{claim}\label{claim:D'2} $\REV_M(D'\cdot \mathbb{I}(B))\leq \REV_{M'}(D'(T)\cdot \mathbb{I}(B)) + \sum_i \ell_i(M')\REV(D_i \cdot \mathbb{I}(v(\{i\}) \geq T)) + O(\varepsilon)\SREV(D).$
\end{claim}
\begin{proof}
Again consider any $v$ from $D'\cdot \mathbb{I}(B)$ and its couple $v'$ from $D'(T) \cdot \mathbb{I}(B)$. If $v(\{i\}) \leq T$ for all $i$, then actually $v = v'$ and they make the same purchase from either $M$ or $M'$, so the contribution to both sides is the same.

If $v(\{i\}) > T$ for some $i$ (by event $B$, there is at most one such $i$), then perhaps $v$ purchases an option from $M$ which costs at most $E$ (note that this will only happen when $v(\{i\})<y_i$, since we are drawing from $D'$ and not $D$). We claim that the total revenue coming from such cases to the LHS is at most $O(\varepsilon)\SREV(D)$. To see this, observe that $\sum_i \Pr[v(\{i\}) \geq T] \leq (1+O(\varepsilon))\SREV(D)/T$ (otherwise this will contradict the definition of $\SREV(D)$ by setting price $T$ on each item). Therefore, even if we make the maximum possible $E$ in all such events, the total contribution is at most $(1+O(\varepsilon))\SREV(D)\cdot E/T = O(\varepsilon)\cdot \SREV(D)$.

If $v(\{i\}) > T$ for some $i$, and further $v(\{i\}) \geq y_i$, then there are a few cases.
\begin{itemize}
\item Case 1: $z_i > v(\{i\}) > T \geq y_i$. Then $v$ will purchase the option $(S_{i,1},p_{i,1})$, and so will $v'$.
\item Case 2: $v(\{i\}) > T \geq z_i$. Then $v$ will purchase the option $(\{i\},p_{i,2})$, and so will $v'$.
\item Case 3: $z_i > v(\{i\}) \geq y_i \geq T$. Then $v$ will purchase the option $(S_{i,1},p_{i,1})$, and contribute at most $y_i(b_i-a_i)+E$ to the revenue of $M$.
\item Case 4: $v(\{i\}) \geq z_i \geq y_i > T$. Then $v$ will purchase the option $(\{i\},p_{i,2})$, and contribute at most $z_i(1-b_i) + y_i(b_i-a_i) + E$ to the revenue of $M$.
\item Case 5: $v(\{i\}) \geq z_i > T \geq y_i$. Then $v$ will purchase the option $(\{i\},p_{i,2})$, and contribute at most $z_i(1-b_i) + y_i(b_i-a_i) +E$ to the revenue of $M$, and $v'$ will purchase the option $(S_{i,1},p_{i,1})$, and contribute at least $y_i(b_i-a_i)$ to the revenue of $M'$.
\end{itemize}

So only the final three bullets have some unaccounted revenue from $D'$ that is not achieved by $D'(T)$. Let us first argue that we can ignore the $+E$ terms. Indeed, observe that each of the final three bullets only occur when $v(\{i\}) > T$ for some $i$. The probability of this event is at most $(1+O(\varepsilon))\SREV(D)/T$, and so therefore the total contribution of the $+E$ terms to this cases is at most $O(\varepsilon)\cdot \SREV(D)$. So just to clarify, we now have:
\begin{align*}
\REV_M(D'\cdot \mathbb{I}(B)) & \leq \REV_{M'}(D'(T)\cdot \mathbb{I}(B)) \\
& \quad + \sum_{i, y_i >T}  \Pr[v(\{i\}) \in [y_i,z_i)] \cdot y_i(b_i-a_i) +\Pr[v(\{i\}) \geq z_i]\cdot (z_i(1-b_i)+y_i(b_i-a_i))  \\
& \quad + \sum_{i, z_i > T \geq y_i} \Pr[v(\{i\}) \geq z_i]\cdot z_i(1-b_i)+ O(\varepsilon)\SREV(D). 
\end{align*}

But now observe that for any $i$ such that $z_i > T$, we have that $z_i\Pr[v(\{i\}) \geq z_i] \leq \REV(D_i \cdot \mathbb{I}(v(\{i\}) \geq T))$. Therefore, 
$$(1-b_i)z_i\Pr[v(\{i\}) \geq z_i] \leq \ell_i(M') \cdot \REV(D_i \cdot \mathbb{I}(v(\{i\}) \geq T))$$ 
as $\ell_i(M') = 1-b_i$ for these items. 

Similarly, for any $i$ such that $y_i > T$, we have: 
\begin{align*}
&\Pr\big[v(\{i\}) \in [y_i,z_i) \big] \cdot y_i(b_i-a_i) +\Pr\big[v(\{i\}) \geq z_i\big]\cdot (z_i(1-b_i)+y_i(b_i-a_i)) \\ 
= & \Pr\big[v(\{i\}) \geq y_i\big]\cdot y_i(b_i-a_i) + \Pr\big[v(\{i\}) \geq z_i\big] \cdot z_i(1-b_i). 
\end{align*}
Also, $\Pr[v(\{i\}) \geq y_i]\cdot y_i \leq \REV(D_i \cdot \mathbb{I}(v(\{i\}) \geq T))$, and $\Pr[v(\{i\}) \geq z_i] \cdot z_i \leq \REV(D_i \cdot \mathbb{I}(v(\{i\}) \geq T))$. As $1-a_i = \ell_i(M')$ (for these items), we get the desired claim.
\end{proof}

Finally, we now argue that going back from $D'(T)$ to $D(T)$ doesn't lose much, using Proposition~\ref{prop:advancedNudge}. Let $M''$ denote $M'$ after multiplying all prices by $(1-\varepsilon)$.

\begin{claim}\label{claim:D'3} $\REV_{M''}(D(T) \cdot \mathbb{I}(B)) \geq (1-\varepsilon)\REV_{M'}(D'(T)\cdot \mathbb{I}(B)) - O(\varepsilon) \cdot \REV(D)$.

\end{claim}
\begin{proof}
The proof will follow from an application of Proposition~\ref{prop:advancedNudge} to $D(T)\cdot \mathbb{I}(B)$ and $D'(T) \cdot \mathbb{I}(B)$, after computing $\delta(\cdot, \cdot)$ via $D(T)|B$ and $D'(T)|B$. Observe that $M''$ and $M'$ indeed satisfy the hypotheses, so we just need to bound $\delta(D(T)\cdot \mathbb{I}(B),D'(T)\cdot \mathbb{I}(B))$.

Couple draws $v$ from $D(T)|B$ and $v'$ from $D'(T)|B$ in the obvious way (first draw $v$ from $D(T)|B$, and then zero out values for items $\neq i$ if $v(\{i\}) \geq T$). Observe first that whenever $v(\{i\}) < T$ for all $T$, that $v = v'$, so the only contribution to $\delta(D, D')$ comes when $v(\{i\}) = T$ for some $i$. But observe also that $D(T)|(B\wedge v(\{i\})=T)$ is subadditive over independent items, and that the maximum difference arises from $v([n]\setminus \{i\})-v'([n]\setminus\{i\}) = v([n]\setminus\{i\})$. But now we can apply Proposition~\ref{prop:RW} to claim that the expected value of $v([n]\setminus\{i\})$ when drawn from $D(T)|(B\wedge v(\{i\}) = T)$ is at most $6\BREV(D) + 4H/\ln(2)$. So we get:
\begin{align*}
\delta(D(T)\cdot \mathbb{I}(B),D'(T) \cdot \mathbb{I}(B)) &\leq \Pr[B] \cdot \delta(D(T)|B,D'(T)|B)\\
&\leq \sum_i \Pr[B] \cdot Pr[v(\{i\})=T|B] \cdot (6\BREV(D) + 4H/\ln(2))\\
&\leq (6\BREV(D) + 4H/\ln(2))\cdot \sum_i \Pr[v(\{i\}) = T]\\
&\leq O(\varepsilon^3) \cdot \REV(D).
\end{align*}

Above, the first line follows by observing that both $v$ and $v'$ are the zero function whenever event $B$ does not occur. If event $B$ occurs, then both are drawn coupled from $D(T)|B$ and $D'(T)|B$. The second line follows by the reasoning in the previous paragraph. The third line simply observes that $\Pr[v(\{i\}) = T|B] \cdot \Pr[B] = \Pr[v(\{i\}) = T]$, and the final line observes that $\sum_i \Pr[v(\{i\}) = T]$ is both $O(\varepsilon^4)$ and also $O(\varepsilon^3)\cdot \SREV(D)/H$.
\end{proof}

We now wrap up by observing that:
\begin{align*}
&(1-\varepsilon)\REV_M(D\cdot \mathbb{I}(B)) \quad \leq \quad (1-\varepsilon)\REV_M(D'\cdot \mathbb{I}(B))\\
&\qquad \qquad \leq (1-\varepsilon)\REV_{M'}(D'(T) \cdot \mathbb{I}(B)) + (1-\varepsilon)\sum_i \ell_i(M')\cdot \REV(D_i \cdot \mathbb{I}(v(\{i\})\geq T))+O(\varepsilon)\SREV(D)\\
&\qquad \qquad \leq \REV_{M''}(D(T) \cdot \mathbb{I}(B)) + \sum_i \ell_i(M'')\cdot \REV(D_i \cdot \mathbb{I}(v(\{i\})\geq T))+O(\varepsilon) \REV(D).
\end{align*}

The first line follows directly from Claim~\ref{claim:D'}, the second from Claim~\ref{claim:D'2}, and the third from Claim~\ref{claim:D'3}. Observe that $M''$ is some mechanism, and so surely the supremum in the RHS of the proposition statement exceeds the bound above. Because our original $M$ was promised by Proposition~\ref{prop:main3}, we also have that $O(\varepsilon)\REV(D) = O(\varepsilon)\REV_M(D)$.
\end{proof}

\begin{proposition}\label{prop:reduction2}
Let $r_i$ denote the optimal price for $D_i\cdot \mathbb{I}(v(\{i\})\geq T)$. For any mechanism $M'$ (which contains only options purchased by some $v \in \text{support}(D(T))$) let  $M''$ denote the mechanism $(M'|_E)^{T,\vec{r}}$, then:
$$\REV_{M''}(D\cdot \mathbb{I}(B)) \geq (1-2\varepsilon)\Big(\REV_{M'}(D(T)\cdot \mathbb{I}(B)) + \sum_i \ell_i(M')\REV(D_i \cdot \mathbb{I}(v(\{i\}) \geq T))\Big)-O(\varepsilon)\cdot \REV(D).$$
\end{proposition}

\begin{proof}
Like the proof of Proposition~\ref{prop:reduction}, we will consider a chain of inequalities going through $D'\cdot \mathbb{I}(B)$, which again draws a value $v$ from $D\cdot \mathbb{I}(B)$ and outputs either $v':=v$ if $v(\{i\}) \leq T$ for all $i$, or $v'$ with $v'(S):= v(S \cap \{i\})$ otherwise. We first observe that:

\begin{claim}\label{claim:final1}
$\REV_{M'|_E}(D(T) \cdot \mathbb{I}(B)) \geq (1-\varepsilon)\REV_{M'}(D(T)\cdot \mathbb{I}(B)) - O(\varepsilon)\cdot \REV(D).$
\end{claim}
\begin{proof}
This follows from a direct application of Proposition~\ref{prop:exclusive}. 
\end{proof}

Next, we want to claim that:

\begin{claim}\label{claim:final2} $\REV_{M'|_E}(D'(T)\cdot \mathbb{I}(B)) \geq  \REV_{M'|_E}(D(T) \cdot \mathbb{I}(B)) - O(\varepsilon)\cdot \SREV(D)$.
\end{claim}
\begin{proof}
To see this, couple draws $v'$ from $D'(T) \cdot \mathbb{I}(B)$ and $v$ from $D(T) \cdot \mathbb{I}(B)$ in the obvious way. Observe that $v = v'$ whenever $v(\{i\}) <T$ for all $i$, and therefore $v$ and $v'$ purchase the same option. 

If $v(\{i\}) = T$, then maybe $v$ purchased an option from $M'$ priced at most $E$. But the total probability that $v(\{i\}) = T$ for \emph{some} $i$ is at most $\frac{\SREV(D)/T}{1-\SREV(D)/T}$, so the total revenue contributed from these cases is $\frac{\SREV(D)\cdot E/T}{1-\SREV(D)/T} = O(\varepsilon)\cdot \SREV(D)$. 

Finally, maybe $v(\{i\}) = T$, and $v$ purchased an option from $M'$ priced more than $E$. Then this option must offer only item $i$, and $v'$ will have exactly the same preferences as $v$ among such options and make the same purchase.
\end{proof}

Now, we want to understand the menu $(M'|_E)^{T,\vec{r}}$ (see Definition~\ref{def:excloptions}), but not yet with the prices further reduced by $(1-\varepsilon)$. Below, let $M'''$ denote the menu $(M'|_E)^{T,\vec{r}}$ with all prices multiplied by $(1-\varepsilon)$ 
(so to be clear, first $M'$ is made $E$-exclusive, which reduces the prices by $(1-\varepsilon)$. Then the exclusive options are added, but the prices are not yet further reduced by $(1-\varepsilon)$).

\begin{claim}\label{claim:final3} $\REV_{M'''}(D'\cdot \mathbb{I}(B)) \geq \REV_{M'|_E}(D'(T)\cdot \mathbb{I}(B)) + \sum_i \ell_i(M') \REV(D_i \cdot \mathbb{I}(v(\{i\}) \geq T)) - O(\varepsilon)\cdot \SREV(D)$.
\end{claim}
\begin{proof}
We want to claim that the RHS is essentially breaking down the contribution to the revenue from $D'\cdot \mathbb{I}(B)$ into two parts: that from buyers with $v(\{i\}) \leq T$ for all $i$, and that from buyers with $v(\{i\}) > T$ for some $i$. 

Let us consider coupled draws $v, v'$ from $D'\cdot \mathbb{I}(B)$ and $D'(T) \cdot \mathbb{I}(B)$, coupled in the obvious way. Then if $v(\{i\}) \leq T$ for all $i$, we in fact have $v = v'$, and $v$ will certainly not purchase one of the $n$ extra options. So $v$ and $v'$ will purchase the same option from $M'|_E$.

Now let us consider the case that $v(\{i\}) > T$ for some $i$ (by event $B$, there is at most one such $i$). Then maybe $v'$ purchases an option priced $\leq E$ from $M'|_E$. But again the total contribution of revenue from such events can be at most $O(\varepsilon)\cdot \SREV(D)$, so this is accounted for by the additional $O(\varepsilon)\cdot \SREV(D)$ term.

Finally, maybe $v(\{i\}) > T$ for some $i$, and $v'$ purchases an option priced $> E$ from $M'|_E$ which offers only item $i$, and pays price $p_i$. Then $v$ will either purchase exactly this option from $M'|_E$, or perhaps the additional option $(\{i\}, p_i + r_i \ell_i(M'))$ if further $v(\{i\}) \geq r_i$. So our total lower bound on the revenue from $D'\cdot \mathbb{I}(B)$ is:
$$\REV_{(M'|_E)_{\vec{r}}}(D'\cdot \mathbb{I}(B))+O(\varepsilon)\cdot \SREV(D) \geq \REV_{M'|_E}(D'(T)\cdot \mathbb{I}(B)) + \sum_i \ell_i(M') r_i \cdot \Pr[v(\{i\}) \geq r_i] .$$

As $r_i \cdot \Pr[v(\{i\}) \geq r_i] = \REV(D_i \cdot \mathbb{I}(v(\{i\}) \geq T))$, we have the desired claim.
\end{proof}

And finally, we just need to relate $\REV_{M''}(D \cdot \mathbb{I}(B))$ to $\REV_{(M'''}(D' \cdot \mathbb{I}(B))$ (recall from the proposition statement that $M'':= (M|_E)^{T,\vec{r}}$).

\begin{claim}\label{claim:final4}$\REV_{M''}(D \cdot \mathbb{I}(B)) \geq (1-\varepsilon)\REV_{M'''}(D' \cdot \mathbb{I}(B)) - O(\varepsilon)\cdot \REV(D)$.
\end{claim}
\begin{proof}
This is a direct application of Proposition~\ref{prop:advancedNudge}. Observe that $M''$ and $M'''$ satisfy the hypotheses, so our only job is to bound $\delta(D \cdot \mathbb{I}(B), D' \cdot \mathbb{I}(B))$.

Couple $(v, v')$ from $(D \cdot \mathbb{I}(B), D' \cdot \mathbb{I}(B))$ in the obvious way. Then observe that $v$ and $v'$ are identical whenever $v(\{i\}) \leq T$ for all $i$. When $v(\{i\}) > T$ for some $i$ (at most one $i$, by event $B$), we have that $v(S) \geq v'(S)$ for all $S$, and this difference is at most $v([n]\setminus\{i\})$ because $v$ is subadditive. We now compute $\delta(D \cdot \mathbb{I}(B), D' \cdot \mathbb{I}(B))$ via $D|B$ and $D'|B$. 

Observe also that $D|(B\wedge v(\{i\})>T)$ is subadditive over independent items, and that $v(\{j\}) \leq H$ for all $j \neq i$. So now we can apply Proposition~\ref{prop:RW} to claim that the expected value of $v([n]\setminus\{i\})$ when drawn from $D|(B\wedge v(\{i\}) >T)$ is at most $6\BREV(D) + 4H/\ln(2)$. So we get that:
\begin{align*}
\delta(D\cdot \mathbb{I}(B),D' \cdot \mathbb{I}(B)) &\leq \Pr[B] \cdot \delta(D|B,D'|B)\\
&\leq \sum_i \Pr[B] \cdot Pr[v(\{i\})=T|B] \cdot (6\BREV(D) + 4H/\ln(2))\\
&\leq (6\BREV(D) + 4H/\ln(2))\cdot \sum_i \Pr[v(\{i\}) = T]\\
&\leq O(\varepsilon^3) \cdot \REV(D).
\end{align*}
Above, the first line follows by observing that both $v$ and $v'$ are the zero function whenever event $B$ does not occur. If event $B$ occurs, then both are drawn coupled from $D|B$ and $D'|B$. The second line follows by the reasoning in the previous paragraph. The third line simply observes that $\Pr[v(\{i\}) > T|B] \cdot \Pr[B] = \Pr[v(\{i\}) > T]$, and the final line observes that $\sum_i \Pr[v(\{i\}) > T]$ is both $O(\varepsilon^4)$ and also $O(\varepsilon^3)\cdot \SREV(D)/H$.
\end{proof}

And now we can put everything together:
\begin{align*}
\REV_{M''}(D \cdot \mathbb{I}(B))&\geq (1-\varepsilon)\REV_{M'''}(D'\cdot \mathbb{I}(B)) - O(\varepsilon) \cdot \REV(D)\\
&\geq (1-\varepsilon)\Big(\REV_{M'|_E}(D'(T) \cdot \mathbb{I}(B))+\sum_i \ell_i(M') \REV(D_i\cdot \mathbb{I}(v(\{i\}) \geq T))\Big) - O(\varepsilon) \cdot \REV(D)\\
&\geq (1-\varepsilon)\Big(\REV_{M'|_E}(D(T)\cdot \mathbb{I}(B)) + \sum_i \ell_i(M')\REV(D_i \cdot \mathbb{I}(v(\{i\}) \geq T))\Big) - O(\varepsilon) \cdot \REV(D)\\
&\geq (1-\varepsilon)^2 \Big(\REV_{M'}(D(T) \cdot \mathbb{I}(B))+\sum_i \ell_i(M')\REV(D_i \cdot \mathbb{I}(v(\{i\})\geq T))\Big) - O(\varepsilon) \cdot \REV(D)\\
&\geq (1-2\varepsilon)\Big(\REV_{M'}(D(T) \cdot \mathbb{I}(B))+\sum_i \ell_i(M')\REV(D_i \cdot \mathbb{I}(v(\{i\})\geq T))\Big) - O(\varepsilon) \cdot \REV(D).
\end{align*}
The first line follows directly from Claim~\ref{claim:final4}, the second from Claim~\ref{claim:final3}, the third from Claim~\ref{claim:final2}, and the fourth from Claim~\ref{claim:final1}. The final line simply bounds $(1-\varepsilon)^2 \geq (1-2\varepsilon)$.
\end{proof}

We now conclude the main theorem of this section, which is a slightly more precise version of Theorem~\ref{thm:reduction} (in that it explicitly describes the process for getting a menu for $D$ from one for $D(T)$).

\begin{theorem}\label{thm:reduction2}
Let $r_i$ denote the optimal price for $D_i \cdot \mathbb{I}(v(\{i\}) \geq T)$, and let $M$ be any mechanism guaranteeing $\REV_M(D(T)) \geq (1-O(\varepsilon))\cdot \sup_{M'}\big\{\REV_{M'}(D(T) +\sum_i \ell_i(M')r_i \cdot \Pr[v(\{i\})\geq r_i] \big\}$. Then:
$$\REV_{(M|_E)^{T,\vec{r}}}(D) \geq (1-O(\varepsilon))\REV(D).$$
Further, observe that $\mc((M|_E)^{T,\vec{r}}) \leq n\mc(M) + n$ and that $\wsmc((M|_E)^{T,\vec{r}}) \leq n\wsmc(M)+n$. 
\end{theorem}

\begin{proof}[Proof of Theorem~\ref{thm:reduction}]
Theorem~\ref{thm:reduction} now follows direction from Theorem~\ref{thm:reduction2} together with Lemma~\ref{lem:modrev}. Theorem~\ref{thm:reduction2} implies that a $(1-O(\varepsilon))$-approximation to $\modrevmax_{\mathcal{D},\mathbb{R}^n}(D(T),\vec{r})$ implies a $(1-\varepsilon)$-approximation to $\revmax_\mathcal{D}(D)$, and Lemma~\ref{lem:modrev} implies that such an approximation can be achieved via a $(1-O(\varepsilon))$-approximation to $\revmax_\mathcal{D}(D(T,\vec{r}))$. Importantly, observe that all inputes are valid because $\mathcal{D}$ is closed under one of the canonical truncations.
\end{proof}


\section{Outline and Proofs: Symmetries and Optimal Mechanisms}\label{sec:symmetriesproofs}
In this section, we provide outlines and complete proofs for Theorems~\ref{thm:iiddemand},~\ref{thm:unit}, and~\ref{thm:constant}. We begin with a proof of Theorem~\ref{thm:symmetries}, which should appear straight-forward to readers familiar with~\cite{DaskalakisW12}. We'll actually prove a slightly stronger statement to accommodate \modrevmax, but specialized to $\mathcal{W} = \{\vec{0}\}$ the statement below is equivalent to Theorem~\ref{thm:symmetries}.

\begin{theorem}[Mild extension of~\cite{DaskalakisW12}]\label{thm:symmetries2} For any item permutation $\Sigma$ of partition size $s$, let $\mathcal{D}$ be the class of distributions $D$ which are $k$-demand over independent items, with each $|\text{support}(D_i)| \leq c$, and symmetric with respect to $\Sigma$. Let also $\mathcal{W}$ be any set of vectors such that $\sigma(\vec{w}) = \vec{w}$ for all $\vec{w} \in \mathcal{W}$. Then an optimal solution to $\modrevmax_{\mathcal{D},\mathcal{W}}$ can be found in time $\poly(n^{cs})$. Moreover, the mechanism output $M$ has $\ssmc(M) \leq n^{cs}$.
\end{theorem}

\begin{proof}
For simplicity of notation, we will use $v_i:= v(\{i\})$ in the proof below, since all $v \leftarrow D$ are $k$-demand. First, let $T_1,\ldots, T_s$ denote the partition witnessing that the partition size of $\Sigma$ is $s$. Define $C$ to pick exactly one representative from each equivalence class of $\text{support}(D)$ under $\Sigma$ in the following way: say that item $j$ is equivalent to item $\ell$ if there exists an $i$ such that $\{j,\ell\} \in T_i$, and denote this by $j \sim \ell$. Then $C$ contains all valuations in the support of $D$ such that $v_j \geq v_\ell$ for all $j \leq \ell$ such that $j \sim \ell$. 

So for example, if $\Sigma$ were the set of all item permutations, then $C$ would contain all valuations in the support of $D$ such that $v_j \geq v_\ell$ for all $j \leq \ell$. Also, let $Q(\vec{v})$ denote the probability, when $\vec{v}'$ is drawn from $D$, that $\vec{v}$ and $\vec{v}'$ are equivalent under $\Sigma$. Observe that $Q(\vec{v})$ can be computed in time $\poly(n,c)$ for any $\vec{v}$.\footnote{To see this, observe that we can compute the probability of drawing exactly $\vec{v}$ by simply multiplying together $\Pr[v_i \leftarrow D_i]$ for all $i$. Afterwards, we just need to multiply by the number of elements equivalent to $\vec{v}$. This can also be computed simply: within each $T_i$, for each possible value $x_1,\ldots, x_c$ in the support of $D_j$ ($j \in T_i$), there are some number $v(i,\ell)$ items in $T_i$ with value $x_\ell$, and therefore there are $\binom{|T_i|}{v(i,1),v(i,2),\ldots, v(i,c)}$ ways to permute within $T_i$, and $\prod_i\binom{|T_i|}{v(i,1),v(i,2),\ldots, v(i,c)}$ total members of the equivalence class.}

\begin{align}
\text{Variables:} \quad &p(\vec{v}),\ \forall \vec{v} \in C\nonumber\\
&x_i(\vec{v}),\ \forall \vec{v} \in C, \ i \in [n]\nonumber\\
&\ell_i,\ \forall i \in [n]\nonumber\\
\text{Maximize:}\quad & \sum_{\vec{v} \in C} Q(\vec{v}) \cdot p(\vec{v}) + \sum_i w_i \ell_i.\nonumber\\
\text{Such that:}\quad &x_j(\vec{v}) \leq 1-\ell_i,\ \forall \vec{v} \in C,\ i\sim j \in [n]\label{eq:one}\\
&\sum_i x_i(\vec{v}) \leq k,\ \forall \vec{v} \in C\label{eq:two}\\
&x_i(\vec{v}) \in [0,1],\ \forall \vec{v} \in C\label{eq:three}\\
&\sum_i v_i \cdot x_i(\vec{v}) - p(\vec{v}) \geq \sum_i v_i \cdot x_i(\vec{v}') - p(\vec{v}'),\ \forall \vec{v},\vec{v}'\in C\label{eq:four}\\
&x_i(\vec{v})\geq x_j(\vec{v}),\ \forall i\sim j, i \leq j.\label{eq:five}
\end{align}
Above, the variables $p(\vec{v})$ denote the price of the menu option purchased by $\vec{v}$ (or any valuation equivalent to $\vec{v}$ under $\Sigma$). Let $x_i(\vec{v})$ denote the probability with which $\vec{v}$ receives item $i$ (and also the probability with which $\sigma(\vec{v})$ receives item $\sigma(i)$). Let $\ell_i$ denote the leftovers for item $i$ in the resulting menu. To be explicit, the resulting menu is $\bigcup_{\vec{v} \in C}(G(\vec{x}(\vec{v})),p(\vec{v}),\Sigma)$, all permutations of the explicitly-found options over permutations in $\Sigma$.

The objective then computes exactly the revenue of the output mechanism, plus the weighted leftovers. 

Equation~\eqref{eq:one} guarantees that the leftovers are computed correctly --- note that our menu will include $\sigma(\vec{x}(\vec{v}))$ for all $\sigma \in \Sigma$, so we do need to have this constraint for all $i \sim j$ (and not just $i = j$). This will (correctly) imply that $\ell_i = \ell_j$ for all $i \sim j$. Equation~\eqref{eq:two} guarantees, without loss of generality, that no option contains more than $k$ options in expectation (and therefore there is a distribution matching the marginal probability which offers at most $k$ items with probability one). This means that the buyer's expected value for an option is simply $\sum_i v_i \cdot x_i(\vec{v}')$. Equation~\eqref{eq:three} guarantees that all probabilities are, in fact, probabilities. 

Equations~\eqref{eq:four} and~\eqref{eq:five} handle the incentive constraints. Equation~\eqref{eq:four} guarantees that every valuation vector $\vec{v} \in C$ prefers the menu option targeted at them versus one designed for any other $\vec{v}' \in C$. However, it does not guarantee that $\vec{v}$ prefers this option to $(\sigma(\vec{x}(\vec{v'})),p(\vec{v'}))$ for any non-trivial $\sigma$. Equation~\eqref{eq:five} guarantees that for all $\vec{v} \in C$, their favorite menu option of the form $(\sigma(\vec{x}(\vec{v'})),p(\vec{v'}))$ for any $\sigma \in \Sigma$ is when $\sigma$ is the identity permutation, and therefore both equations together guarantee that every valuation vector indeed prefers their intended option.

The feasibility constraints are exactly the space of all symmetric menus, and the objective function correctly calculates their revenue plus weighted leftovers. Therefore, the solution to this LP is the optimal menu. The menu is strongly symmetric with respect to $\Sigma$, and has strong symmetric menu complexity at most $|C|$, and the LP has size $\poly(|C|)$. Finally, we argue that there are at most $n^{cs}$ elements of $C$, which will conclude the proof.

If $n(i, x)$ denotes the number of items $j$ in $T_i$ for which $v_j= x$, then observe that fixing $n(i,x)$ for all $x$ uniquely determines an element of $C$ (because these values must be permuted into decreasing order within each $T_i$). There are only $c$ different values of $x$ in the support of $D_j$, and $s$ different values of $i$, and there are at most $n$ possible options for each $n(i,x)$. Therefore, the total number of such vectors is at most $n^{cs}$. 
\end{proof}

\subsection{Warmup: i.i.d. items}\label{sec:iid}
In this section, we consider a $k$-demand buyer with i.i.d. items. Qualitatively similar results for this setting in the bounded case are already known due to~\cite{DaskalakisW12}. While the main purpose of this section is as a warmup, our results here do make quantitative improvements, and now extend to the unbounded case due to Theorem~\ref{thm:reduction}. 

We begin by proposing a discretization of $D$ for which we can directly apply Theorem~\ref{thm:symmetries}.

\begin{defn}\label{def:canon} Let $D$ be $k$-demand independent items and $t$-bounded. The canonical value-discretization $D'_\delta$ with parameter $\delta$ is defined to first draw $\vec{v} \leftarrow D$, and then round all values down to the nearest number which can be written as $t\REV(D)$ times an integral power of $(1-\delta)$ (that is, round $v_i$ down to $(1-\delta)^{\lfloor \ln_{1-\delta}(v_i/(t \REV(D)))\rfloor}t\REV(D)$), and then round down any values less than $\delta \REV(D)/k$ to $0$. Observe that $D'_\delta$ is still $k$-demand over independent items.
\end{defn}

\begin{proposition}\label{prop:iidcouples} Let $D$ be $k$-demand over independent items and $t$-bounded, then 
\[\delta(D, D'_\delta) \leq \delta \VAL(D) + \delta \REV(D).\]
\end{proposition}
\begin{proof}
Consider an intermediate $D'$ which first draws $v$ from $D$ and then rounds down all values to the nearest power of $(1-\delta)$ (but doesn't further round especially small values to $0$). It's clear that $v(S) \geq v'(S)$ for all $S$. Therefore, the maximum that we can possibly lose while rounding down from $v$ to $v'$ is $\delta v([n])$, implying that $\delta(D, D') \leq \delta \VAL(D)$. 

Consider now further rounding down $v'$ to $v'_\delta$ by replacing all values $\leq \delta \REV(D)/k$ with $0$. Then because $v'$ is $k$-demand, we clearly lose at most $\delta \REV(D)$ in value for any set, and therefore $\delta(D', D'_\delta) \leq \delta \REV(D)$. Because $\delta(\cdot,\cdot)$ is a metric, we conclude that $\delta(D, D'_\delta) \leq \delta \VAL(D) + \delta \REV(D)$ by the triangle inequality.
\end{proof}

\begin{corollary}\label{cor:coupledist}
Let $D$ be $k$-demand over independent items and $t$-bounded for $t \geq 1$. Then $\delta(D, D'_\delta) = O(\delta t)\cdot \REV(D)$.
\end{corollary}

\begin{proof}
Follows directly from Proposition~\ref{prop:iidcouples} and Proposition~\ref{prop:RW} (which provides an upper bound on $\VAL(D)$ of $O(t)\REV(D)$ as $D$ is subadditive over independent items and $t$-bounded).
\end{proof}

\begin{proposition}\label{prop:iiddemand} Let $D$ be $k$-demand over i.i.d. items and $t$-bounded, and $\mathcal{W}$ be the subset of $\mathbb{R}_+^n$ with $w_i = w_j$ for all $i,j$. Then a $(1-\varepsilon)$-approximation to $\modrevmax_{\mathcal{D},\mathcal{W}}$ can be found in time $n^{O(t^2\ln(k)/\varepsilon^4)}$. Moreover, the mechanism output has strong symmetric menu complexity at most $n^{O(t^2\ln(k)/\varepsilon^{4})}$. 
\end{proposition}
\begin{proof}
We just apply Proposition~\ref{prop:advancedNudge} as follows. If we take $\delta = O(\varepsilon^{2}/t)$, then we'll have $\delta(D, D'_\delta) = O(\varepsilon^2)\REV(D)$. This means that $\REV(D'_\delta) \geq \REV(D)-O(\varepsilon)\cdot \REV(D)$ by Corollary~\ref{cor:expectedNudge}. So if we take $M$ to be the optimal mechanism for $D'_\delta$, which has strong symmetric menu complexity at most $n^{\ln_{1-\delta}(k/\delta)+1}=n^{O(\ln(k)/\delta^2)}$ by Theorem~\ref{thm:symmetries} (this is because there are only $\ln_{1-\delta}(k/\delta)+1$ integers between $0$ and $\ln_{1-\delta}(k/\delta)+1$ for values that $v_i \leftarrow (D'_\delta)_i$ can take), then Proposition~\ref{prop:advancedNudge} tells us that taking $M'$ to be $M$ but with prices multiplied by $(1-\varepsilon)$, then $M'$ loses only an additional $\delta(D, D'_\delta)/\varepsilon = O(\varepsilon)\REV(D)$ revenue.

Evaluating $\ln(k)/\delta^2 = O(\ln(k) t^2/\varepsilon^4)$ completes the proof.
\end{proof}

\begin{proof}[Proof of Theorem~\ref{thm:iiddemand}]
Theorem~\ref{thm:iiddemand} now follows immediately from Theorem~\ref{thm:reduction2} and Proposition~\ref{prop:iiddemand}.
\end{proof}

\subsection{Unit-demand over Asymmetric Independent Items}
In this section, we show that even for an arbitrary distribution which is unit-demand over independent items, which may initially have no symmetries whatsoever, we can still make use of Theorem~\ref{thm:symmetries} anyway. Observe that simply discretizing the values will no longer suffice, as there are asymmetric distributions even when each marginal has support two. So we will additionally have to modify the probabilities within each marginal. The high-level plan is as follows:

\begin{itemize}
\item First, apply the canonical value discretization. We already know from Corollary~\ref{cor:coupledist} how to bound the error from this. However, this doesn't immediately give us a symmetric distribution as even with two distinct values per marginal we could have $n$ distinct marginals.
\item Next, we have to discretize the probability vector for each marginal. The challenge here is that a simple nudge-and-round argument, even one that makes use of expectations, is doomed to fail if $v$ and its couple disagree by $>\varepsilon$ on \emph{any} item. Here is where we really get mileage out of Proposition~\ref{prop:advancedNudge} via a careful coupling of the original $D$ and discretized distribution.
\item Finally, we invoke Theorem~\ref{thm:symmetries}. The first bullet will result in at most {$-\ln(\delta\varepsilon)/\delta$} different values, and the second step will result in $\ln(n)/\delta$ different possible probabilities, for a total of at most $(\ln(n)/\delta)^{-\ln(\delta \varepsilon)/\delta}$ different possible probability vectors. After discretizing the $w_i$'s as well, this will leave us with at most $O((\ln(n)/\delta)^{-\ln(\delta\varepsilon)/\delta})$ different ``profiles'' for an item, and therefore the resulting distribution is symmetric under some $\Sigma$ with partition size $O((\ln(n)/\delta)^{-\ln(\delta\varepsilon)/\delta})$.
\end{itemize}

We now proceed to define our probability discretization.

\begin{defn}[Canonical probability-discretization $D^*_\delta$ and careful coupling]\label{def:unitcouple}
Let $D$ be $k$-demand over independent items and $t$-bounded. Define the canonical probability-discretization $D^*_\delta$ of $D$ with parameter $\delta$ to be the following distribution:
\begin{itemize}
\item First, draw $v'$ from $D'_{k\delta}$, as defined in the canonical value-discretization.
\item Define $q_i(x):=\Pr[v'_i = x]$, and $q'_i(x)$ to be $q_i(x)$ rounded down to the nearest power of $(1-\delta)$. If $q_i(x) < \delta^2/n^2$, instead set $q'_i(x) = 0$. Move the remaining probability mass to $0$.
\end{itemize}

Define the \emph{careful coupling} between $D$ and $D^*_\delta$ by ``converting" a draw from $D$ into a draw from $D^*_\delta$ in the following way:
\begin{itemize}
\item Draw $v \leftarrow D$, let $v'$ be its couple in $D'_{k\delta}$.
\item For all $i$ independently, do the following. Flip a coin which is heads with probability $q'_i(v'_i)/q_i(v'_i)$ (define $0/0$ to be $0$). If it lands heads, keep $v'_i = v_i$. Otherwise, update $v'_i= 0$.
\end{itemize}
\end{defn}

\begin{lemma}\label{lem:carefulcoupling} Let $D$ be $k$-demand over independent items and $t$-bounded for $t \geq 1$. Let also $M$ be any menu in which all randomized allocations with probability one award at most $k$ items. Then $\delta_M(D, D^*_\delta) \leq O(\delta t k)\cdot \max \{\SREV(D),\BREV(D)\}$.
\end{lemma}
The idea in the proof of Lemma~\ref{lem:carefulcoupling} is the following. Simple nudge-and-rounds argue that if $v$ and its couple $v'$ are close \emph{on every set}, then no matter what sets appear in the nudge-and-round calculations, the gap is small. Proposition~\ref{prop:advancedNudge} more carefully observes that it is wasteful to consider \emph{arbitrary} sets, when in fact only two such sets actually appear in the nudge-and-round analysis. And if the set which appears has only a single item (or is a distribution over sets of size one), then it only matters whether $v$ and $v'$ are close on this particular item, and not ones which never appear in the analysis. 

\begin{proof}[Proof of Lemma~\ref{lem:carefulcoupling}]
This proof will really make use of the carefulness in Proposition~\ref{prop:advancedNudge} as Corollary~\ref{cor:expectedNudge} does not suffice. We first observe that for all couples $(v,v')$ from $(D'_{k\delta},D^*_\delta)$, we have that $v(S) \geq v'(S)$ for all $S$. Therefore, in the definition of $\delta_M(D, D')$, the term:
$$\sup_{S_M(\cdot)}\Big\{\mathbb{E}_{(v,v')\leftarrow (D'_{k\delta}, D^*_\delta)}[v'(S_M(v')) - v(S_M(v'))] \Big\} = 0.$$
So let us now focus on the other term in the definition of $\delta_M(D, D')$. Here, observe that no matter how $S_M(\cdot)$ is defined, its output is some randomized allocation which awards at most $k$ items. So we can rewrite:
\begin{align*}
\mathbb{E}_{(v,v')\leftarrow (D'_{k\delta}, D^*_\delta)}[v(S_M(v)) - v'(S_M(v))] &= \mathbb{E}_{v \leftarrow D'_{k\delta}}\left[\mathbb{E}_{S \leftarrow S_M(v)}\left[\sum_{i \in S}\mathbb{E}_{v'\leftarrow D^*_\delta|v}\left[v_i - v'_i\right]\right]\right]\\
&\leq\mathbb{E}_{v \leftarrow D'_{k\delta}}\left[\mathbb{E}_{S \leftarrow S_M(v)}\left[\sum_{i \in S} t(1-q'_i(v_i)/q_i(v_i)))\REV(D)\right]\right]
\end{align*}

Above, the first line just rewrites the expectation to first draw $v$, which defines $S_M(v)$, and then draw $S$ from $S_M(v)$, and then we can draw $v'$ from $D^*_\delta$ conditioned on $v$. Crucially for this line, we needed $S_M(v)$ to be defined independently of $v'$. The second line observes that drawing $v'$ conditioned on $v$ maybe rounds each of at most $k$ values all the way down to zero, which hurts at most $t \REV(D)$, but occurs not with particularly high probability.

We now split up the above expectation into two cases: those $v$ for which $q'_i(v_i) = 0$ for some $i$, and the rest. We'll call those $v$ for which $q'_i(v_i) = 0$ for some $i$ \emph{rare}, and denote by $R$ the set of such $v$.
\begin{align*}
\mathbb{E}_{v \leftarrow D'_{k\delta}}&\left[\mathbb{E}_{S \leftarrow S_M(v)}\left[\sum_{i \in S} t(1-q'_i(v_i)/q_i(v_i))\REV(D)\right]\right]\\
&= \mathbb{E}_{v \leftarrow D'_{k\delta}\cdot \mathbb{I}(v \in R)}\left[\mathbb{E}_{S\leftarrow S_M(v)}\left[\sum_{i \in S} t(1-q'_i(v_i)/q_i(v_i))\REV(D)\right]\right]\\
&\quad + \mathbb{E}_{v \leftarrow D'_{k\delta}\cdot \mathbb{I}(v \notin R)}\left[\mathbb{E}_{S \leftarrow S_M(v)}\left[\sum_{i \in S}t(1-q'_i(v_i)/q_i(v_i))\REV(D)\right]\right]\\
&\leq \Pr[v \in R] \cdot tk\REV(D) ~+~ \mathbb{E}_{v \leftarrow D\cdot \mathbb{I}(v \notin R)}\left[\mathbb{E}_{S \leftarrow S_M(v)}\left[k\delta t\REV(D)\right]\right]\\
&\leq \delta tk\REV(D)/n ~+~ \delta kt\REV(D)\\
& = O(\delta kt) \cdot \REV(D),
\end{align*}

where the first inequality follows from noting that $\frac{q'_i(v_i)}{q_i(v_i)} \geq 1-\delta$. Since we have already established that $\delta_M(D, D'_{k\delta}) \leq \delta(D, D'_{k\delta}) = O(kt\delta)$ (Corollary~\ref{cor:coupledist}), and that $\delta(\cdot,\cdot)$ is a metric, we get that $\delta_M(D, D^*_\delta) = O(kt\delta)$.
\end{proof}

We now conclude the main proposition of this section:

\begin{proposition}\label{prop:unit} Let $D$ be $k$-demand over independent items and $t$-bounded for $t \geq 1$. Then a $(1-\varepsilon)$-approximation to $\modrevmax(D,\vec{w})$ can be found in time $n^{O(\ln(ntk/\varepsilon)^{O(kt\ln(kt/\varepsilon)/\varepsilon^2)})}$. Moreover, the mechanism output has strong symmetric menu complexity at most $n^{O(\ln(ntk/\varepsilon)^{O(kt\ln(kt/\varepsilon)/\varepsilon^2)})}$. 
\end{proposition}

\begin{proof}
We will again make use of Theorem~\ref{thm:symmetries} on $D^*_\delta$. First, observe that there are not too many possible marginals in $D^*_\delta$. In particular, there are only $\ln(t/\delta)/\delta$ distinct values, and each non-zero value is drawn with one of at most $O(\ln(n/\delta))$ probabilities, so there are at most $(\ln(n/\delta))^{\ln(t/\delta)/\delta}$ different marginals. This certainly means that $D^*_\delta$ is symmetric with respect to some $\Sigma$ of partition size $(\ln(n/\delta))^{\ln(t/\delta)/\delta}$, but this isn't quite enough yet because we need $\vec{w}$ to be invariant under $\Sigma$ as well.

So first, check if any $w_i \geq\REV(D)/\varepsilon$. If so, then just let $M$ be the empty menu. We claim this is at least a $(1-\varepsilon)$-approximation. The empty menu will get $\sum_i w_i \geq \REV(D)/\varepsilon$ since $M$ allocates no items. Also, the benchmark is at most $\REV(D) + \sum_i w_i$, meaning that ignoring the revenue from $D$ entirely costs us at most a $\varepsilon$ fraction of our benchmark. 

If not, then round all $w_i$ down to the nearest power of $(1-\delta)$, rounding all $w_i$ which are less than $\varepsilon \REV(D)/n$ down to $0$. Refer to the new rounded values as $\underline{w_i}$. Now, observe that there are at most $O(\ln(n/\varepsilon)/\delta)$ different values of $\underline{w_i}$. So we can say that item $i \sim j$ iff the marginal $D^*_\delta$ onto items $i$ and $j$ are the same and also $\underline{w_i} = \underline{w_j}$. Now, there are still at most $O(\ln(n/\delta))^{\ln(t/\delta)/\delta}$ equivalence classes (as long as $t = \Omega(1/\varepsilon)$ and $\delta = O(\varepsilon$), so we can solve $\modrevmax(D^*_\delta,\underline{\vec{w}})$ in time $n^{O(\ln(n/\delta)^{\ln(t/\delta)/\delta})}$, finding a menu with strong symmetric menu complexity at most $n^{O(\ln(n/\delta)^{\ln(t/\delta)/\delta})}$. 

Now we just need to calculate the error. Let $M$ denote the optimal solution to $\modrevmax(D,\vec{w})$, with quality $\sum_i w_i \ell_i(M) + \REV_M(D)$. Then by Proposition~\ref{prop:advancedNudge}, the mechanism $M'$ with all prices rounded down to the nearest multiple of $(1-\varepsilon)$ satisfies:
$$\sum_i \underline{w_i}\ell_i(M') + \REV_{M'}(D^*_\delta) \geq (1-\delta)\sum_i w_i \ell_i(M) + (1-\varepsilon) \REV_M(D) - O(\delta kt)\REV(D)/\varepsilon - \varepsilon\REV(D).$$

So as long as we have $\delta = O(\varepsilon^2/(tk))$, then the subtracted terms account for at most $\varepsilon \cdot \textsf{OPT}$, and the quality of the optimal solution to $\modrevmax(D^*_\delta,\underline{\vec{w}})$ is close to the desired optimum. Once we solve $\modrevmax(D^*_\delta,\underline{\vec{w}})$ and then multiply the prices by $(1-\varepsilon)$, we'll get a menu $M^*$ satisfying:
$$\sum_i \underline{w_i}\ell_i(M^*) + \REV_{M^*}(D) \geq (1-\varepsilon)\cdot \left(\sum_i \underline{w_i} \ell_i(M') + \REV_{M'}(D^*_\delta)\right) - O(\delta k t)\REV(D)/\varepsilon.$$

Combining both equations and setting $\delta = O(\varepsilon^2/(tk))$ proves the proposition.
\end{proof}

\begin{proof}[Proof of Theorem~\ref{thm:unit}]
Theorem~\ref{thm:unit} now follows immediately from Theorem~\ref{thm:reduction2} and Proposition~\ref{prop:unit} (and upperbounding $\ln(1/\varepsilon)$ with $1/\varepsilon$ for ease of presentation).
\end{proof}


\subsection{$k$-Demand Buyer with Constant Support Per Item}\label{sec:constant}

In this section we prove Theorem~\ref{thm:constant}, which considers a bidder who is $k$-demand over independent items, where each $D_i$ has support at most  $c$ (note that the supports may be disjoint). The proof follows mostly from Lemma~\ref{lem:carefulcoupling}, and we just need to again count the number of distinct marginals.

\begin{proposition}\label{prop:kdemconstant} Let $D$ be $k$-demand over independent items and $t$-bounded for $t \geq 1$, and have each $D_i$ have support at most $c$. Then a $(1-\varepsilon)$-approximation to $\modrevmax(D,\vec{w})$ can be found in time $n^{O(\ln(nt/\varepsilon))^c}$. Moreover, the mechanism output has strong symmetric menu complexity at most $n^{O(\ln(nt/\varepsilon))^c}$.
\end{proposition}
\begin{proof}
Note that most of this proof is repeating steps in the proof of Proposition~\ref{prop:unit}. The only difference is in counting the number of distinct marginals.
We will again make use of Theorem~\ref{thm:symmetries} on $D^*_\delta$. First, observe that there are not too many possible marginals in $D^*_\delta$. In particular, there are only $c$ distinct values, and each non-zero value is drawn with one of at most $O(\ln(n/\delta))$ probabilities, so there are at most $O(\ln(n/\delta))^{c}$ different marginals. This certainly means that $D^*_\delta$ is symmetric with respect to some $\Sigma$ of partition size $O(\ln(n/\delta))^{c}$, but this isn't quite enough yet because we need $\vec{w}$ to be invariant under $\Sigma$ as well.

So first, check if any $w_i \geq\REV(D)/\varepsilon$. If so, then just let $M$ be the empty menu. We claim this is at least a $(1-\varepsilon)$-approximation. The empty menu will get $\sum_i w_i \geq \REV(D)/\varepsilon$ since $M$ allocates no items. Also, the benchmark is at most $\REV(D) + \sum_i w_i$, meaning that ignoring the revenue from $D$ entirely costs us at most a $\varepsilon$ fraction of our benchmark. 

If not, then round all $w_i$ down to the nearest power of $(1-\delta)$, rounding all $w_i$ which are less than $\varepsilon \REV(D)/n$ down to $0$. Refer to the new rounded values as $\underline{w_i}$. Now, observe that there are at most $O(\ln(n/\varepsilon)/\delta)$ different values of $\underline{w_i}$. So we can say that item $i \sim j$ iff the marginal $D^*_\delta$ onto items $i$ and $j$ are the same and also $\underline{w_i} = \underline{w_j}$. Now, there are still at most $O(\ln(n/\delta))^{c}$ equivalence classes (as long as $c\geq 2$ and $n = \omega(\log(1/\varepsilon))$), so we can solve $\modrevmax(D^*_\delta,\underline{\vec{w}})$ in time $n^{O(\ln(n/\delta))^c}$, finding a menu with strong symmetric menu complexity at most $n^{O(\ln(n/\delta))^c}$. 

Now we just need to calculate the error. Let $M$ denote the optimal solution to $\modrevmax(D,\vec{w})$, with quality $\sum_i w_i \ell_i(M) + \REV_M(D)$. Then by Proposition~\ref{prop:advancedNudge}, the mechanism $M'$ with all prices rounded down to the nearest multiple of $(1-\varepsilon)$ satisfies:
$$\sum_i \underline{w_i}\ell_i(M') + \REV_{M'}(D^*_\delta) \geq (1-\delta)\sum_i w_i \ell_i(M) + (1-\varepsilon) \REV_M(D) - O(\delta kt)\REV(D)/\varepsilon - \varepsilon\REV(D).$$

So as long as we have $\delta = O(\varepsilon^2/(tk))$, then the subtracted terms account for at most $\varepsilon \cdot \textsf{OPT}$, and the quality of the optimal solution to $\modrevmax(D^*_\delta,\underline{\vec{w}})$ is close to the desired optimum. Once we solve $\modrevmax(D^*_\delta,\underline{\vec{w}})$ and then multiply the prices by $(1-\varepsilon)$, we'll get a menu $M^*$ satisfying:
$$\sum_i \underline{w_i}\ell_i(M^*) + \REV_{M^*}(D) \geq (1-\varepsilon)\cdot \left(\sum_i \underline{w_i} \ell_i(M') + \REV_{M'}(D^*_\delta)\right) - O(\delta k t)\REV(D)/\varepsilon.$$

Combining both equations and setting $\delta = O(\varepsilon^2/(tk))$ proves the proposition.
\end{proof}

\begin{proof}[Proof of Theorem~\ref{thm:constant}]
The proof of Theorem~\ref{thm:constant} now follows from Proposition~\ref{prop:kdemconstant} and Theorem~\ref{thm:reduction2}.
\end{proof}

\section{Outline and Proofs: Selling Separately with Low Symmetric Menu Complexity}\label{sec:srevproofs}
In this section we provide the proof of Theorem~\ref{thm:apxsrev}, which draws on ideas from~\cite{BabaioffGN17}. Essentially, we will partition the items up into exclusive buckets $B_0,\ldots, B_k$, and one joint bucket $B$, and allow the buyer to purchase at most one item from each exclusive bucket (deciding separately for each exclusive bucket) and add this to either the entire joint bucket or none of the joint bucket. $B_0$ will be a special exclusive bucket where the prices might differ for each item inside, but for all other buckets the prices will be identical. We first argue that a mechanism of this form has low strongly symmetric menu complexity.

\begin{defn}[Bucket Mechanism] Let $B_0,\ldots, B_k, B$ partition $[n]$. Let each item $j$ in $B_0$ have a price $p_j$, and each bucket $i > 0$ have a price $q_i$, and bucket $B$ have a price $q$. Then the bucket mechanism associated with $B_0,\ldots, B_k, B, \vec{p},\vec{q},q$ allows the buyer to pick any set $S$ and pay price $p(S)$ such that:
\begin{itemize}
\item For all $i$, we have $S \cap B_i \leq 1$.
\item Either $B \subseteq S$ or $S \cap B = \emptyset$.
\item $p(S) = \sum_{j \in B_0 \cap S} p_j + \sum_{i,\ |S \cap B_i| =1} p_i + q\cdot \mathbb{I}(S \subseteq B)$.
\end{itemize}
\end{defn}

\begin{lemma}\label{lem:buckets}
Every bucket mechanism $M$ with $k$ buckets has $\mc(M) \leq 2n^{k+1}$, and $\ssmc(M) \leq n\cdot 2^{k+1}$.
\end{lemma}
\begin{proof}
The menu complexity is straight-forward to compute: the total number of sets which are available for purchase can be counted by picking any of the $|B_i|$ items in $B_i$, separately for each $i$, and then deciding whether or not to add $B$. There are $n$ choices for each bucket, and two for whether or not to add $B$.

To compute the strong symmetric menu complexity, consider the item permutation group $\Sigma$ which contains all permutations which separately permute items in $B_i$, for all $i > 0$. Observe that for any $\sigma \in \Sigma$ (below, define $\sigma(S):= \{\sigma(i),i \in S\}$):
\begin{itemize}
\item If the option $(S,p)$ is on $M$, then so is $(\sigma(S),p)$. 
\item There are at most $2^{k+1}n$ equivalence classes under $\Sigma$. 
\end{itemize}
To see the final bullet, consider any two sets $S, T$ for which:
\begin{itemize}
\item $S \cap B_0 = T \cap B_0$,
\item $|S \cap B_i| = |T \cap B_i|$, for all $i > 0$,
\item $S \cap B = T \cap B$,
\end{itemize}
there exists a $\sigma \in \Sigma$ such that $\sigma(S) = T$. Observe that there are only $2^{k+1}n$ choices above to make \footnote{If all the items are in $B_0$ then the menu size will be $n$, otherwise if there is at least one item not in $B_0$ then there are at most $n$ choices to make regarding to which item (if any) to take from $B_0$.}. 
\end{proof}

\begin{proof}[Proof of Theorem~\ref{thm:apxsrev}]
Now, we define a bucket mechanism with $O(1/\varepsilon^3)$ buckets which has revenue at least $(1-\varepsilon)\srev(D)$. To begin, break the items into three categories. Below, let $q_i = \Pr[v_i \geq p_i]$, and $\srev(D):= \sum_i p_i q_i$:
\begin{itemize}
\item High: items for which $p_i \geq \srev(D)/\varepsilon^2$. We will set $B_0$ to be all high items.
\item Medium: items for which $p_i \in [\varepsilon^{3} \srev(D),\srev(D)/\varepsilon^2)$. All medium items will be in buckets $B_1,\ldots, B_k$.
\item Low: items for which $p_i \leq \varepsilon^{3} \srev(D)$. We will set $B$ to be all low items.
\end{itemize}

First, we observe that $\sum_{i \in B_0} q_i \leq \varepsilon^2$. Assume for contradiction otherwise. Then we could set price $\srev(D)/\varepsilon^2$ on each item separately and get revenue $>\srev(D)$, a contradiction. We therefore conclude that, conditioned on $v_i \geq p_i$, $v_j < p_j$ for all $j \in B_0 \setminus \{i\}$ with probability at least $1-\varepsilon^2$ (by union bound), and therefore the revenue generated from selling $B_0$ exclusively is at least $\sum_{i \in B_0} (1-\varepsilon^2) p_i q_i$.

Next, consider the random variable $X = \sum_{i \in B}p_i \cdot \mathbb{I}(v_i \geq p_i)$. Then $\mathbb{E}[X] = \sum_{i \in B} p_i q_i$, and $X$ is a sum of independent random variables, each supported on $[0,\varepsilon^3 \srev(D) ]$. If $\sum_{i \in B} p_i q_i \leq \varepsilon \srev(D)$, then it is safe to just ignore these items. Otherwise if $\sum_{i \in B} p_i q_i \geq \varepsilon\srev(D)$, we claim that pricing the grand bundle recovers almost all the revenue from these items. It follows from a basic Chernoff bound that the probability that $X$ exceeds $(1-\varepsilon/2)\sum_{i \in B} p_i q_i$ is at least $1-e^{-\varepsilon^2 \cdot (1/\varepsilon^3)/3} =1- o(\varepsilon)$. So either the revenue from selling $B$ is not negligible and selling all items together gets at least $(1-\varepsilon)\sum_{i \in B} p_i q_i$, or their contribution to the revenue is $< \varepsilon \srev(D)$ and the items can be discarded altogether.

Finally, consider the medium items. First, round all prices in this range down to the nearest power of $(1-\varepsilon)$. Observe that there are at most $\ln_{1-\varepsilon}(1/\varepsilon^5) = O(\ln(1/\varepsilon)/\varepsilon)$ such powers of $(1-\varepsilon)$ in this range. Next, we will place all items in a bucket so that either (a) the total sum of probabilities in the bucket is at most $\varepsilon$ or (b) there is at most one item in the bucket. Observe that once we have such a bucketing, it will again be the case that for all medium items $i$, the probability that the buyer is willing to purchase any other items in $i$'s bucket is at most $\varepsilon$ by union bound, and therefore the revenue of selling each of these buckets exclusively is at least $\sum_{i\notin B \cup B_0} (1-\varepsilon)p_iq_i$. 

Therefore, this scheme generates revenue at least $(1-\varepsilon)\sum_i p_i q_i = (1-\varepsilon) \srev(D)$. The only remaining detail is to count the number of buckets for the medium items. We'll first put all items with $q_i \geq \varepsilon/2$ into their own buckets. For the remaining items, greedily pack them into buckets with other items of the same price until the sum of probabilities within that bucket would exceed $\varepsilon$. Observe that now there are two kinds of buckets: \emph{heavy} buckets whose sum of probabilities exceeds $\varepsilon/2$ and (b) \emph{light} buckets whose sum of probabilities does not exceed $\varepsilon/2$. Observe that because each $p_i \geq \varepsilon^3\srev(D)$ that there can be at most $2/\varepsilon^3$ heavy buckets (otherwise we'd contradict the definition of $\srev(D)$). There is only at most one light bucket per level of the discretization as it is just collecting the leftovers from that level. Hence the number of light buckets is $O(\ln(1/\varepsilon)/\varepsilon)$ and therefore the total number of buckets is $O(1/\varepsilon^3)$.

So in conclusion, the above bucketing scheme has $O(1/\varepsilon^3)$ buckets, and has the property that the revenue is at least $(1-\varepsilon)\sum_i p_iq_i$, proving Theorem~\ref{thm:apxsrev}.
\end{proof}

\section{A Barrier Example for an Additive Buyer} 
\label{sec:approachLimits}
In this section we highlight an example of a $(1/\varepsilon)$-bounded distribution which is additive over independent items but serves as a barrier to proving good (symmetric) menu complexity bounds. In a formal sense, the known approaches for bounding the menu complexity of $(1-\varepsilon)$-approximately optimal mechanisms for a bounded distribution are:
\begin{itemize}
\item Argue that $\srev(D) \geq (1-\varepsilon)\VAL(D)$, perhaps because each $D_i$ is nearly a point-mass (recall that while selling separately does not itself have low menu complexity, this suffices by Theorem~\ref{thm:apxsrev}).
\item {Argue that $\BREV(D) \geq (1-\varepsilon) \VAL(D)$, perhaps because for all $i$, $v(\{i\})\leq \varepsilon^2\VAL(D)$ with probability one, and therefore $v([n])$ concentrates tightly around its expectation.}
\item (New, from Section~\ref{sec:symmetries}) Argue that $D$ is ``close'' to a highly symmetric distribution $D'$. Then use Theorem~\ref{thm:symmetries} to argue that $D'$ has a near-optimal mechanism of low symmetric menu complexity, followed by Proposition~\ref{prop:advancedNudge} to argue that this menu (with discounts) also suffices for $D$.
\end{itemize}

We provide an example for which all three of these approaches fail, highlighting the main challenge for future work. We overview the construction in Example~\ref{ex:badexample} below, and highlight its main features in Proposition~\ref{prop:features}, deferring a proof of Proposition~\ref{prop:features} to Appendix~\ref{app:approachLimits}.

\begin{example}\label{ex:badexample} For even $n$, an $\varepsilon < 1$ (one interesting choice discussed below is $\varepsilon = 1/9$), and $k = \Theta(\frac{\ln(n)}{\varepsilon})$, consider $n$ vectors in $\vec{r}_1,\ldots, \vec{r}_n \in \{0,1\}^k$ with $|\vec{r}_i|_1 = k/2$ for all $i$ (i.e., each has exactly $k/2$ ones and $k/2$ zeroes). Let it also be the case that for all $i,j\in [n]$, we have $|\{\ell ~:~ r_{i\ell} \neq r_{j\ell}\}| \geq k/6$.\footnote{We will later prove that such vectors exist, which follows by the probabilistic method.} Define $D_i$ so that:
\[  \Pr[v_i = x] = \left\{
\begin{array}{ll}
	\frac{\varepsilon}{n}\cdot \mathbb{I}(i \text{ is even}) & x = 1\\
	\frac{2 \varepsilon}{n}\cdot \mathbb{I}(i \text{ is odd}) & x = \frac{1}{2} \\
      \frac{(\ln n) (1-\varepsilon)^{-\ell}}{nk}\cdot \mathbb{I}(r_{i\ell}=1)  & x = \frac{\varepsilon(1-\varepsilon)^\ell}{\ln n}, \ell = 0,..., k-1\\
      1- \Pr[v_i > x] & x = 0 \\
\end{array} 
\right. \]
\end{example}

\begin{proposition}\label{prop:features} Any distribution $D$ satisfying the definition in Example~\ref{ex:badexample} has the following properties:
\begin{itemize}
\item $\VAL(D) =3 \varepsilon/2$.
\item $\srev(D) =\varepsilon$.
\item $\BREV(D) \leq \varepsilon$.
\item $v_i \leq 1$ for all $i$ with probability $1$. Therefore, $D$ is $1/\varepsilon$-bounded.
\item For all $\Sigma$ with partition size $s$, and all $D'$ such that $D'$ is symmetric with respect to $\Sigma$, the coupling distance $\delta(D, D') \geq \varepsilon^2 \cdot \frac{n-s}{6n}$. In particular, if $s = o(n)$, then $\delta(D, D') \geq \varepsilon^2/6 - o(\varepsilon^2)$.
\end{itemize}
\end{proposition}

Observe that Proposition~\ref{prop:features} rules out any of the known approaches achieving better than a $8/9$-approximation. { Indeed, in order to prove that either of $\srev(D)$ or $\BREV(D)$ beats a $2/3$-approximation, a better bound than $\VAL(D)$ on the optimal revenue would be necessary. This may indeed be the right approach, but because $D$ is $1/\varepsilon$-bounded, it is already well inside the range where techniques like those of our Section~\ref{sec:reduction} or~\cite{BabaioffGN17} can yield traction. This rules out the first two approaches. In addition, the partition size of $\Sigma$ appears in the \emph{exponent} of the symmetric menu complexity for optimal mechanisms on distributions that are symmetric with respect to $\Sigma$, so the final bullet asserts that a direct application of the approach in Section~\ref{sec:symmetries} cannot beat a $(1-\sqrt{\varepsilon}/8)$-approximation with subexponential symmetric menu complexity}.\footnote{This is not trivial, but not hard, to deduce. See Lemma~\ref{lem:nosymm} in Appendix~\ref{app:approachLimits}.} In particular, the construction is valid for any $\varepsilon <1$, so we can take $\varepsilon$ small enough to have $\sqrt{\varepsilon}/8 < 1/9$, which would rule out an $8/9$-approximation via any of the three known approaches.

Focusing on arguments for this example (and slight generalizations) should be illuminating for future progress. It seems that the missing ingredient is a near-optimal bound on the optimal revenue without relying on coupling with a symmetric distribution. The interesting feature of $D$ is that it is highly asymmetric, but the values in the support of the marginals which contribute to the asymmetry are small (much smaller than $\VAL(D)$). Normally, this would imply that the expected value for the grand bundle concentrates, but the point masses at $1$ (or $1/2$) ruin such a concentration.

\medskip
\noindent


\bibliographystyle{alpha}
\bibliography{MasterBib}

\appendix

\section{Omitted Proofs}
\subsection{Proof of Lemma~\ref{lem:triangle}}\label{app:triangle}
\begin{proof}[Proof of Lemma~\ref{lem:triangle}]
First, it is clear that $\delta(D,D') \geq \delta_M(D,D')\geq 0$ for all $D, D'$. It is also clear that $\delta(D,D) = \delta_M(D, D) = 0$ for all $D$ via the identity coupling. And it is also clear that $\delta(D, D') = \delta(D',D)$ as the text in both definitions is identical after swapping the role of $D, D'$. 

To confirm that both satisfy the triangle inequality, consider three distributions $D, D', D''$, and consider any coupling which draws $(v, v')$ from $(D, D')$ and another which draws $(v', v'')$ from $(D', D'')$. Consider now the coupled draws $(v,v',v'')$ from $(D, D', D'')$ defined in the natural way (couple draws $(v,v')$ and $(v', v'')$ so that the $v'$ component is equal, and then concatenate). 

Observe first that:
\begin{align*}
\max_{S \subseteq [n]}\{v(S) - v'(S)\} + \max_{S \subseteq [n]}\{v'(S) - v''(S)\} &\geq \max_{S \subseteq [n]} \{v(S) - v''(S)\}.\\
\max_{S \subseteq [n]}\{v'(S) - v(S)\} + \max_{S \subseteq [n]}\{v''(S) - v'(S)\} &\geq \max_{S \subseteq [n]} \{v''(S) - v(S)\}.\\
\Rightarrow\max_{S \subseteq [n]}\{v(S) - v'(S)\} + \max_{S \subseteq [n]}\{v'(S) - v''(S)\} +&\max_{S \subseteq [n]}\{v'(S) - v(S)\} + \max_{S \subseteq [n]}\{v''(S) - v'(S)\}\\
& \geq \max_{S \subseteq [n]} \{v(S) - v''(S)\}+\max_{S \subseteq [n]} \{v''(S) - v(S)\}.
\end{align*}

This therefore implies that:

\begin{align*}
\mathbb{E}_{(v, v') \leftarrow (D, D')} \big[\max_{S \subseteq [n]} \{v(S) - v'(S)\} + \max_{S \subseteq [n]}&\{v'(S) - v(S)\} \big] \\
+ \mathbb{E}_{(v', v'') \leftarrow (D', D'')} \big[\max_{S \subseteq [n]} \{v'(S) - v''(S)\} &+ \max_{S \subseteq [n]}\{v''(S) - v'(S)\} \big] \\
&\geq \mathbb{E}_{(v, v'') \leftarrow (D, D'')} \big[\max_{S \subseteq [n]} \{v(S) - v''(S)\} + \max_{S \subseteq [n]}\{v''(S) - v(S)\} \big].
\end{align*}
 Therefore, we have established that for any pair of couplings between $D, D'$ and $D', D''$, we can find a coupling for $D, D''$ such that he expected error is smaller. This also implies that the infimum over couplings between $D, D''$ is certainly smaller than the sum of infimums over couplings between $D, D'$ and $D', D''$, establishing that $\delta(D, D') + \delta(D', D'') \geq \delta(D, D'')$, and the triangle inequality.

A nearly-identical proof establishes the same for $\delta_M$. Consider the same coupling among $(D, D', D'')$, and consider any $S_M, S'_M$. Then:
\begin{align*}
v(S_M(v)) - v''(S_M(v)) &= v(S_M(v)) - v'(S_M(v)) + v'(S_M(v)) - v''(S_M(v)).\\
v''(S'_M(v'')) - v(S'_M(v'')) &= v''(S'_M(v'')) - v'(S'_M(v'')) + v'(S'_M(v'')) - v(S'_M(v'')).\\
\end{align*}
\begin{align*}
\Rightarrow &\mathbb{E}_{(v,v'')\leftarrow (D, D'')}[v(S_M(v)) - v''(S_M(v)) + v''(S'_M(v'')) - v(S'_M(v''))]\\
&= \mathbb{E}_{(v,v',v'') \leftarrow (D, D')}[v(S_M(v)) - v'(S_M(v)) + v'(S'_M(v'')) - v(S'_M(v''))] \\
&\quad \quad+ \mathbb{E}_{(v,v',v'') \leftarrow (D', D'')}[v'(S_M(v)) - v''(S_M(v)) + v''(S'_M(v'')) - v'(S'_M(v''))].\\
&\leq  \sup_{S_M(\cdot)}\{\mathbb{E}_{(v, v') \leftarrow (D, D')}[v(S_M(v)) - v'(S_M(v))]\} + \sup_{S_M(\cdot)}\{\mathbb{E}_{(v, v') \leftarrow (D, D')}[v'(S_M(v')) - v(S_M(v'))]\}\\
&\quad \quad +\sup_{S_M(\cdot)}\{\mathbb{E}_{(v', v'') \leftarrow (D', D'')}[v'(S_M(v')) - v''(S_M(v'))]\} + \sup_{S_M(\cdot)}\{\mathbb{E}_{(v', v'') \leftarrow (D', D'')}[v''(S_M(v'')) - v'(S_M(v''))]\}\\
&= \delta_M(D, D') + \delta_M(D', D'').
\end{align*}

Above, the first implication follows immediately from the top two equalities. The inequality follows by observing that one candidate for $S_M(\cdot)$ is to take as input $v'$, and output $S'_M(v'')$ (where $v''$ is drawn from $D''$ conditioned on being coupled with $v'$), or to take as input $v'$ and output $S_M(v)$ (where $v$ is drawn from $D$ conditioned on being coupled with $v$). 

Observe now that we have established that for \emph{any} pairs of mappings $S_M(\cdot)$, $S'_M(\cdot)$, that we must have $\mathbb{E}_{(v, v'') \leftarrow (D, D'')}[v(S_M(v)) - v''(S_M(v))] + \mathbb{E}_{(v, v'') \leftarrow (D, D'')}[v''(S_M(v'')) - v(S_M(v''))]$, and therefore this inequality holds for the supremum over pairs of mappings as well, establishing that $\delta_M(D, D'') \leq \delta_M(D,D') + \delta_M(D', D'')$.
\end{proof}

\subsection{Proof of Theorem~\ref{thm:subadd}}\label{app:subadd}

\begin{proof}[Proof of Theorem~\ref{thm:subadd}] 

By Theorem~\ref{thm:reduction}, it is without loss of generality that, we only need to consider bounded distributions $D$ such that $v\{i\} \leq \REV(D)/\varepsilon^4$ for all $i$. We will propose a coupled distribution $D'$ (that does not need to be subadditive over independent items) and show that the distance between these distributions is small. This, together with the fact that we have reduced the support to a bounded (and discrete) domain will easily imply a finite bound on the menu complexity. 

Consider the following distribution $D'$. First, draw a valuation $v \sim D$. If $v < \varepsilon^2 \REV(D)$, let $v' = 0$. Otherwise round $v'$ to the nearest multiple of $\varepsilon^2 \REV(D)$. Note that $D'$ is, by construction, stochastically dominated by $D$ so term the $\max_S \{v'(S)-v(S)\}$ in $\delta(D, D')$ is $0$ and we only need to bound $v(S) -v'(S)$. Note that for all $S \subseteq [n]$, $v(S) - v'(S) \leq \varepsilon^2 \REV(D)$. Therefore by Corollary~\ref{cor:expectedNudge} we get that $\REV(D') \geq \REV(D) - \delta(D,D') \geq (1-\varepsilon) \REV(D)$. 

Let $M$ be the optimal mechanism for $D'$, and let $M'$ be the menu with all the prices discounted by $(1-\varepsilon)$. By Proposition~\ref{prop:advancedNudge} we get that $\REV_M'(D)\geq (1-\varepsilon)\REV_M(D')$, and together with the above conclusion we get $\REV_{M'}(D)\geq(1-\varepsilon)\REV(D)$. Note that $M$ and $M'$ have the same menu complexity. 

We now show that $M$ has menu complexity at most $\left(\frac{n}{\epsilon^6}\right)^{2^n}$, by bounding the total number distinct valuations in the support of $D'$. In order to define a valuation $v'$ from the support of $D'$ we need to assign to each subset $S \subseteq 2^{[n]}$ one of $n/\varepsilon^6$ discretized values. This is because $v([n]) \leq n \REV(D)/\varepsilon^4$ (due to subadditivity) and we are discretizing by multiples of $\varepsilon^2 \REV(D)$. Therefore there are at most $n/\varepsilon^6$ possible values for a subset. This implies that the number of possible valuations is upper bounded by $\left(\frac{n}{\varepsilon^6}\right)^{2^n}$. Thus, the total number of distinct options a buyer from $D'$ buys is at most $\left(\frac{n}{\varepsilon^6}\right)^{2^n}$, which implies the menu complexity of $M$ (and $M'$) is at most $\left(\frac{n}{\varepsilon^6}\right)^{2^n}$.
\end{proof}
\subsection{Proofs from Section~\ref{sec:setup}}\label{app:setup}

\begin{proof}[Proof of Proposition~\ref{prop:carefulNudge}]
Observe that we have, for any $v$, and any other $(T, q) \in M$:
$$v(S) - p \geq v(T) - q.$$

For whatever option $(T', (1-\varepsilon)q)$ $v$ chooses to purchase in $M'$, we have:
$$v(T') - (1-\varepsilon)q \geq v(S') - (1-\varepsilon)p.$$

If we now let $T$ denote the original menu option from which $T'$ was derived, we know that $v(T) \geq v(T')$, we we can combine the inequalities to get:
$$q \geq p - \left(v(S) - v(S')\right).$$

Note that the payment in $M'$ is $(1-\varepsilon)q$. Taking an expectation over both sides, we get:
$$\REV_{M'}(D) = (1-\varepsilon)\mathbb{E}[q] \geq (1-\varepsilon)\mathbb{E}[p] - \mathbb{E}[v(S) - v(S')]/\varepsilon = (1-\varepsilon)\REV_M(D) - \mathbb{E}[v(S) - v(S')]/\varepsilon.$$
\end{proof}

\begin{proof}[Proof of Lemma~\ref{lem:modrev}]
Consider an instance of $\modrevmax_{\mathcal{D},\mathcal{W}}$ with input $D,\vec{w}$. Because $D$ is $c$-bounded, let's take $T:=\max\{1,c\REV(D)\}$, and we know that $v(\{i\}) \leq T$ with probability one when $v(\cdot) \leftarrow D$. Consider now the truncation $D(T,\vec{w})$, and any mechanism $M$ for $D$. Let $M'$ denote the same mechanism $M$ after adding the option to purchase any item $i$ exclusively for price $(n^2T^3)\ell_i(M)-nT$. 

\begin{itemize}
\item First, observe that only a buyer with value $n^2T^3$ for an item will purchase one of the new options. All other $v(\cdot)$ drawn from $D$ will make their same purchase from $M$.
\item So because $\sum_i w_i \leq T$, the total probability that $v(\{i\}) = n^2T^3$ for any $i$ is at most $\frac{1}{T^2n^2}$, and when this event doesn't occur, the valuation is just drawn from $D$. So $\REV_{M'}(D(T,\vec{w})\cdot \mathbb{I}(v(\{i\}) < n^2T^3,\ \forall i)) \geq (1-\frac{1}{T^2n^2})\REV_M(D)$.
\item Next, consider a buyer with value $v(\{i\}) = n^2T^3$. With probability at least $1-\frac{1}{T^2n^2}$ they have $v(\{j\}) \leq T$ for all $j \neq i$, and therefore value at most $(n-1)T+(1-\ell_i(M))n^2T^3$ for any option on menu $M$ (by definition of leftovers and subadditivity). Their utility for picking the exclusive option is exactly $n^2T^3 - (n^2T^3)\ell_i(M) + nT$, which is strictly larger. So this buyer will certainly purchase the exclusive option. Therefore: $\REV_{M'}(D(T,\vec{w}) \cdot \mathbb{I}(\exists\ i, v(\{i\}) = n^2T^3))\geq \sum_i (1-\frac{1}{n^2T^2}) \ell_i(M) \frac{w_i}{n^2T^3}(n^2T^3 - nT)\geq (1-\frac{2}{nT^2})\sum_i \ell_i(M) w_i$.
\end{itemize}

Combining the last two bullets together yields:
$$\REV_{M'}(D(T,\vec{w}))\geq \left (1-\frac{2}{n}\right) \left(\REV_M(D) + \sum_i w_i \ell_i(M) \right).$$

Similarly, consider any mechanism $M$ for $D(T,\vec{w})$, and consider $M'$ to be the same mechanism after keeping only options purchased when $v(\{i\}) \leq T$ for all $i$. Then clearly all buyers with $v(\{i\}) \leq T$ purchase the same option, and there are only more of these buyers in $D$ since we have removed those with $v(\{i\})=n^2T^3$. Now we wish to figure out how much revenue could have possibly been lost from the high buyers.

Observe that the maximum price possibly charged for an option on $M'$ is $nc\REV(D)$, and so there is some option priced $\leq nc\REV(D)$ which awards item $i$ with probability $\ell_i(M')$. Therefore, the maximum that a buyer would have possibly paid in $M$ when their value for item $i$ is $n^2T^3$ is $n^2T^3 \ell_i(M') + nc\REV(D)$. Therefore, even if buyers paid exactly this much for item $i$ whenever their value was large, the total revenue contributed would be at most:
$$\sum_i \ell_i(M') w_i + \frac{nc\REV(D)}{n^2T^3} \leq \sum_i \ell_i(M')w_i + \frac{\REV(D)}{n}.$$

So we have shown both that the optima for both problems are equivalent up to a multiplicative $(1\pm 1/n)$, and that in fact any (near)-optimal solution for $\revmax(D(T,\vec{w}))$ can be efficiently transformed into one for $\modrevmax(D,\vec{w})$. 

\end{proof}
\subsection{Results from Prior Work}
For completeness (and because sometimes notation changes over time and it can be challenging to find the exact result we cite), we repeat proofs of prior results below.
\subsubsection{Proof of Lemma~\ref{lem:RW15}}
For completeness, we repeat a proof of Lemma~\ref{lem:RW15}.

In the lemma below, recall that $D_S:= \times_{i \in S} D_i$ (that is, the distribution $D$ restricted to items in $S$. Lemma~\ref{lem:adaptedRW} is very similar to the ``approximate marginal mechanism'' in~\cite{RubinsteinW15} (which is exactly Lemma~\ref{lem:adaptedRW} applied to $Y = \text{support}(D)$). For completeness, we include proofs of some steps.

\begin{lemma}\label{lem:adaptedRW}
Let $D$ be subadditive over independent items, and let $S,\bar{S}$ partition $[n]$. Define $D^+$ to first sample $v(\cdot) \leftarrow D$, and then output the valuation function $v^+(\cdot)$ with $v^+(X):= v(S \cap X) + v(X \cap \bar{S})$. Then for any subdomain $Y \subseteq \text{support}(D)$, and $\varepsilon > 0$ 
$$\REV(D \cdot \mathbb{I}(v \in Y)) \leq \frac{\REV(D^+ \cdot \mathbb{I}(v \in Y))}{1-\varepsilon}+\frac{\VAL(D_S \cdot \mathbb{I}(v\in Y))}{\varepsilon}. $$
\end{lemma}

\begin{proof} First, observe that because $v$ is subadditive that $v^+(X) \geq v(X)$ for all $X$, and therefore $D^+$ stochastically dominates $D$. Moreover, observe that for any $X$ that:
\begin{align*}
v(X) + v(X\cap S) \geq v(X\cap \bar{S}) + v(X \cap S)=v^+(X) \geq v(X)
\end{align*}
The first inequality follows from monotonicity, and the second follows from subadditivity. Therefore, $v^+(X) - v(X) \leq v(S)$ for all $X$. This immediately implies that for any randomized option $X$ that $v^+(X)-v(X) \leq v(S)\cdot \mathbb{I}(v \in Y)$, which in turn implies that for any menu $M$: $\delta_M(D, D') \leq \mathbb{E}[v(S) \cdot \mathbb{I}(v \in Y)] = \VAL(D_S \cdot \mathbb{I}(v \in Y)]$. 

We can now apply Proposition~\ref{prop:advancedNudge} directly, using $M$ as the optimal mechanism for $D\cdot \mathbb{I}(v \in Y)$. 
\end{proof}

The following corollary is proved in~\cite{RubinsteinW15} and follows from Lemma~\ref{lem:adaptedRW} very similarly to a comparable bound in~\cite{HartN17} for additive buyers from their ``marginal mechanism'' lemma (we do not include a proof of the Marginal Mechanism lemma, and refer the reader to~\cite{CaiH13, HartN17, RubinsteinW15}). 

\begin{lemma}[Marginal Mechanism on subdomain~\cite{CaiH13, HartN17}]\label{lem:marginal} For any $S \subseteq [n]$, and $D$ such that with probability one, when $v(\cdot) \leftarrow D$ we have $v(X) = v(X \cap S) + v(\bar{X} \cap S)$ (that is, $v(\cdot)$ is additive across $S$), and $Y \subseteq \text{support}(D)$: $\REV(D) \leq \REV(D_S) + \VAL(D_{\bar{S}})$. 
\end{lemma}

\begin{proof}[Proof of Lemma~\ref{lem:RW15}]
Define $Y$ to be the subdomain where $v(S) \geq v(\bar{S})$, and $\bar{Y} := \text{support}(D)\setminus Y$. Then:
\begin{align*}
\REV(D) & \leq \REV(D  \cdot \mathbb{I}(v \in Y)) +\REV(D  \cdot \mathbb{I}(v \in \bar{Y})) \\ 
& \leq  \frac{3}{2}\left(\REV(D^+ \cdot \mathbb{I}(v \in Y))+\REV(D^+ \cdot \mathbb{I}(v \in \bar{Y}))\right)+3\left(\VAL(D_S \cdot \mathbb{I}(v\in \bar{Y}))+\VAL(D_{\bar{S}} \cdot \mathbb{I}(v\in Y))\right) \\ 
&\leq  \frac{3}{2}\left(\REV(D^+ \cdot \mathbb{I}(v \in Y))+\REV(D^+ \cdot \mathbb{I}(v \in \bar{Y})\right)+3\left(\REV(D_S) + \REV(D_{\bar{S}})\right) \\ 
& \leq \frac{3}{2}\left(\REV(D_S\cdot \mathbb{I}(v \in Y))+\VAL(D_{\bar{S}} \cdot \mathbb{I}(v \in Y))+\REV(D_{\bar{S}}\cdot \mathbb{I}(v \in \bar{Y}))+\VAL(D_S \cdot \mathbb{I}(v \in \bar{Y}))\right)\\
&\quad+3\left(\REV(D_S)+\REV(D_{\bar{S}})\right) \\ 
& \leq \frac{3}{2}\left(2\REV(D_S)+2\REV(D_{\bar{S}})\right)+3\left(\REV(D_S)+\REV(D_{\bar{S}})\right) \\
& = 3 \left(\REV(D_S)+\REV(D_{\bar{S}})\right) + 3\left(\REV(D_S)+\REV(D_{\bar{S}})\right) = 6\left(\REV(D_S)+\REV(D_{\bar{S}})\right).
\end{align*}

The first inequality follows by Sub-Domain Stitching (Lemma~\ref{lem:sds}). The second inequality follows from Lemma~\ref{lem:adaptedRW} on each of the two terms with $\varepsilon=1/3$.

For the third inequality, observe that one way to sell items only in $S$ is to sample (independently of the buyers' values) a draw from $v(\bar{S})\leftarrow D_{\bar{S}}$, and set price $v(\bar{S})$ to receive all items in $S$. If we treat the variables defining $v(\bar{S})$, and separately those defining $v(S)$ as a full valuation function, then the buyer will purchase set $S$ \emph{exactly when $v \in Y$}, and the expected revenue of this mechanism is therefore $\VAL(D_{\bar{S}} \cdot \mathbb{I}(v \in Y))$. We therefore conclude that:
$$\REV(D_S) \geq \VAL(D_{\bar{S}} \cdot \mathbb{I}(v \in Y)).$$
$$\REV(D_{\bar{S}}) \geq \VAL(D_S \cdot \mathbb{I}(v \in \bar{Y}))).$$

The fourth inequality follows from a direct application of subdomain stitching, together with the observation that $D_S$ is exactly the distribution $D^+_S$ ($D^+$ restricted to items $S$), and that $D_{\bar{S}}$ is $D^+_{\bar{S}}$.

The final inequality observes that $\REV(D\cdot \mathbb{I}(E)) \leq \REV(D)$ for all $D$ and events $E$, and also uses the same two inequalities above upper bounding $\VAL(D_{\bar{S}}\cdot \mathbb{I}(v \in Y)$ and $\VAL(D_S \cdot \mathbb{I}(v \in \bar{Y}))$.
\end{proof}


\subsubsection{Proof of Lemma~\ref{lem:RZ}}

\begin{proof}[Proof of Lemma~\ref{lem:RZ}]

Let $(S,p)$ dominate $(T,q)$. We argue that no buyer will ever purchase $(T,q)$. Consider any buyer with value $v_i \geq q/\Pr[i \in T]$ who prefers $(T,q)$ to $(\emptyset,0)$. The difference in utility for a buyer with value $v_i$ for item $i$ between receiving $(S,p)$ and $(T,q)$ is:
$$v_i \cdot \Pr[i \in S] - p - v_i \cdot \Pr[i \in T] + q.$$

We claim this is always $> 0$. Indeed:
\begin{align*}
v_i \cdot \Pr[i \in S] - p - v_i \cdot \Pr[i \in T] + q = (v_i - p/\Pr[i \in S]) \cdot \Pr[i \in S] - (v_i - q/\Pr[i \in T])\cdot \Pr[i \in T] > 0.
\end{align*}
The last inequality follows as the definition of dominance guarantees that $v_i - p/\Pr[i \in S] \geq v_i - q/\Pr[i \in T]$, and $\Pr[i \in S] \geq \Pr[i \in T]$. Therefore, the first term is the product of two non-negative terms which are each larger than the corresponding non-negative terms in the second product, and the product is at least as large. Strict inequality follows as if both terms are equal, then $(S,p) = (T,q)$. Finally, observe that the term $(v_i - q/\Pr[i \in T])$ is indeed non-negative for any valuation $v$ that would prefer $(T,q)$ to $(\emptyset, 0)$. Therefore, any dominated option is never purchased, and we can safely remove it from the menu without affecting revenue at all.
\end{proof}


\section{Formal Representation of a Symmetric Menu for a $k$-Demand Buyer}

Here, we formally establish that a $k$-demand buyer can find their favorite option on menu $M$ using time $\poly(\wsmc(M))$. 

\begin{lemma}\label{lem:polytimemenu}
Let $v(\cdot)$ be $k$-demand over independent items, and let $M:=\cup_j M_j$ where each $M_j$ is symmetric with respect to an item permutation $\Sigma_j$ and $\sum_j \ssmc(M_j) = C$. Then $\arg\max_{(S,p) \in M} \{v(S) - p\}$ can be found in time $\poly(n,C)$.
\end{lemma}
\begin{proof}
First, observe that it is w.l.o.g. to consider only menus $M$ such that each $(S,p) \in M$ has $|S| \leq k$ with probability one. Now, the only property that matters to a buyer with valuation $v(\cdot)$ is $\Pr[i \in S]$ for all $i,S$, because the utility $v$ derives from $(S,p)$ is exactly $\sum_i v_i \Pr[i \in S] - p$. When presented with a menu $M = \cup_j M_j$, where each $\ssmc(M_j) = d_j$, we can process each $(S,p,\Sigma_j)$ one at a time. Because $\Sigma_j$ is an item permutation, we can look at how it partitions the items into $T_1,\ldots, T_\ell$. Within each $T_x$, there is an option to permute the items in $T_x$ so that the relative ordering of $v_i$ and $\Pr[i \in S]$ match within $T_x$. The permutation which causes this for all $x$ is clearly the buyer's favorite option among all those in $(S,p,\Sigma_j)$. The buyer's value for this option is simply a dot product. The permutation and dot product together can be computed in time $O(n)$. By definition, there are $\sum_j d_j = C$ representatives to check, so the total time taken is $\poly(n,C)$.
\end{proof}

\section{Details for the Barrier Example}\label{app:approachLimits}
Here we provide the missing details for our barrier example of Section~\ref{sec:approachLimits}.

\begin{lemma}\label{lem:nosymm}
Let $D$ satisfy the conclusions of Proposition~\ref{prop:features}. Then an application of Corollary~\ref{cor:expectedNudge} cannot guarantee that $\REV(D') \geq (1-\sqrt{\varepsilon}/8)\REV(D) $.
\end{lemma}
\begin{proof}
To see this, observe that an application of Corollary~\ref{cor:expectedNudge} requires a choice of parameter $\delta$. For this $\delta$, we lose $\delta \REV(D) + \delta(D, D')/\delta$. In order to possibly get subexponential symmetric menu complexity from Theorem~\ref{thm:symmetries}, $D'$ must be symmetric with respect to some $\Sigma$ of partition size $o(n)$, and therefore $\delta(D, D') \geq \varepsilon^2/6 - o(\varepsilon^2)$ (note that $n$ is the parameter inside the little-oh, so this term approaches $0$ independent of $\varepsilon$ as $n \rightarrow \infty$). Therefore, as $\REV(D) \geq \srev(D) = \varepsilon$, for any $\delta$, the lost revenue is at least $\delta \varepsilon + \varepsilon^2/(6\delta) - o(\varepsilon^2/\delta)$, which for any choice of $\delta$ is at least $\varepsilon^{3/2}/3-o(\varepsilon^{3/2})$. 

As $\REV(D) \leq \VAL(D)$, this means that the lost revenue in the bound is at least a $\frac{\varepsilon^{3/2}/3 - o(\varepsilon^{3/2})}{3\varepsilon/2} = 2\sqrt{\varepsilon}/9 - o(\sqrt{\varepsilon})$ fraction, which exceeds a $\sqrt{\varepsilon}/8$ fraction as $n \rightarrow \infty$. 
\end{proof}

\begin{proof}[Proof of Proposition~\ref{prop:features}]
Let  $k = \frac{\ln(n /\ln(n))}{\varepsilon}+1$. We first argue that $n$ vectors $\vec{r}_1,\ldots, \vec{r}_n \in \{0,1\}^k$ exist, each with exactly $k/2$ ones, and such that for all $i,j$ there are at least $k/6$ indices where one $\vec{r}_i$ has a 1 but $\vec{r}_j$ does not.

\begin{claim} For $k =  \frac{\ln(n /\ln(n))}{\varepsilon}+1$, there exists $\vec{r}_1,\ldots, \vec{r}_n$ satisfying the hypotheses of Example~\ref{ex:badexample}.
\end{claim}
\begin{proof}
We use the probabilistic method, and let each $\vec{r}_i$ be a random element of $\{0,1\}^k$ conditioned on having exactly $k/2$ ones. Observe that for each index $\ell$, the probability that $r_{i\ell} = r_{j\ell}+1$ is exactly $1/4$. Also observe that the events for all $\ell$ are negatively correlated (more precisely, the probability that this fails for all $\ell \in S$ is at most $(3/4)^{|S|}$), and therefore Chernoff bounds with negative correlation (e.g.~\cite{ImpagliazzoK10}) imply that except with probability $e^{-\Omega(k)} = n^{-\Omega(1/\varepsilon)}$, there are at least $k/6$ such $\ell$. Taking a union bound over all $k^2$ pairs proves the claim.
\end{proof}

Now we want to prove the bullets.

\begin{claim}$\VAL(D) = 3\varepsilon/2$.
\end{claim}
\begin{proof}
This follows simply by observing that the large value in the support contributes exactly $\varepsilon/n$ to $\VAL(D_i)$, and there are $k/2$ other values in the support which each contribute exactly $\varepsilon/(nk)$ to $\VAL(D_i)$, for a total of $\VAL(D_i) = 3\varepsilon/(2n)$. 
\end{proof}

\begin{claim}$\srev(D) = \varepsilon$. 
\end{claim}  
\begin{proof}
To see that $\srev(D) \geq \varepsilon$, consider setting price $1$ on all even items and $1/2$ on all odd items.

To see the upper bound, consider any price $\frac{\varepsilon (1-\varepsilon)^j}{\ln(n)}$ for some $j \in \{0, 1, ... , k-1\}$. Then the item sells with probability at most $\frac{2\varepsilon}{n} + \sum_{h=0}^{j} \frac{\ln(n) (1-\varepsilon)^{-h}}{nk} \leq \frac{2\varepsilon}{n}+\frac{\ln(n)((1-\varepsilon)^{-(j+1)} -1)}{nk}\cdot\frac{1-\varepsilon}{\varepsilon}$. Therefore the revenue of setting price $\frac{\varepsilon (1-\varepsilon)^j}{\ln n}$ is at most $(1-\varepsilon)^j \frac{2\varepsilon^2}{n \ln n} + \frac{(1-(1-\varepsilon)^{j+1})}{nk} < \frac{\varepsilon}{n}$.
\end{proof}

\begin{claim}$\BREV(D) \leq \varepsilon-\varepsilon^2$. 

\end{claim}
\begin{proof}
Let $p$ be a price we consider for the grand bundle. Clearly if $p < \varepsilon - \varepsilon^2$ then the revenue will be at most $\varepsilon-\varepsilon^2$ since the bundle sells with probability at most $1$. Let $A$ be the event that none of the $x_i$ draws a value of $1$ or $1/2$. By the union bound $\Pr(A) \geq (1-2\varepsilon)$ and $\Pr(\lnot A) \leq 2\varepsilon$. Conditioned on $A$, the expected value of the bundle is $\mu \leq \frac{\varepsilon}{2(1-\frac{2\varepsilon}{n})}$. Let $\delta > 0$ be such that $p = \frac{1+\delta}{1-2\varepsilon/n} \frac{\varepsilon}{2}$. Note that conditioned on $A$ the range of $x_i$ is $[0, \frac{\varepsilon}{\ln n}]$. We blow up the values by $\frac{\ln n}{\varepsilon}$ in order to apply the following Chernoff bound for $[0, 1]$ continuous random variables: 
$$\Pr(\sum x_i > p \land A) \quad \leq \quad \Pr( \sum x_i > p | A) \Pr(A) \quad \leq \quad (1-\varepsilon) e^{\frac{-\delta^2 \ln n}{4(1-\varepsilon)(1+\delta/3)}}.$$

Note that if $p < \ln (n)$ then the second term in the upper bound above becomes $\text{poly}(\frac{1}{n})$ so the revenue is at most $(1-\varepsilon) \frac{\ln n}{n} < (1-\varepsilon) \varepsilon$. If $p > \ln n$ then since $\delta = \Omega(p)$ the probability of sell becomes $O(e^{-\delta})$. Therefore the revenue is $(1-\varepsilon) \delta O(e^{-\delta}) < (1-\varepsilon) \varepsilon$ for sufficiently large $n$. 

In the event that $A$ does not occur at least one of the values will be large. By the work above, conditioned on $A$, the contribution from the small items to the bundle value will be at most $2\varepsilon$ with high probability. 

So consider setting price $p \in (\varepsilon-\varepsilon^2,1/2+2\varepsilon)$. Then such a price will sell with low probability conditioned on $A$, or when event $A$ does not occur. But the probability of $\lnot A$ is at most $3\varepsilon/2$ by the union bound, and therefore generates revenue at most $3\varepsilon/4 + o(\varepsilon)$. 

Consider instead setting price $p \in (1/2+2\varepsilon,1+2\varepsilon)$. Then the grand bundle will sell with low probability conditioned on $A$, or when two of the odd items have value $1/2$, or when an even item has value $1$. By the union bound, this happens with probability at most $\varepsilon/2 + \varepsilon^2/n$, generating revenue at most $\varepsilon/2 + o(\varepsilon)$.

Finally, consider any $p \in (k+2\varepsilon,k+1+2\varepsilon)$. Then the grand bundle only sells when at least $k$ items take their high value. The probability that this occurs by the union bound is $O(\varepsilon^k)$, and therefore no large prices can yield good revenue either.
\end{proof}

The fourth bullet is trivial, but we do quickly confirm that each $D_i$ is indeed a valid distribution.

\begin{claim}For all $i$, we have $\Pr[v_i =0] \geq 0$.
\end{claim}
\begin{proof}
The mass at $x \neq 0$ is non-negative by construction, and by definition the probabilities sum to $1$, so we just need to make sure that the mass at $0$ is non-negative. Showing that $\frac{2\varepsilon}{n}+\left( \sum_{i=0}^{k-1} \frac{(\ln n) (1-\varepsilon)^{-i}}{nk} \right) < 1$ suffices since the density will add up to $1$ by construction. Some arithmetic shows that 

\begin{align*}
\frac{2\varepsilon}{n} + \sum_{i=0}^{k-1} \frac{(1-\varepsilon)^{-i} \ln n}{nk} & \leq \frac{2\varepsilon}{n} +  \frac{\ln n}{nk} \left(\frac{(1-\varepsilon)^{-k+1}-(1-\varepsilon)}{\varepsilon} \right) \\
 &\leq \frac{2\varepsilon}{n} + \frac{1}{\varepsilon k} < 1,
\end{align*} 
where the last line follows by noting that $\varepsilon$ is a constant and taking sufficiently large values of $n$.
\end{proof}
\begin{claim}
For any $i \neq j$, $\delta(D_i,D_j) \geq \varepsilon^2/(3n)$.
\end{claim}

\begin{proof}
We say $D_i, D_j$ \emph{disagree} on a coordinate $\ell$ if $\vec{r_i\ell} \neq \vec{r_j\ell}$. Suppose $D_i, D_j$ disagree on coordinate $\ell$ (and say $\vec{r_{i\ell}} = 1$). In the best case, either $\vec{r}_{j(\ell-1)} = 1$ or $\vec{r}_{j(\ell+1)} = 1$. In the former case, we couple as much mass as we can from $D_i$ at $\frac{\varepsilon(1-\varepsilon)^\ell }{\ln n}$ to $D_j$ at $\frac{\varepsilon(1-\varepsilon)^{\ell-1}}{\ln n}$. The contribution to the distance metric we care about from this displacement will be at least  

$$  \left( \frac{\ln n (1-\varepsilon)^{-\ell}-\ln n (1-\varepsilon)^{-(\ell-1)}}{nk}\right) \frac{ \varepsilon(1-\varepsilon)^\ell }{\ln n} + \frac{\ln n (1-\varepsilon)^{-(\ell-1)}}{nk} \left(\frac{\varepsilon ((1-\varepsilon)^{\ell-1}-(1-\varepsilon)^\ell) }{\ln n} \right) ~~\geq~~ \frac{2\varepsilon^2}{nk}. $$


The calculations for the latter case are nearly identical. Importantly, observe that each disagreement contributes error $2\varepsilon^2/(nk)$, and there are at least $k/6$ disagreements. So the total contribution is at least $\varepsilon^2/(3n)$.
\end{proof}

\begin{corollary} Let $D'$ be symmetric with respect to $\Sigma$, and let $\Sigma$ have partition size $s$. Then $\delta(D,D') \geq \varepsilon^2\frac{n-s}{6n}$.
\end{corollary}
\begin{proof}
Consider any $\Sigma$ which has partition size $s$. Then there are at most $s$ distinct marginals of $D'$. Note further that if there are $x$ items in a partition, the total coupling error contributed by that partition must be at least $(x-1)\varepsilon^2/(6n)$ (because if any item in that partition is within $\varepsilon^2/(6n)$ of the marginal, all of the rest must be strictly further by the above claim). Because there are at most $s$ distinct marginals in $D'$, we get that the sum over all partitions of their size minus one is at least $n-s$, proving the corollary.
\end{proof}
This completes the proof of Proposition~\ref{prop:features}.
\end{proof}
\IGNORE{
}

\end{document}